\PassOptionsToPackage{usenames,dvipsnames,svgnames,table}{xcolor}
\PassOptionsToPackage{pdftex=true,hidelinks}{hyperref}

\newif\iflongversion
\longversiontrue

\documentclass[a4paper,USenglish,thm-restate]{lipics-v2021}
\nolinenumbers

\pdfoutput=1
\usepackage{amssymb,thmtools}
\usepackage{format}
\usepackage{macros}
\usepackage{paper}

\iflongversion
\title{The Weakest Failure Detector for Genuine Atomic Multicast (Extended Version)}
\else
\title{The Weakest Failure Detector for Genuine Atomic Multicast}
\fi
\titlerunning{The Weakest Failure Detector for Genuine Atomic Multicast}

\author{Pierre Sutra}{Telecom SudParis, France }{pierre.sutra@telecom-sudparis.eu}{https://orcid.org/0000-0002-0573-2572}{}
\authorrunning{P. Sutra}
\Copyright{P. Sutra}

\EventEditors{Christian Scheideler}
\EventNoEds{1}
\EventLongTitle{36th International Symposium on Distributed Computing (DISC 2022)}
\EventShortTitle{DISC 2022}
\EventAcronym{DISC}
\EventYear{2022}
\EventDate{October 25--27, 2022}
\EventLocation{Augusta, Georgia, USA}
\EventLogo{}
\SeriesVolume{246}
\ArticleNo{5}

\iflongversion
\hideLIPIcs
\else
\fi

\begin{document}

\maketitle
\begin{abstract}
  Atomic broadcast is a group communication primitive to order messages across a set of distributed processes.
  Atomic multicast is its natural generalization where each message $m$ is addressed to $\dst(m)$, a subset of the processes called its destination group.
  A solution to atomic multicast is \emph{genuine} when a process takes steps only if a message is addressed to it.
  Genuine solutions are the ones used in practice because they have better performance.
  
  Let $\Gr$ be all the destination groups and $\Fa$ be the cyclic families in it, that is the subsets of $\Gr$ whose intersection graph is hamiltonian.
  This paper establishes that the weakest failure detector to solve genuine atomic multicast is
  $\WFD=(\land_{g,h \in \Gr}~\Sigma_{g \inter h}) \land (\land_{g \in \Gr}~\Omega_g) \land \gamma$,
  where
  \begin{inparaenumorig}[]
  \item $\Sigma_P$ and $\Omega_P$ are the quorum and leader failure detectors restricted to the processes in $P$, and
  \item $\gamma$ is a new failure detector that informs the processes in a cyclic family $\f \in \Fa$ when $\f$ is faulty.
  \end{inparaenumorig}
  
  We also study two classical variations of atomic multicast.
  The first variation requires that message delivery follows the real-time order.
  In this case, $\WFD$ must be strengthened with $1^{g \inter h}$, the indicator failure detector that informs each process in $g \union h$ when $g \inter h$ is faulty.
  The second variation requires a message to be delivered when the destination group runs in isolation.
  We prove that its weakest failure detector is at least $\WFD \land (\land_{g, h \in \Gr}~\Omega_{g \inter h})$.
  This value is attained when $\Fa=\emptySet$.
\end{abstract}

\section{Introduction}
\labsection{introduction}


\subparagraph*{Context}
Multicast is a fundamental group communication primitive used in modern computing infrastructures.
This primitive allows to disseminate a message to a subset of the processes in the system, its destination group.
Implementations exist over point-to-point protocols such as the Internet Protocol.
Multicast is atomic when it offers the properties of atomic broadcast to the multicast primitive:
each message is delivered at most once, and delivery occurs following some global order.
Atomic multicast is used to implement strongly consistent data storage \cite{isis,granola,pstore,janus}.

It is easy to see that atomic multicast can be implemented atop atomic broadcast.
Each message is sent through atomic broadcast and delivered where appropriate.
Such a naive approach is however used rarely in practice because it is inefficient when the number of destination groups is large \cite{multiringpaxos,versus}.
To rule out naive implementations, Guerraoui and Schiper~\cite{gamcast} introduce the notion of genuineness.
An implementation of atomic multicast is \emph{genuine} when a process takes steps only if a message is addressed to it.

Existing genuine atomic multicast algorithms that are fault-tolerant have strong synchrony assumptions on the underlying system.
Some protocols (such as \cite{SchiperP08}) assume that a perfect failure detector is available.
Alternatively, a common assumption is that the destination groups are decomposable into disjoint groups, each of these behaving as a logically correct entity.
Such an assumption is a consequence of the impossibility result established in \cite{gamcast}.
This result states that genuine atomic multicast requires some form of perfect failure detection in intersecting groups.
Consequently, almost all protocols published to date (e.g., \cite{ramcast,tempo,whitebox,fastcast,FritzkeIMR01,delporte}) assume the existence of such a decomposition.

\subparagraph*{Motivation}
A key observation is that the impossibility result in \cite{gamcast} is established when atomic multicast allows a message to be disseminated to \emph{any} subset of the processes.
However, where there is no such need, weaker synchrony assumptions may just work.
For instance, when each message is addressed to a single process, the problem is trivial and can be solved in an asynchronous system. 
Conversely, when every message is addressed to all the processes in the system, atomic multicast boils down to atomic broadcast, and thus ultimately to consensus.
Now, if no two groups intersect, solving consensus inside each group seems both necessary and sufficient.
In this paper, we further push this line of thought to characterize the necessary and sufficient synchrony assumptions to solve genuine atomic multicast.

Our results are established in the unreliable failure detectors model \cite{CT96,petr}.
A failure detector is an oracle available locally to each process that provides information regarding the speed at which the other processes are taking steps.
%
Finding the weakest failure detector to solve a given problem is a central question in distributed computing literature \cite{petr}.
In particular, the seminal work in \cite{omega} shows that a leader oracle ($\Omega$) is the weakest failure detector for consensus when a majority of processes is correct.
If any processes might fail, then a quorum failure detector ($\Sigma$) is required in addition to $\Omega$ \cite{sigma}.
%

A failure detector is realistic when it cannot guess the future.
In \cite{realistic}, the authors prove that the perfect failure detector ($P$) is the weakest realistic failure detector to solve consensus.
Building upon this result, Schiper and Pedone~\cite{SchiperP08} shows that $P$ is sufficient to implement genuine atomic multicast.
However, $P$ is the weakest only when messages are addressed to all the processes in the system. 
The present paper generalizes this result and the characterization given in \cite{gamcast} (see~\reftab{recap}).
It establishes the weakest failure detector to solve genuine atomic multicast for any set of destination groups.

\subparagraph*{Primer on the findings}
Let $\Gr$ be all the destination groups and $\Fa$ be the cyclic families in it, that is the subsets of $\Gr$ whose intersection graph is hamiltonian.
This paper shows that the weakest failure detector to solve genuine atomic multicast is
$\WFD=(\land_{g,h \in \Gr}~\Sigma_{g \inter h}) \land (\land_{g \in \Gr}~\Omega_g) \land \gamma$,
where
\begin{inparaenumorig}[]
\item $\Sigma_P$ and $\Omega_P$ are the quorum and leader failure detectors restricted to the processes in $P$, and
\item $\gamma$ is a new failure detector that informs the processes in a cyclic family $\f \in \Fa$ when $\f$ is faulty.
\end{inparaenumorig}%
Our results regarding $\gamma$ are established wrt. realistic failure detectors.

\begin{table}[t]
  \small
  \centering
  \begin{tabular}{c|c|c|c}
    \textbf{Genuiness} & \textbf{Order} & \textbf{Weakest} &  \\ 
    \hline
    \no     & Global       & $\Omega \land \Sigma$                                               & \cite{omega,sigma} \\
    \yes    & $\cdot$       & $\notin \mathcal{U}_2$                                                     & \cite{gamcast} \\
    $\cdot$& $\cdot$             & $\leq P$                                                              & \cite{SchiperP08} \\
    $\cdot$& $\cdot$             & $\WFD$                                                                & \refsection{necessity}, \refsection{sufficiency} \\
    $\cdot$& Strict       & $\WFD \land (\land_{g,h \in \Gr}~1^{g \inter h})$                            & \refsection{variations:rt} \\
    $\cdot$& Pairwise     & $(\land_{g,h \in \Gr}~\Sigma_{g \inter h}) \land (\land_{g \in \Gr}~\Omega_g)$ & \refsection{discussion} \\
    \strong & Global       & \textbf{if} $\Fa = \emptySet$ \textbf{then} $\WFD \land (\land_{g, h \in \Gr}~\Omega_{g \inter h})$         & \multirow{2}{*}{\refsection{variations:strong}} \\
            &              & \hspace{5em} \textbf{else} $\geq \WFD \land (\land_{g, h \in \Gr}~\Omega_{g \inter h})$    
  \end{tabular}
  \caption{
    \labtab{recap}
    About the weakest failure detector for atomic multicast.
    \emph{
      (
      \strong = strongly genuine
      )
    }
  }
  \vspace{-1em}
\end{table}


This paper also studies two classical variations of the atomic multicast problem.
The strict variation requires that message delivery follows the real-time order.
In this case, we prove that $\WFD$ must be strengthen with $1^{g \inter h}$, the indicator failure detector that informs each process in $g \union h$ when $g \inter h$ is faulty.
The strongly genuine variation requires a message to be delivered when its destination group runs in isolation.
In that case, the weakest failure detector is at least $\WFD \land (\land_{g, h \in \Gr}~\Omega_{g \inter h})$.
This value is attained when $\Fa=\emptySet$.

\subparagraph*{Outline of the paper}
\refsection{problem} introduces the atomic multicast problem and the notion of genuineness.
We present the candidate failure detector in \refsection{candidate}.
\refsection{sufficiency} proves that this candidate is sufficient.
Its necessity is established in \refsection{necessity}.
\refsection{variations} details the results regarding the two variations of the problem.
We cover related work and discuss our results in \refsection{discussion}.
\refsection{conclusion} closes this paper.
\iflongversion
For clarity, some of our proofs are deferred to the Appendix.
\else
Due to space constraints, all the proofs are deferred to the extended version \cite{longversion}.
\relatedversion{A full version of the paper is available at \cite{longversion}, \url{TBD}.}
\fi
 
\section{The Atomic Multicast Problem}
\labsection{problem}


\subsection{System Model}
\labsection{model}

In \cite{CT96}, the authors extend the usual model of asynchronous distributed computation to include failure detectors.
The present paper follows this model with the simplifications introduced in \cite{ever,commit}.
\iflongversion
This model is recalled in \refappendix{model}.
\else
This model is recalled in \cite{longversion}.
\fi

\subsection{Problem Definition}
\labsection{problem:def}

Atomic multicast is a group communication primitive that allows to disseminate messages between processes.
This primitive is used to build transactional systems \cite{granola,pstore} and partially-replicated (aka., sharded) data stores \cite{tempo,janus}.
In what follows, we consider the most standard definition for this problem \cite{isis,Hadzilacos94amodular,survey}.
In the parlance of Hadzilacos and Toueg~\cite{Hadzilacos94amodular}, it is named uniform global total order multicast.
Other variations are studied in \refsection{variations}.

Given a set of messages $\msgSet$, the interface of atomic multicast consists of operations $\multicast(m)$ and $\deliver(m)$, with $m \in \msgSet$.
Operation $\multicast(m)$ allows a process to \emph{multicast} a message $m$ to a set of processes denoted by $\dst(m)$.
This set is named the \emph{destination group} of $m$.
When a process executes $\deliver(m)$, it delivers message $m$, typically to an upper applicative layer.

Consider two messages $m$ and $m'$ and some process $p \in \dst(m) \inter \dst(m')$.
Relation $m \delOrderOf{p} m'$ captures the local delivery order at process $p$.
This relation holds when, at the time $p$ delivers $m$, $p$ has not delivered $m'$.
The union of the local delivery orders gives the \emph{delivery order}, that is $\delOrder = \union_{p \in \procSet} \delOrderOf{p}$.
The runs of atomic multicast must satisfy:%

\begin{itemize}
\item[(\emph{Integrity})]
  For every process $p$ and message $m$, $p$ delivers $m$ at most once, and only if $p$ belongs to $\dst(m)$ and $m$ was previously multicast.
\item[(\emph{Termination})]
  For every message $m$, if a correct process multicasts $m$, or a process delivers $m$, eventually every correct process in $\dst(m)$ delivers $m$.
\item[(\emph{Ordering})]
  The transitive closure of $\delOrder$ is a strict partial order over $\msgSet$.
\end{itemize}

Integrity and termination are two common properties in group communication literature.
They respectively ensure that only sound messages are delivered to the upper layer and that the communication primitive makes progress.
Ordering guarantees that the messages could have been received by a sequential process.
A common and equivalent rewriting of this property is as follows:

\begin{itemize}
\item[(\emph{Ordering})]
  Relation $\delOrder$ is acyclic over $\msgSet$.
\end{itemize}

If the sole destination group is $\procSet$, that is the set of all the processes, the definition above is the one of atomic broadcast.
In what follows, $\Gr \subseteq 2^{\procSet}$ is the set of all the destinations groups, i.e., $\Gr = \{g : \exists m \in \msgSet \ldotp g=\dst(m)\}$.
For some process $p$, $\Gr(p)$ denotes the destination groups in $\Gr$ that contain $p$.
Two groups $g$ and $h$ are \emph{intersecting} when $g \inter h \neq \emptySet$.


\subparagraph*{What can be sent and to who}
The process that executes $\multicast(m)$ is the sender of $m$, denoted $\src(m)$.
As usual, we consider that processes disseminate different messages (i.e., $\src$ is a function).
A message holds a bounded payload $\payload(m)$, and we assume that atomic multicast is not payload-sensitive. 
This means that for every message $m$, and for every possible payload $x$, there exists a message $m' \in \msgSet$ such that $\payload(m')=x$, $\dst(m')=\dst(m)$ and $\src(m')=\src(m)$.

\subparagraph*{Dissemination model}
In this paper, we consider a closed model of dissemination.
This means that to send a message to some group $g$, a process must belong to it (i.e., $\src(m) \in \dst(m)$).
In addition, we do not restrict the source of a message.
This translates into the fact that for every message $m$, for every process $p$ in $\dst(m)$, there exists a message $m'$ with $\dst(m)=\dst(m')$ and $\src(m')=p$.
Under the above set of assumptions, the atomic multicast problem is fully determined by the destination groups $\Gr$.

\subsection{Genuineness}
\labsection{problem:genuineness}


At first glance, atomic multicast boils down to the atomic broadcast problem:
to disseminate a message it suffices to broadcast it, and upon reception only messages addressed to the local machine are delivered.
With this approach, every process takes computational steps to deliver every message, including the ones it is not concerned with.
As a consequence, the protocol does not scale \cite{multiringpaxos,versus}, even if the workload is embarrassingly parallel (e.g., when the destinations groups are pairwise disjoint).

Such a strategy defeats the core purpose of atomic multicast and is thus not satisfying from a performance perspective.
To rule out this class of solutions, Guerraoui and Schiper~\cite{gamcast} introduce the notion of \emph{genuine} atomic multicast.
These protocols satisfy the minimality property defined below.

\begin{itemize}
\item[(\emph{Minimality})]
  In every run $\run$ of $\alg$, if some correct process $p$ sends or receives a (non-null) message in $\run$, there exists a message $m$ multicast in $\run$ with $p \in \dst(m)$.
\end{itemize}


%

All the results stated in this paper concern genuine atomic multicast.
To date, this is the most studied variation for this problem (see, e.g., \cite{ramcast,whitebox,fastcast}).

\section{The Candidate Failure Detector}
\labsection{candidate}

This paper characterizes the weakest failure detector to solve genuine atomic multicast.
Below, we introduce several notions related to failure detectors then present our candidate.

\subparagraph*{Family of destination groups}
A family of destination groups is a set of (non-repeated) destination groups $\f=(g_i)_i$.
For some family $\f$, $\cpaths(\f)$ are the closed paths in the intersection graph of $\f$ visiting all its destination groups.%
\footnote{
  The intersection graph of a family of sets $(S_i)_i$ is the undirected graph whose vertices are the sets $S_i$, and such that there is an edge linking $S_i$ and $S_j$ iff $S_i \inter S_j \neq \emptySet$.
}
Family $\f$ is \emph{cyclic} when its intersection graph is hamiltonian, that is when $\cpaths(\f)$ is non-empty.
A cyclic family $\f$ is \emph{faulty at time $t$} when every path $\pi \in \cpaths(\f)$ visits an edge $(g,h)$ with $g \inter h$ faulty at $t$.

In what follows, $\Fa$ denotes all the cyclic families in $2^{\Gr}$.
Given a destination group $g$, $\Fa(g)$ are the cyclic families in $\Fa$ that contain $g$.
For some process $p$, $\Fa(p)$ are the cyclic families $\f$ such that $p$ belongs to some group intersection in $\f$ (that is, $\exists g,h \in \f \sep p \in g \inter h$).

To illustrate the above notions, consider \reffigure{family}.
In this figure, $\procSet=\{p_1,\ldots,p_5\}$ and we have four destination groups:
$\gone=\{p_1,p_2\}$, $\gtwo=\{p_2,p_3\}$, $\gthree=\{p_1,p_3,p_4\}$ and $\gfour=\{p_1,p_4,p_5\}$.
The intersection graphs of $\f=\{\gone,\gtwo,\gthree\}$ and $\f'=\{\gone,\gthree,\gfour\}$ are depicted respectively in \reffigurestwo{family:ig:1}{family:ig:2}.
These two families are cyclic.
This is also the case of $\f''= \Gr = \{\gone,\gtwo,\gthree,\gfour\}$ whose intersection graph is the union of the two intersection graphs of $\f$ and $\f'$.
This family is faulty when $\gtwo \inter \gone = \{p_2\}$ fails.
Group $\gtwo$ belongs to two cyclic families, namely $\Fa(\gtwo)=\{\f,\f''\}$.
Process $p_1$ belongs to all cyclic families, that is $\Fa(p_1)=\Fa$.
Differently, since $p_5$ is part of no group intersection, $\Fa(p_5)=\emptySet$.

\tikzstyle{basic} = [draw, -latex',Black,-]
\tikzstyle{adot} = [inner sep=1.5pt, circle, fill]

\begin{figure}[t]
  \centering
  \captionsetup{justification=centering}
  \begin{subfigure}[t]{.3\textwidth}
    \centering
    \begin{tikzpicture}[
        every path/.style={black},
        scale=1.25,
        transform shape
      ]
      \node [] (p1) {$p_1$};
      \node [right of = p1] (p2) {$p_2$};
      \node [below right of = p2] (p3) {$p_3$};
      \node [below left of = p3] (p4) {$p_4$};
      \node [left of = p4] (p5) {$p_5$};
      
      \draw[Blue] \convexpath{p1,p2}{12pt};
      \draw[Crimson] \convexpath{p2,p3}{12pt};
      \draw[OliveGreen] \convexpath{p1,p3,p4}{12pt};
      \draw[Purple] \convexpath{p1,p4,p5}{12pt};
    \end{tikzpicture}
    \caption{\labfigure{family:groups}}
  \end{subfigure}
  \begin{subfigure}[t]{.3\textwidth}
    \centering  
    \begin{tikzpicture}
      [
        scale=1.25,
        auto,
        shifttl/.style={shift={(-\shiftpoints,\shiftpoints)}},
        shifttr/.style={shift={(\shiftpoints,\shiftpoints)}},
        shiftbl/.style={shift={(-\shiftpoints,-\shiftpoints)}},
        shiftbr/.style={shift={(\shiftpoints,-\shiftpoints)}},
      ]
      \begin{scope}[<-,thick]
        \node[adot,label=270:$\gone$] (g1) at (0,0) {};
        \node[adot,label=$\gtwo$] (g2) at (0,1) {};
        \node[adot,label=270:$\gthree$] (g3) at (2,0) {};
        \path[basic] (g1) edge node [midway,left] {$\{p_2\}$} (g2);
        \path[basic] (g1) edge node [midway,below] {$\{p_1\}$} (g3);
        \path[basic] (g2) edge node [midway,above] {$\{p_3\}$} (g3);
      \end{scope}
    \end{tikzpicture}
    \caption{\labfigure{family:ig:1}}
  \end{subfigure}
  \begin{subfigure}[t]{.3\textwidth}
    \centering  
    \begin{tikzpicture}
      [
        scale=1.25,
        auto,
        shifttl/.style={shift={(-\shiftpoints,\shiftpoints)}},
        shifttr/.style={shift={(\shiftpoints,\shiftpoints)}},
        shiftbl/.style={shift={(-\shiftpoints,-\shiftpoints)}},
        shiftbr/.style={shift={(\shiftpoints,-\shiftpoints)}},
      ]
      \begin{scope}[<-,thick]
        \node[adot,label=270:$\gone$] (g1) at (0,0) {};
        \node[adot,label=270:$\gthree$] (g3) at (2,0) {};
        \node[adot,label=$\gfour$] (g4) at (2,1) {};
        \path[basic] (g1) edge node [midway,below] {$\{p_1\}$} (g3);
        \path[basic] (g1) edge node [midway,above] {$\{p_1\}$} (g4);
        \path[basic] (g3) edge node [midway,right] {$\{p_1,p_4\}$} (g4);
      \end{scope}
    \end{tikzpicture}
    \caption{\labfigure{family:ig:2}}
  \end{subfigure}
  \caption{
    \labfigure{family}
    From left to right:
    the four groups $\gone$, $\gtwo$, $\gthree$ and $\gfour$, and the intersection graphs of the two cyclic families $\f=\{\gone,\gtwo,\gthree\}$ and $\f'=\{\gone,\gthree,\gfour\}$.
  }
\end{figure}

\subparagraph*{Failure detectors of interest}
Failure detectors are grouped into classes of equivalence that share common computational power.
Several classes of failure detectors have been proposed in the past.
This paper makes use of two common classes of failure detectors, $\Sigma$ and $\Omega$, respectively introduced in \cite{sigma} and \cite{omega}.
We also propose a new class $\gamma$ named the \emph{cyclicity failure detector}.
All these classes are detailed below.
\begin{itemize}
\item
  The quorum failure detector ($\Sigma$) captures the minimal amount of synchrony to implement an atomic register.
  When a process $p$ queries at time $t$ a detector of this class, it returns a non-empty subset of processes $\Sigma(p,t) \subseteq \procSet$ such that:
  \begin{itemize}
  \item[(\emph{Intersection})] $\forall p, q \in \procSet \sep \forall t, t' \in \naturalSet \sep \Sigma(p,t) \inter \Sigma(q,t') \neq \emptySet$
  \item[(\emph{Liveness})] $\forall p \in \correct \sep \exists \tau \in \naturalSet \sep \forall t \geq \tau \sep \Sigma(p,t) \subseteq \correct$
  \end{itemize}
  The first property states that the values of any two quorums taken at any times intersect.
  It is used to maintain the consistency of the atomic register.    
  The second property ensures that eventually only correct processes are returned.
\item 
  Failure detector $\Omega$ returns an eventually reliable leader \cite{eventualconsistency}.
  In detail, it returns a value $\Omega(p,t) \in \procSet$ satisfying that:
  \begin{itemize}
  \item[(\emph{Leadership})] $\correct \neq \emptySet \implies (\exists l \in \correct \sep \forall p \in \correct \sep \exists \tau \in \naturalSet \sep \forall t \geq \tau \sep \Omega(p,t) = l)$
  \end{itemize}
  $\Omega$ is the weakest failure detector to solve consensus when processes have access to a shared memory.
  For message-passing distributed systems, $\Omega \land \Sigma$ is the weakest failure detector.
\item
  The cyclicity failure detector ($\gamma$) informs each process of the cyclic families it is currently involved with.
  In detail, failure detector $\gamma$ returns at each process $p$ a set of cyclic families $\f \in \Fa(p)$ such that:
  \begin{itemize}
  \item[(\emph{Accuracy})] $\forall p \in \procSet \sep \forall t \in \naturalSet \sep (\f \in \Fa(p) \land \f \notin \gamma(p,t)) \implies \isFaulty{\f}{t}$
  \item[(\emph{Completeness})] $\forall p \in \correct \sep \forall t \in \naturalSet \sep (\f \in \Fa(p) \land \isFaulty{\f}{t}) \implies \exists \tau \in \naturalSet \sep \forall t' \geq \tau \sep \f \notin \gamma(p,t')$
  \end{itemize}
  Accuracy ensures that if some cyclic family $\f$ is not output at $p$ and $p$ belongs to it, then $\f$ is faulty at that time.
  Completeness requires that eventually $\gamma$ does not output forever a faulty family at the correct processes that are part of it.
  Hereafter, $\gamma(g)$ denotes the groups $h$ such that $g \inter h \neq \emptySet$ and $g$ and $h$ belong to a cyclic family output by $\gamma$.
\end{itemize}

To illustrate the above definitions, we may consider again the system depicted in \reffigure{family}.
Let us assume that $\correct = \{p_1, p_4, p_5\}$.
The quorum failure detector $\Sigma$ can return $\gone$ or $\gthree$, then $\gfour$ forever.
Failure detector $\Omega$ may output any process, then at some point in time, one of the correct processes (e.g., $p_1$) ought to be elected forever.
At processes $p_1$, $\gamma$ returns initially $\{\f,\f',\f''\}$.
Then, once families $\f$ and $\f''$ are faulty---this should happen as $p_2$ is faulty---the output eventually stabilizes to $\{\f'\}$.
When this happens, $\gamma(\gone)=\{\gthree, \gfour\}$.

\subparagraph*{Conjunction of failure detectors}
We write $C \land D$ the conjunction of the failure detectors $C$ and $D$ \cite{ever}.
For a failure pattern $F$, failure detector $C \land D$ returns a history in $D(F) \times C(F)$.

\subparagraph*{Set-restricted failure detectors}
For some failure detector $D$, $D_{P}$ is the failure detector obtained by restricting $D$ to the processes in $P \subseteq \procSet$.
This failure detector behaves as $D$ for the processes $p \in P$, and it returns $\bot$ at $p \notin P$.
In detail, let $F \inter P$ be the failure pattern $F$ obtained from $F$ by removing the processes outside $P$, i.e., $(F \inter P)(t)=F(t) \inter P$.
Then, $D_P(F)$ equals $D(F \inter P)$ at $p \in P$, and the mapping $p \times \naturalSet \rightarrow \bot$ elsewhere.
To illustrate this definition, $\Omega_{\{p\}}$ is the trivial failure detector that returns $p$ at process $p$.
Another example is given by $\Sigma_{\{p_1,p_2\}}$ which behaves as $\Sigma$ over $\procSet=\{p_1,p_2\}$.

\subparagraph*{The candidate}
Our candidate failure detector is $\WFD_{\Gr} = (\land_{g,h \in \Gr}~\Sigma_{g \inter h}) \land (\land_{g \in \Gr}~\Omega_{g}) \land \gamma$.
When the set of destinations groups $\Gr$ is clear from the context, we shall omit the subscript.

\section{Sufficiency}
\labsection{sufficiency}

This section shows that genuine atomic multicast is solvable with the candidate failure detector.
A first observation toward this result is that consensus is wait-free solvable in $g$ using $\Sigma_g \land \Omega_g$.
Indeed, $\Sigma_g$ permits to build shared atomic registers in $g$ \cite{sigma}.
From these registers, we may construct an obstruction-free consensus and boost it with $\Omega_g$ \cite{alpha}.
Thus, any linearizable wait-free shared objects is implementable in $g$ \cite{syn:1382}.
Leveraging these observations, this section depicts a solution built atop (high-level) shared objects.

Below, we first introduce group sequential atomic multicast (\refsection{sufficiency:gs}).
From a computability perspective, this simpler variation is equivalent to the common atomic multicast problem.
This is the variation that we shall implement hereafter.
We then explain at coarse grain how to solve genuine atomic multicast in a fault-tolerant manner using failure detector $\WFD$ (\refsection{sufficiency:overview}).
\iflongversion
Further, the details of our solution are presented (\refsection{sufficiency:algorithm}) and its correctness established (\refsection{sufficiency:correctness}).
\else
Further, the details of our solution are presented and its correctness informally argued (\refsection{sufficiency:algorithm}).
\fi

\subsection{A simpler variation}
\labsection{sufficiency:gs}

Group sequential atomic multicast requires that each group handles its messages sequentially.
In detail, given two messages $m$ and $m'$ addressed to the same group, we write $m \hb m'$ when $\src(m')$ delivers $m$ before it multicasts $m'$.
This variation requires that if $m$ and $m'$ are multicast to the same group, then $m \hb m'$, or the converse, holds.
\refprop{sufficiency:gs} below establishes that this variation is as difficult as (vanilla) atomic multicast.
Building upon this insight, this section depicts a solution to group sequential atomic multicast using failure detector \WFD.

\begin{proposition}
  \labprop{sufficiency:gs}
  Group sequential atomic multicast is equivalent to atomic multicast.
\end{proposition}

\iflongversion
\begin{proof}
  Clearly, a group sequential algorithm implements (vanilla) atomic multicast.
  For the converse direction, we proceed as detailed below.

  The reduction algorithm uses an instance $\A$ of a solution to the base problem.
  For each group $g$, the algorithm also maintains a list $L_g$, shared among all the members of $g$.
  This list is implemented using an instance of atomic multicast running only between the members of $g$.
  In addition, each process stores a mapping indicating whether a message $m$ is delivered locally or not.
  
  Initially, every message is not delivered.
  The first time $\A.\deliver(m,\any)$ triggers, a process $p$ marks $m$ as delivered and executes the upcall to deliver $m$.
  
  To multicast a message $m$, a process $p$ first adds $m$ to the list $L_g$.
  Then $p$ waits that every message prior to $m$ in $L_g$ is delivered, helping if necessary.
  To help delivering some message $m'$, $p$ executes $\A.\multicast(m',p)$.
  When all the messages before $m$ in $L_g$ are delivered locally, $p$ calls $\A.\multicast(m,p)$.
  Then process $p$ returns from its invocation to $\multicast(m)$ in the group sequential algorithm.
\end{proof}
\fi

\subsection{Overview of the solution}
\labsection{sufficiency:overview}

First of all, we observe that if the groups are pairwise disjoint, it suffices that each group orders the messages it received to solve atomic multicast.
To this end, we use a shared log $\LOG_g$ per group $g$.
Then, consider that the intersection graph of $\Gr$ is acyclic, i.e., $\Fa$ is empty, yet groups may intersect.
In that case, it suffices to add a deterministic merge procedure in each group intersection, for instance, using a set of logs $\LOG_{g \inter h}$ when $g \inter h \neq \emptySet$.

Now, to solve the general case, cycles in the order built with the shared logs must be taken into account.
To this end, we use a fault-tolerant variation of Skeen's solution \cite{Birman:1987,Guerraoui:97}:
in each log, the message is bumped to the highest initial position it occupies in all the logs.
In the original algorithm \cite{Birman:1987}, as in many other approaches (e.g., \cite{whitebox,fastcast}), such a procedure is failure-free, and processes simply agree on the final position (aka., timestamp) of the message in the logs.
In contrast, our algorithm allows a disagreement when the cyclic family becomes faulty.
This disagreement is however restricted to different logs, as in the acyclic case.

\subsection{Algorithm}
\labsection{sufficiency:algorithm}

\begin{algorithm}[!t]

  \small
  \caption{Solving atomic multicast with failure detector $\WFD$ -- code at process $p$}
  \labalg{sufficiency}

  \begin{algorithmic}[1]

    \begin{variables}
      \State $(\LOG_{h \inter h' })_{h, h' \in \Gr(p)}$ \labline{sufficiency:var:1}
      \State $(\CONS_{m,\f})_{m \in \msgSet, \f \subseteq \Gr}$ \labline{sufficiency:var:fix}
      \State $\phase[m] \assign \lambda m.\phStart$ \labline{sufficiency:var:2}
    \end{variables}

    \vspace{.5em}

    \begin{action}{$\multicast(m)$} \labline{sufficiency:amcast:1} \Comment{$g = dst(m) \land g \in \Gr(p)$}
      \Precondition $\phase[m] = \phStart$ \labline{sufficiency:amcast:2}
      \Effect $\LOG_{g}.\append(m)$ \labline{sufficiency:amcast:3}
    \end{action}

    \vspace{.5em}

    \begin{action}{$\pending(m)$} \labline{sufficiency:pending:1}
      \Precondition $\phase[m] = \phStart$ \labline{sufficiency:pending:2}
      \Precondition $m \in \LOG_g$ \labline{sufficiency:pending:3}
      \Precondition $\forall m' <_{\LOG_g} m  \sep \phase[m'] \geq \phCommit$ \labline{sufficiency:pending:3b} 
      \Effect \textbf{for all} $h \in \Gr(p)$ \textbf{do} \labline{sufficiency:pending:4}
      \Effect \hspace{1em} $i \assign \LOG_{g \inter h}.\append(m)$ \labline{sufficiency:pending:5}
      \Effect \hspace{1em} $\LOG_{g}.\append(m,h,i)$ \labline{sufficiency:pending:6}
      \Effect $\phase[m] \assign \phPending$ \labline{sufficiency:pending:7}
    \end{action}

    \vspace{.5em}

    \begin{action}{$\commit(m)$} \labline{sufficiency:commit:1}
      \Precondition $\phase[m] = \phPending$ \labline{sufficiency:commit:2}
      \Precondition $\forall h \in \gamma(g) \sep (m,h,\any) \in \LOG_{g}$ \labline{sufficiency:commit:3}
      \Effect $\Let~k = \max \{ i : \exists (m,\any,i) \in \LOG_{g} \}, $ \labline{sufficiency:commit:4}
      \Effect $\Let~\f = \{ h : \exists \f' \in \Fa(p) \sep g,h \in \f' \land g \inter h \neq \emptySet \}$ \labline{sufficiency:commit:fix:1} 
      \Effect $k \assign \CONS_{m,\f}.\propose(k)$ \labline{sufficiency:commit:fix:2}
      \Effect \textbf{for all} $h \in \Gr(p)$ \textbf{do} \labline{sufficiency:commit:5}
      \Effect \hspace{1em} $\LOG_{g \inter h}.\bumpAndLock(m,k)$ \labline{sufficiency:commit:6}
      \Effect $\phase[m] \assign \phCommit$ \labline{sufficiency:commit:7}
    \end{action}

    \vspace{.5em}
    
    \begin{action}{$\stabilize(m,h)$} \labline{sufficiency:stabilize:1}
      \Precondition $\phase[m] = \phCommit$ \labline{sufficiency:stabilize:2}
      \Precondition $h \in \Gr(p)$ \labline{sufficiency:stabilize:3}
      \Precondition $\forall m' <_{\LOG_{g \inter h}} m \sep \phase[m'] \geq \phStable$ \labline{sufficiency:stabilize:4}
      \Effect $\LOG_{g}.\append(m,h)$ \labline{sufficiency:stabilize:5}
    \end{action}

    \vspace{.5em}

    \begin{action}{$\stable(m)$} \labline{sufficiency:stable:1}
      \Precondition $\phase[m] = \phCommit$ \labline{sufficiency:stable:2}
      \Precondition $\forall h \in \gamma(g) \sep (m,h) \in \LOG_{g}$ \labline{sufficiency:stable:3}
      \Effect $\phase[m] \assign \phStable$ \labline{sufficiency:stable:4}
    \end{action}
    
    \vspace{.5em}
    
    \begin{action}{$\deliver(m)$} \labline{sufficiency:deliver:1}
      \Precondition $\phase[m] = \phStable$ \labline{sufficiency:deliver:2}
      \Precondition $\forall m' <_{\LOG_{g \inter h}} m \sep \phase[m'] = \phDelivered$ \labline{sufficiency:deliver:3}
      \Effect $\phase[m] \assign \phDelivered$ \labline{sufficiency:deliver:4}
    \end{action}
    
  \end{algorithmic}

\end{algorithm}

\refalg{sufficiency} depicts a solution to (group sequential) genuine atomic multicast using failure detector $\WFD$.
To the best of our knowledge, this is only algorithm with \cite{SchiperP08} that tolerates arbitrary failures.
\refalg{sufficiency} is composed of a set of actions.
An action is executable once its preconditions (\textbf{pre:}) are true.
The effects (\textbf{eff:}) of an action are applied sequentially until it returns.
%
\refalg{sufficiency} uses a log per group and per group intersection.
Logs are linearizable, long-lived and wait-free.
Their sequential interface is detailed below.

\subparagraph*{Logs}
A \emph{log} is an infinite array of slots.
Slots are numbered from $1$.
Each slot contains one or more data items.
A datum $d$ is at position $k$ when slot $k$ contains it.
This position is obtained through a call to $\pos(d)$;
$0$ is returned if $d$ is absent.
A slot $k$ is \emph{free} when it contains no data item.
In the initial state, every slot is free.
The head of the log points to the first free slot after which there are only free slots (initially, slot $1$).
Operation $\append(d)$ inserts datum $d$ at the slot pointed by the head of the log then returns its position.
If $d$ is already in the log, this operation does nothing.
When $d$ is in the log, it can be locked with operation $\bumpAndLock(d,k)$.
This operation moves $d$ from its current slot $l$ to slot $\max(k,l)$, then locks it.
Once locked, a datum cannot be bumped anymore.
Operation $\locked(d)$ indicates if $d$ is locked in the log.
We write $d \in L$ when datum $d$ is at some position in the log $L$.
A log implies an ordering on the data items it contains.
When $d$ and $d'$ both appear in $L$, $d <_{L} d'$ is true when the position of $d$ is lower than the position of $d'$, or they both occupy the same slot and $d < d'$, for some a priori total order ($<$) over the data items.

\subparagraph*{Variables}
\refalg{sufficiency} employs two types of shared objects at a process.
First, for any two groups $h$ and $h'$ to which the local process belongs, \refalg{sufficiency} uses a log $\LOG_{h \inter h'}$ (\refline{sufficiency:var:1}).
Notice that, when $h=h'$, the log coincides with the log of the destination group $h$, i.e., $\LOG_h$.
Second, to agree on the final position of a message, \refalg{sufficiency} also employs consensus objects (\refline{sufficiency:var:fix}).
Consensus objects are both indexed by messages and group families.
Given some message $m$ and appropriate family $\f$, \refalg{sufficiency} calls $\CONS_{m,\f}$ (\reflinestwo{sufficiency:commit:fix:1}{sufficiency:commit:fix:2}).
Two processes call the same consensus object at \refline{sufficiency:commit:fix:2} only if both parameters match.
Finally, to store the status of messages addressed to the local process, \refalg{sufficiency} also employs a mapping $\phase$ (\refline{sufficiency:var:2}).
A message is initially in the $\phStart$ phase, then it moves to $\phPending$ (\refline{sufficiency:pending:7}), $\phCommit$ (\refline{sufficiency:commit:7}), $\phStable$ (\refline{sufficiency:stable:4}) and finally the $\phDelivered$ (\refline{sufficiency:deliver:4}) phase.
Phases are ordered according to this progression.

\subparagraph*{Algorithmic details}
We now detail \refalg{sufficiency} and jointly argue about its correctness.
\iflongversion
For clarity, our argumentation is informal---the full proof appears in \refsection{sufficiency:correctness}.
\else
For clarity, our argumentation is informal---the full proof appears in \cite{longversion}.
\fi

To multicast some message $m$ to $g=\dst(m)$, a process adds $m$ to the log of its destination group (\refline{sufficiency:amcast:3}).
When $p \in g$ observes $m$ in the log, $p$ appends $m$ to each $\LOG_{g \inter h}$  with $p \in g \inter h$ (\refline{sufficiency:pending:5}).
Then, $p$ stores in the log of the destination group of $m$ the slot occupied by $m$ in $\LOG_{g \inter h}$ (\refline{sufficiency:pending:6}).
This moves $m$ to the $\phPending$ phase.

Similarly to Skeen's algorithm \cite{Birman:1987}, a message is then bumped to the highest slot it occupies in the logs.
This step is executed at \reflines{sufficiency:commit:1}{sufficiency:commit:7}.
In detail, $p$ first agree with its peers on the highest position $k$ occupied by $m$ (\reflines{sufficiency:commit:4}{sufficiency:commit:fix:2}).
Observe here that only the processes in $g$ that share some cyclic family with $p$ take part to this agreement (\refline{sufficiency:commit:fix:1}).
%
%
Then, for each group $h$ in $\Gr(p)$, $p$ bumps $m$ to slot $k$ in $\LOG_{g \inter h}$ and locks it in this position (\refline{sufficiency:commit:6}).
This moves $m$ to the $\phCommit$ phase.

The next steps of \refalg{sufficiency} compute the predecessors of message $m$.
With more details, once $m$ reaches the $\phStable$ phase and is ready to be delivered, the messages that precede it in the logs at process $p$ cannot change anymore.
\iflongversion
In \refsection{sufficiency:correctness}, this invariant is proved in \reflem{sufficiency:ordering:3}.
\fi

If $g$ does not belong to any cyclic family, stabilizing $m$ is immediate:
the precondition at \refline{sufficiency:stable:3} is always vacuously true.
In this case, $m$ is delivered in an order consistent with the order it is added to the logs (\refline{sufficiency:stabilize:4}).
This comes from the fact that when $\Fa = \emptySet$ ordering the messages reduces to a pairwise agreement between the processes.
\iflongversion
\reflem{sufficiency:ordering:5} establish this second invariant in \refsection{sufficiency:correctness}.
\fi

When $\Fa \neq \emptySet$, stabilizing $m$ is a bit more involved.
Indeed, messages can be initially in cyclic positions, e.g., $C = m_1 <_{\LOG_{g_1 \inter g_2}} m_2 <_{\LOG_{g_2 \inter g_3}} m_3 <_{\LOG_{g_3 \inter g_1}} m_1$, preventing them to be delivered.
As in \cite{Birman:1987}, bumping messages helps to resolve such a situation.

The bumping procedure is executed globally.
A process must wait that the positions in the logs of a message are cycle-free before declaring it $\phStable$.
Waiting can cease when the cyclic family is faulty (\refline{sufficiency:stable:3}).
This is correct because messages are stabilized in the order of their positions in the logs (\reflines{sufficiency:stabilize:1}{sufficiency:stabilize:5}).
Hence, if a cycle $C$ exists initially in the positions, either
\begin{inparaenum}
\item not all the messages in $C$ are delivered, or 
\item the first message to get $\phStable$ in $C$ has no predecessors in $C$ in the logs.
\end{inparaenum}%
In other words, for any two messages $m$ and $m'$ in $C$, if $m \delOrder m'$ then $m$ is $\phStable$ before $m'$.
\iflongversion
In \refsection{sufficiency:correctness}, \reflem{sufficiency:ordering:cycle:3} proves this third invariant.
\fi

A process indicates that message $m$ with $g=\dst(m)$ is stabilized in group $h$ with a pair $(m,h)$ in $\LOG_g$ (\refline{sufficiency:stabilize:5}).
When this holds for all the groups $h$ intersecting with $g$ such that there exists a correct family $\f$ with $f \in \Fa(p)$ and $g,h \in \f$, $m$ is declared $\phStable$ at $p$ (\refline{sufficiency:stable:3}).
Once $\phStable$, a message $m$ can be delivered (\reflines{sufficiency:deliver:1}{sufficiency:deliver:4}).

\refalg{sufficiency} stabilizes then delivers messages according to their positions in the logs.
To maintain progress, these positions must remain acyclic at every correct process.
Furthermore, this should also happen globally when a cyclic family is correct.
Both properties are ensured by the calls to consensus objects (\refline{sufficiency:commit:fix:2}).
\iflongversion
In \refsection{sufficiency:correctness}, \reflemtwo{sufficiency:liveness:acyclic}{sufficiency:liveness:stable} establish that these two properties hold.
\fi

\subparagraph*{Implementing the shared objects}
In each group $g$, consensus is solvable since $\WFD$ provides $\Sigma_g \land \Omega_g$.
This serves to implement all the objects $(\CONS_{m,\f})_{m,\f}$ when $\dst(m)=g$.
Logs that are specific to a group, namely $(\LOG_g)_{g \in \Gr}$, are also built atop consensus in $g$ using a universal construction \cite{syn:1382}.

Failure detector $\WFD$ does not offer the means to solve consensus in $g \inter h$.
Hence we must rely on either $g$ or $h$ to build $\LOG_{g \inter h}$.
Minimality requires processes in a destination group to take steps only in the case a message is addressed to them.
To achieve this, we have to slightly modify the universal construction for $\LOG_{g \inter h}$, as detailed next.

First, we consider that this construction relies on an unbounded list of consensus objects.%
\footnote{
  In the failure detector model, computability results can use any amount of shared objects.
}
Each consensus object in this list is contention-free fast \cite{AttiyaGK05}.
This means that it is guarded by an adopt-commit object ($\AC$) \cite{Gafni98} before an actual consensus object $\CONS$ is called.
Upon calling $\propose$, $\AC$ is first used and if it fails, that is ``adopt'' is returned, $\CONS$ is called.
Adopt-commit objects are implemented using $\Sigma_{g \inter h}$, while consensus objects are implemented atop some group, say $g$, using $\Sigma_g \land \Omega_g$.
This modification ensures that when processes execute operations in the exact same order, only the adopt-commit objects are called.
As a consequence, when no message is addressed to either $g$ or $h$ during a run, only the processes in $g \inter h$ executes steps to implement an operation of $\LOG_{g \inter h}$.
\iflongversion
In \refsection{sufficiency:correctness}, \refprop{sufficiency:genuineness} establishes this result.

\subsection{Correctness of \refalg{sufficiency}}
\labsection{sufficiency:correctness}

In what follows, we prove that \refalg{sufficiency} is a genuine implementation of group sequential atomic multicast.
To establish this result, we proves a series of statement on the runs of \refalg{sufficiency}.
When considering such a run, as all the shared objects used in \refalg{sufficiency} are linearizable, we reason directly upon the linearization.

\subparagraph*{Additional notations \& definitions}
Consider some run of \refalg{sufficiency}.
We note $\del(m)$ when $m$ is delivered at some process and $\del_p(m)$ when $p$ has delivered $m$.
$\alwaysTL$ and $\eventuallyTL$ are respectively the ``always'' and ``eventually'' modal operators of temporal logic.
For some variable $x$, ${x}^{p,t}$ denotes to the value of $x$ at time $t$ on process $p$.
If $x$ is shared, we shall omit the name of the process.
When $t=\infty$, we refer to the value of $x$ at the end of the run.
We write $\phStable^{p,t}$ the stable messages at time $t$ at process $p$, that is $\phStable^{p,t} \equaldef \{m : \phase^{p,t}[m] = \phStable\}$.
For each phase $\phi \in \{\phStart, \phPending, \phCommit, \phDelivered\}$, the set $\phi^{p,t}$ is defined similarly.
$\phi^{t}$ are all the messages in phase $\phi$ across the system, i.e., $\phi^t \equaldef  \union_{p} \phi^{p,t}$.

\subsubsection{Trivias}
\labsection{sufficiency:correctness:trivias}

\begin{table}
  \small
  \setlength{\columnsep}{2em}
  \begin{multicols}{2}

    About logs.
    
    \begin{claim}
      \labclaim{sufficiency:trivias:logs:1}
      $\alwaysTL(d \in L \implies \alwaysTL(d \in L))$
    \end{claim}

    \begin{claim}
      $\alwaysTL(L.\pos(d)=k \implies \alwaysTL(L.\pos(d) \geq k))$
    \end{claim}

    \begin{claim}
      \labclaim{sufficiency:trivias:logs:2}
      $\alwaysTL(L.\locked(d) \implies \alwaysTL(L.\locked(d)))$
    \end{claim}

    \begin{claim}
      $\alwaysTL((L.\locked(d) \land L.\pos(d)=k) \implies \alwaysTL(L.\pos(d) = k))$
    \end{claim}
    
    \begin{claim}
      \labclaim{sufficiency:trivias:logs:3}
      $\alwaysTL(L.\locked(d) \land d <_{L} d' \implies \alwaysTL(d <_{L} d'))$
    \end{claim}

    \begin{claim}
      \labclaim{sufficiency:trivias:logs:4}
      $\alwaysTL((L.\locked(d') \land d \notin L \land \eventuallyTL(d \in L)) \implies \alwaysTL(d \in L \implies d' <_{L} d))$
    \end{claim}

    \begin{claim}
      \labclaim{sufficiency:trivias:logs:5}
      $\alwaysTL(L.\locked(d) \land \eventuallyTL(d' <_{L} d) \implies d' <_{L} d)$
    \end{claim}

    
    \columnbreak

    Others.

    \begin{claim}
      \labclaim{sufficiency:trivias:1}
      $(\dst(m) \inter \dst(m') \neq \emptySet \land \del(m) \land \del(m')) \implies m' \delOrder m \lor m \delOrder m'$
    \end{claim}
    
    \begin{claim}
      \labclaim{sufficiency:trivias:3}
      $\alwaysTL(m \in \LOG_{g \inter h} \implies \dst(m) = g \lor \dst(m) = h)$
    \end{claim}

    \begin{claim}
      \labclaim{sufficiency:trivias:4}
      $\alwaysTL(m' <_{\LOG_{x}} m \implies x = \dst(m) \inter \dst(m'))$
    \end{claim}

    \begin{claim}
      \labclaim{sufficiency:trivias:5}
      $\del_p(m) \implies p \in \dst(m)$
    \end{claim}

    \begin{claim}
      \labclaim{sufficiency:trivias:6}
      $\del(m) \implies m \in \LOG_{\dst(m)}^{\infty}$
    \end{claim}

    \begin{claim}
      \labclaim{sufficiency:trivias:7}
      $
      \phase^{p,t_0}[m]~=~\phDelivered
      \implies
      \phase^{p,t_1<t_0}[m] = \phStable
      \implies
      \phase^{p,t_2<t_1}[m] = \phCommit
      \implies
      \phase^{p,t_3<t_2}[m] = \phPending
      \implies
      \phase^{p,t_4<t_3}[m] = \phStart
      $
    \end{claim}

    \begin{claim}
      \labclaim{sufficiency:trivias:8}
      $\alwaysTL(\phase^{p}[m] = \phi \implies \alwaysTL(\phase^{p}[m] \geq \phi))$
    \end{claim}

  \end{multicols}

  \caption{
    \labtab{trivias}
    Base invariants in \refalg{sufficiency}.
  }
    
\end{table}

\reftab{trivias} lists some base invariants that follow immediately from the specification of a log and the structure of \refalg{sufficiency}.

\subsubsection{Ordering}
\labsection{sufficiency:correctness:ordering}

First, we establish that \refalg{sufficiency} maintains the ordering property of the group communication primitive.
Formally,

\begin{proposition}
  \labprop{sufficiency:ordering}
  In every run of \refalg{sufficiency}, relation $\delOrder$ is acyclic over $\msgSet$.
\end{proposition}

The proof goes by contradiction.
Let $C = m_1 \delOrder m_2 \delOrder \ldots m_{\cardinalOf{C}} \delOrder m_1$ be a cycle in $\delOrder$ that does not contain any sub-cycle.
By definition of relation $\delOrder$, for each message $m_{i \in [1,\cardinalOf{C}]}$, there exists a process $p_i$ that delivers $m_{i}$ before $m_{i+1}$.
We shall note $\f=(g_i)_i$ the family of the destination groups of $(m_i)_i$.
Below, we prove some key invariants regarding the way \refalg{sufficiency} delivers a message $m \in C$ with $g=\dst(m)$.

\begin{lemma}
  \lablem{sufficiency:ordering:1}
  $\phase^{p,t}[m] \geq \phCommit \implies \forall h \in \Gr(p) \sep \LOG_{g \inter h}^{t}.\locked(m)$
\end{lemma}

\begin{proof}
  By \refclaim{sufficiency:trivias:7}, if $m$ is $\phDelivered$, it must be first $\phStable$.
  Similarly, to be $\phStable$, message $m$ must be first $\phCommit$.
  Process $p$ commits message $m$ at \refline{sufficiency:commit:7}.
  Let it be at some time $t' \leq t$.
  Process $p$ executes \reflinestwo{sufficiency:commit:5}{sufficiency:commit:6} before time $t'$.
  Thus, $m$ is bumped and locked in each $\LOG_{g \inter h}$ with  $h \in \Gr(p)$.
  For each such group $h$, applying \refclaim{sufficiency:trivias:logs:2}, $\LOG_{g \inter h}^{t}.\locked(m)$.
\end{proof}

Let $m'$ be a message with $h=\dst(m')$.

\begin{lemma}
  \lablem{sufficiency:ordering:2}
  $(m \delOrderOf{p} m') \implies \neg~(m' \delOrderOf{p} m)$
\end{lemma}

\begin{proof}
  To get delivered, a message must first reach the $\phDelivered$ phase (\refline{sufficiency:deliver:4}).
  Each action is guarded by a precondition on the phase of the message---precisely at line \ref{line:alg:sufficiency:pending:2}, \ref{line:alg:sufficiency:commit:2}, \ref{line:alg:sufficiency:stabilize:2}, \ref{line:alg:sufficiency:stable:2} and \ref{line:alg:sufficiency:deliver:2}.
  It follows that the $\phDelivered$ phase is terminal.
  Consequently, a message is delivered at most once.
  Thus $m \neq m'$.
  Assume $m \delOrderOf{p} m'$.
  Thus, process $p$ delivers $m$ and either this occurs before $p$ delivers $m'$ or $p$ never delivers $m'$.
  In both cases, $\neg~(m' \delOrderOf{p} m)$.
\end{proof}

\begin{lemma}
  \lablem{sufficiency:ordering:3}
  $m \delOrderOf{p} m' \implies (m <_{\LOG_{g \inter h}^{\infty}} m' \lor m' \notin \LOG_{g \inter h}^{\infty})$
\end{lemma}

\begin{proof}
  (By contradiction.)
  As $m \delOrderOf{p} m'$, by \refclaim{sufficiency:trivias:5}, $p \in g$.
  Then, applying the definition of relation $\delOrder$, $p \in h$.
  Process $p$ delivers $m$ at \refline{sufficiency:deliver:4}.
  Let $t$ be the time at which this happens.
  According to the precondition at \refline{sufficiency:deliver:2}, message $m$ was stable at $p$ before.
  Thus, applying \reflem{sufficiency:ordering:1}, $\LOG_{g \inter h}^{t}.\locked(m)$.
  By our contradiction hypothesis, $m' <_{\LOG_{g \inter h}^{\infty}} m$.
  Hence, $\LOG_{g \inter h}^{p,t}.\locked(m) \land \eventuallyTL(m' <_{\LOG_{g \inter h}} m)$.
  By \refclaim{sufficiency:trivias:logs:5}, $m' <_{\LOG_{g \inter h}^{p,t}} m$.
  The precondition at \refline{sufficiency:deliver:3} requires that $\phase^{p,t}[m']=\phDelivered$.
  Thus, $m'$ was delivered before.
  Applying \reflem{sufficiency:ordering:2}, $\neg (m \delOrderOf{p} m')$.
\end{proof}

\begin{lemma}
  \lablem{sufficiency:ordering:4}
  $(m \delOrderOf{p} m') \implies \neg (m' \delOrderOf{q} m)$
\end{lemma}

\begin{proof}
  (By contradiction)
  Messages $m$ and $m'$ are delivered by respectively $p$ and $q$.
  From \reflem{sufficiency:ordering:1}, $m \in \LOG_{g \inter h}^{\infty}$ and $m' \in \LOG_{g \inter h}^{\infty}$.
  Hence, by \reflem{sufficiency:ordering:3}, $m <_{\LOG_{g \inter h}^{\infty}} m'$ and $m' <_{\LOG_{g \inter h}^{\infty}} m$.
\end{proof}

\begin{lemma}
  \lablem{sufficiency:ordering:5}
  $f \notin \Fa \implies \cardinalOf{C} \in {1,2}$
\end{lemma}

\begin{proof}
  (By contradiction.)
  As $\cardinalOf{C} \geq 3$, we have $C = m_1 \delOrderOf{p_1} m_2 \delOrderOf{p_2} m_3 \ldots$.
  Consider that all the destination groups are pair-wise distinct.
  Then $\pi=g_1g_2{\ldots}g_{\cardinalOf{C} \geq 3}g_1$ is a closed path in the intersection graph of $\f$ visiting all the destination groups.
  A contradiction to $\f \notin \Fa$.
  Without lack of generality assume that $g_1 = g_{2 \leq k \leq \cardinalOf{C}}$.
  Since process $p_1$ delivers $m_1$, either $m_1 \delOrderOf{p_1} m_{k}$ or the converse holds.
  (Case $m_1 \delOrderOf{p_1} m_{k}$)
  Then $m_1 \delOrderOf{p_{1}} m_{k} \delOrderOf{p_{k+1}} \ldots m_{\cardinalOf{C}} \delOrderOf{p_{\cardinalOf{C}}} m_1$ is a sub-cycle of $C$.
  (Otherwise)
  Applying~\reflem{sufficiency:ordering:4}, $m_k \delOrderOf{p_{k}} m_{1} \delOrderOf{p_1} \ldots m_{k-1} \delOrderOf{p_{k-1}} m_k$ is a sub-cycle of $C$.
\end{proof}

\subparagraph*{(Case $\f \notin \Fa$.)}
\reflem{sufficiency:ordering:5} establishes that, as $C$ does not contain any sub-cycle, its size must be either $1$ or $2$.
$\cardinalOf{C}=1$ is forbidden by \reflem{sufficiency:ordering:2}.
\reflem{sufficiency:ordering:4} prevents $\cardinalOf{C}=2$.

\subparagraph*{(Otherwise.)}
Clearly, there exists some time $t_0 > 0$ until which family $\f \in \Fa$ is correct.
%
Let $m \in C$ be a message and $g=\dst(m)$ be its destination group.
Assume that $m$ is $\phStable$ at some process $p$ with $\f \in \Fa(p)$.
Consider this happens first at time $t \leq t_0$.
Then,

\begin{lemma}
  \lablem{sufficiency:ordering:cycle:1}
  $\forall h \in \Fa(g) \sep (m,h) \in \LOG_g^{t}$
\end{lemma}

\begin{proof}
  By definition, $p$ executes \refline{sufficiency:stable:4} at time $t$.
  Since $\f$ is still correct at time $t$ and $\gamma$ ensures the accuracy property, $f \in \gamma(p,t)$.
  Hence, $\gamma^{p,t}(g)=\Fa(g)$.
  According to \refline{sufficiency:stable:3}, $\forall h \in \gamma^{p,t}(g) \sep (m,h) \in \LOG^{t}_{g}$.  
\end{proof}

Choose $h \in \Fa(g)$ and $\hat{t} \leq t$.

\begin{lemma}
  \lablem{sufficiency:ordering:cycle:2}
  $(m,h) \in \LOG_g^{\hat{t}} \implies \LOG_{g \inter h}^{\hat{t}}.\locked(m)$
\end{lemma}

\begin{proof}
  Some process $q$ executes \refline{sufficiency:stabilize:5} at time $t' \leq \hat{t}$.
  From the preconditions at \reflines{sufficiency:stabilize:2}{sufficiency:stabilize:4}, $q \in h$ and $\phase^{q,t'}[m]=\phCommit$.
  From $\phase^{q,t'}[m]=\phCommit$, $q$ executes \reflines{sufficiency:commit:4}{sufficiency:commit:7} at some earlier time $t'' < t'$.
  As $q \in h$, $q$ called $\LOG_{g \inter h}.\bumpAndLock(m,\any)$ at time $t''$.
\end{proof}

Consider two messages $m$ and $m'$ such that $m' \delOrder m \in C$.
Note $\phStable^{\f,t}$ the messages $\phStable$ at time $t$ at any process $p$ with $\f \in \Fa(p)$.
%
Then,

\begin{lemma}
  \lablem{sufficiency:ordering:cycle:3}
  $\forall t \leq t_0 \sep m \in \phStable^{\f,t} \implies m' \in \phStable^{\f,t}$
\end{lemma}

\begin{proof}
  Let $h=\dst(m')$.
  Then, $h \in \Fa(g)$.
  From \reflem{sufficiency:ordering:cycle:1}, $(m,h) \in \LOG_{g}^{t}$.
  Thus, \refline{sufficiency:stabilize:5} was executed before time $t$.
  Let $p \in g \inter h$ be the process that executes this step first at some time $t'<t$.
  If $m' \in C$, then some process delivers it.
  Applying \reflem{sufficiency:ordering:3}, $m' <_{\LOG_{g \inter h}^{\infty}} m$.
  Then, from \reflem{sufficiency:ordering:cycle:2}, $\LOG_{g \inter h}^{t'}.\locked(m)$
  By \refclaim{sufficiency:trivias:logs:5}, $m' <_{\LOG_{g \inter h}^{t'}} m$.
  As $m <_{\LOG_{g \inter h}^{t'}} m'$, the precondition at \refline{sufficiency:stabilize:4} requires that $\phase^{t',p}[m]=\phStable$
\end{proof}

Without lack of generality, consider that $m_{\cardinalOf{C}}$ is the message delivered first in $C$, say at process $p$ and time $t$.
At time $t-1$, according to the precondition at \refline{sufficiency:deliver:2}, $m_{\cardinalOf{C}}$ is stable at $p$.
Let $t'\leq t-1$ be the time at which $m_{\cardinalOf{C}}$ is stable first.
By applying inductively \reflem{sufficiency:ordering:cycle:3} to $C$, we obtain a contradiction.

\subsubsection{Liveness}
\labsection{sufficiency:correctness:liveness}

In what follows, we prove that \refalg{sufficiency} ensures the termination property of atomic multicast.
To this end, we first observe that all of the objects used in this algorithm are wait-free.
As a consequence, no process may block at some point in time when executing one of the actions of \refalg{sufficiency}.
This base observation is commonly used in our proof.

Consider that some message $m$ is multicast by a correct process.
It follows that \refline{sufficiency:amcast:3} is executed.
Similarly, if $m$ is delivered at some process, it must have reached the $\phPending$ phase at that process (\refclaim{sufficiency:trivias:7}).
At that time, the precondition at \refline{sufficiency:pending:3} is true.
Hence, in the above two cases, $m$ reaches $\LOG_g$ at some point in time.
By \refclaim{sufficiency:trivias:logs:1}, $\eventuallyTL(\alwaysTL(m \in \LOG_g))$.
In what follows, we establish that any such message $m$ reaching the log of its destination group $g$ is eventually delivered at every correct process in $g$.

\begin{lemma}
  \lablem{sufficiency:liveness:1}
  $\alwaysTL((\forall p \in \correct \inter g \sep \eventuallyTL(\phase^p[m] \geq \phPending)) \implies (\forall p \in \correct \inter g \sep \eventuallyTL(\phase^p[m] \geq \phCommit)))$
\end{lemma}

\begin{proof}
  Pick $q \in \correct \inter g$.
  By assumption, $m$ is eventually $\phPending$ at $q$.
  If $m$ reaches a higher phase than $\phPending$, applying \refclaim{sufficiency:trivias:7}, $m$ must go through the $\phCommit$ phase at $q$.
  Otherwise, consider that $m$ is always $\phPending$ at $q$.
  We show that the precondition at \refline{sufficiency:commit:3} is eventually always true, leading to a contradiction.
  
  Consider $h \in \Fa(g)$.
  If $(g \inter h)$ is faulty, then eventually every cyclic family $\f$, with $g, h \in \f$, is also faulty.
  By the completeness property of $\gamma$, eventually $h \notin \gamma(g)$ always holds at $q$.
  Otherwise, let $q'$ be a correct process in $g \inter h$.
  By assumption and \refclaim{sufficiency:trivias:7}, $q'$ eventually executes \refline{sufficiency:pending:6}, adding a tuple $(m,h,\any)$ to $\LOG_g$.
  Thus, the precondition at \refline{sufficiency:commit:3} is eventually always true at process $q$, as required.  
\end{proof}

\begin{lemma}
  \lablem{sufficiency:liveness:2}
  $\alwaysTL(m <_{\LOG_g} m' \implies \alwaysTL(m <_{\LOG_g} m'))$
\end{lemma}

\begin{proof}
  By \refclaim{sufficiency:trivias:3}, $\dst(m) = \dst(m') = g$.
  If $m <_{\LOG_g} m'$, necessarily $m$ and $m'$ are both in $\LOG_g$.
  A message is added to $\LOG_g$ either at \refline{sufficiency:amcast:3} or \refline{sufficiency:pending:5}.
  Observe that in the later case, it must already be in $\LOG_g$.
  In fact, only the ordering of the execution of \refline{sufficiency:amcast:3} matters to prove the lemma.
  Notice that this line is necessarily executed by the process multicasting the message.

  For starters, assume that \refline{sufficiency:amcast:3} is executed first for $m$.
  Since the problem being solved is group sequential atomic multicast, it must be the case that $m \hb m'$.
  Name $q$ the process that multicasts $m'$ at some time $t$.
  As $m \hb m'$, $\phase^{q,t}[m] = \phDelivered$.
  By \reflem{sufficiency:ordering:1}, $\LOG_{g}^{t}.\locked(m)$.
  By \refclaim{sufficiency:trivias:logs:4}, $\alwaysTL(m <_{\LOG_g} m')$ from that point in time.
  Otherwise, \refline{sufficiency:amcast:3} is executed first for $m'$.
  In that case, following the same reasoning as above, $m'$ is locked in $\LOG_{g}$ at the time $m$ is appended to the log.
  By \refclaim{sufficiency:trivias:logs:4}, $\alwaysTL(\neg~(m <_{\LOG_g} m'))$.
\end{proof}

From which, it follows:

\begin{corollary}
  \labcor{sufficiency:liveness:1}
  $\forall t, t' < t \sep m <_{\LOG_g^{t}} m' \implies (m <_{\LOG_g^{t'}} m' \lor m' \notin \LOG_g^{t'})$
\end{corollary}

Let $\min(\LOG_g)$ be the message at the lowest position in $\LOG_g$ such that for some $p \in \correct \inter g$, $\phase^{p}[m] \leq \phPending$.

\begin{lemma}
  \lablem{sufficiency:liveness:4}
  $\alwaysTL( m = \min(\LOG_g) \implies \eventuallyTL(\forall p \in \correct \inter g \sep \phase^{p}[m] = \phCommit))$
\end{lemma}

\begin{proof}
  Let $m$ be $\min(\LOG_g)$ at some point in time $t_0$.
  Consider a process $p \in \correct \inter g$.
  Initially $\phase^{p}[m]=\phStart$.
  For the sake of contradiction, consider that this holds forever after $t_0$
  Necessarily, the precondition at \refline{sufficiency:pending:3b} is always false.
  Thus, at $t > t_0$, there must exist $m'$ such that $m' <_{\LOG_g^{t'}} m$ and $\phase^{p,t}[m'] \leq \phPending$.
  By \refcor{sufficiency:liveness:1} and \refclaim{sufficiency:trivias:8}, this is also the case at time $t_0$.
  A contradiction to the very definition of $m$.
  As a consequence, $m$ moves eventually to a higher phase than $\phStart$.
  Let $t_p \geq t_0$ be the time at which this happens.
  At time $t = \max_{p \in \correct \inter g} t_p$, $m$ is in a phase $\phi \geq \phPending$ at all the processes in $\correct \inter g$.
  By \reflem{sufficiency:liveness:1}, $m$ reaches later eventually the $\phCommit$ phase at all such processes.
\end{proof}

Consider a message $m$ with $g=\dst(m)$, either multicast by a correct process or eventually delivered.

\begin{lemma}
  \lablem{sufficiency:liveness:commit}
  $\forall p \in \correct \inter g \sep \alwaysTL(m \in \LOG_g \implies \eventuallyTL(\phase^{p}[m] = \phCommit))$
\end{lemma}

\begin{proof}
  First consider that $\neg~\eventuallyTL(m = \min(\LOG_g))$.
  By definition of $\min(\LOG_g)$, $\alwaysTL(\forall p \in \correct \inter g$, $\phase^{p}[m] > \phPending)$.
  We conclude using \refclaim{sufficiency:trivias:7}.
  Otherwise, $m$ is at some point in time equal to $\min(\LOG_g)$.
  In this case, \reflem{sufficiency:liveness:4} tells us that the lemma holds.
\end{proof}

From what precedes, a message reaching the log of its destination group is eventually committed at every correct process in $g$.
Next, we prove that it eventually stabilizes.
As pointed in \refsection{sufficiency:algorithm}, if $\Fa$ is empty then a pair-wise agreement on the order of messages is sufficient to reach the $\phStable$ phase.
However, in the general case, stabilization must happen following the order in which messages appears in \emph{all} the logs (\refline{sufficiency:stabilize:4}).
As a consequence, we must show that there is no cycle to ensure that this procedure eventually terminates.
To do so, we first prove that in the absence of failures of the cyclic family, the message occupies the same final position in the logs of the intersection groups.

\tikzstyle{basic} = [draw, -latex',Black,-]
\tikzstyle{dot} = [inner sep=1.5pt, circle, fill]
\tikzstyle{dasheddot} = [inner sep=1.5pt, circle, fill]

\tikzstyle{familyf} = [draw,-latex',Blue,-]
\tikzstyle{familyfprime} = [draw,-latex',Red,-]
\tikzstyle{familyfsecond} = [draw,-latex',OliveGreen,-]

\begin{figure}[t]
  \centering
  \captionsetup{justification=centering}
    \begin{subfigure}[t]{.3\textwidth}
    \centering
    \begin{tikzpicture}[
        every path/.style={black},
        scale=1,
        transform shape
      ]

      \node[dot,label=90:$g$] (g) at (0,0) {};
      \node[dot,label=90:$g'$] (gprime) at (-1.4,0) {};
      \node[dot,label=180:$\hat{g}$] (ghat) at (-1,-.8) {};
      \node[dot,label=90:$g''$] (gsecond) at (1,1) {};
      \node[dot,label=270:$h$] (h) at (0,-1.4) {};

      \path[basic] (g) edge node [midway,below,xshift=-.5em] {$\{p\}$} (gprime);
      \path[basic] (g) edge node [midway,right] {$\{p'\}$} (gsecond);
      \path[basic] (g) -- (ghat);
      \path[basic] (gprime) -- (ghat);     

      \node[label={\color{Blue}{\f}}] (f) at (-.5,1) {};
      \draw[familyf]  ([yshift=-1pt]gprime.north) to[bend left=30] ([yshift=-1pt,xshift=-1pt]gsecond);
      \draw[familyf]  ([yshift=1pt]gprime.east) to ([yshift=1pt]g);
      \draw[familyf]  ([xshift=1pt]g.north) to ([yshift=-1pt,xshift=-2pt]gsecond);

      \node[label={\color{Red}{\f'}}] (fprime) at (.6,-1) {};
      \draw[familyfprime]  ([yshift=-1pt]ghat.north) to ([yshift=-1pt]g);
      \draw[familyfprime]  ([xshift=-1pt]g.south) to ([yshift=1pt,xshift=-1pt]h.north);
      \draw[familyfprime]  ([xshift=1pt]ghat) to[bend right=30] ([yshift=2pt]h.south);
    \end{tikzpicture}
    \caption{\labfigure{fix:1}}
  \end{subfigure}
  \begin{subfigure}[t]{.3\textwidth}
    \centering
    \begin{tikzpicture}[
        every path/.style={black},
        scale=1,
        transform shape
      ]

      \node[dot,label=90:$g$] (g) at (0,0) {};
      \node[dot,label=90:$g'$] (gprime) at (-1.4,0) {};
      \node[dot,label=180:$\hat{g}$] (ghat) at (-1,-.8) {};
      \node[dot,label=90:$g''$] (gsecond) at (1,1) {};
      \node[dot,label=270:$h$] (h) at (0,-1.4) {};

      \node[label={\color{OliveGreen}{\f''}}] (fsecond) at (.5,-.5) {};
      \draw[familyfsecond]  ([yshift=-1pt]gprime.north) to[bend left=30] ([yshift=-1pt,xshift=-1pt]gsecond);
      \draw[familyfsecond]  ([xshift=1pt]g.north) to ([yshift=-1pt,xshift=-2pt]gsecond);
      \draw[familyfsecond]  ([xshift=-1pt]g.south) to ([yshift=1pt,xshift=-1pt]h.north);
      \draw[familyfsecond]  ([xshift=1pt]gprime.south) to[bend right=30] ([yshift=2pt]h.west);      
    \end{tikzpicture}
    \caption{\labfigure{fix:2}}
  \end{subfigure}
  \begin{subfigure}[t]{.3\textwidth}
    \centering
    \begin{tikzpicture}[
        every path/.style={black},
        scale=1,
        transform shape
      ]

      \node[dot,label=90:$g$] (g) at (0,0) {};
      \node[dot,label=90:$g''$] (gsecond) at (1,1) {};
      \node[dasheddot] () at (.5,1) {};
      \node[dasheddot] () at (.1,.95) {};
      \node[dot,label=90:$h_{i-1}$] (himinusone) at (-.3,.85) {};
      \node[dot,label=270:$h$] (h) at (0,-1.4) {};
      \node[dasheddot] () at (-.25,-1.3) {};
      \node[dasheddot] () at (-.5,-1.17) {};
      \node[dot,label=270:$h_i$] (hi) at (-.75,-1) {};

      \node[label={\color{OliveGreen}{\f''}}] (fsecond) at (.5,-.5) {};
      \draw[familyfsecond]  ([yshift=-1pt]himinusone) to[bend left=10] ([yshift=-1pt,xshift=-1pt]gsecond);
      \draw[familyfsecond]  ([xshift=1pt]g.north) to ([yshift=-1pt,xshift=-2pt]gsecond);
      \draw[familyfsecond]  ([xshift=-1pt]g.south) to ([yshift=1pt,xshift=-1pt]h.north);
      \draw[familyfsecond]  ([xshift=1pt]hi.east) to[bend right=5] ([yshift=2pt]h.west);
      \draw[familyfsecond]  (hi.north) to (himinusone.south);
    \end{tikzpicture}
    \caption{\labfigure{fix:3}}
  \end{subfigure}    
  \caption{
    Illustration of the construction of family $\f'' \in \Fa(p')$ in \reflem{sufficiency:liveness:fix:1}:
    \emph{
      (a) the two families $\f$ and $\f'$,
      (b) case $\f \inter \f' = \{g\}$,
      (c) case $\f \inter \f' \setminus \{g\} \neq \emptySet$.
    }
    \labfigure{fix}
  }
\end{figure}

\begin{lemma}
  \lablem{sufficiency:liveness:fix:1}
  For some group $g$ and some process $q$, define $H(q,g)$ as $\{ h : \exists \f' \in \Fa(q) \sep g,h \in \f' \land g \inter h \neq \emptySet \}$.
  Consider some $\f \in \Fa$ with $g,g',g'' \in \f$, $p \in g \inter g'$ and $p' \in g \inter g''$.
  Then, $H(p,g) = H(p',g)$.
\end{lemma}

\begin{proof}
  The proof is illustrated in \reffigure{fix}.
  Choose $h \in H(p,g)$.  
  By definition, there exists $\f' \in \Fa(p)$ such that $g \inter h \neq \emptySet$ and $g,h \in \f'$.
  For some appropriate numbering, $\f'=(g_i)_{i \in [1,n]}$ with $g_1=g$ and $g_2=h$. 
  Similarly, we have $\f=(h_i)_{i}$ with $h_1=g'$, $h_2=g$ and $h_3=g''$.
  As $\f' \in \Fa(p)$ and $g \in \f'$, there exists $\hat{g} \in \f'$ such that $p  \in g \inter \hat{g}$.
  Hence, $p \in g' \inter \hat{g}$.
  We construct a cyclic family $\f''$ as follows.
  Assume $\f$ and $\f'$ have no other common group than $g$.
  Then $\f''$ is simply set to $\f \union \f'$.
  Otherwise, let $h_i$ be the first group in $(\f \inter \f') \setminus \{g\}$.
  $\f''$ is set to $g_1, g_2, \ldots, h_i, h_{i-1}, \ldots h_3$.
  As $p' \in \f''$, $g,h \in \f''$ and $g \inter h \neq \emptySet$, we have $h \in H(p',g)$.
\end{proof}

\begin{lemma}
  \lablem{sufficiency:liveness:fix:2}
  Assume some process $p$ executes $\CONS_{m,\any}.\propose(k)$ at time $t$ (\refline{sufficiency:commit:fix:2}).
  For any $\f \in \Fa(p) \inter \correct$, if $h \in \f$ and $g \inter h \neq \emptySet$, then $(m,h,l) \in \LOG_g^t$ and $k \geq l$.
\end{lemma}

\begin{proof}
  As $\f$ is correct, $\f \in \gamma^{p,t}$ by the accuracy property of failure detector $\gamma$.
  The precondition at \refline{sufficiency:commit:3} requires that $(m,h,l) \in \LOG_g^t$, for some $l$.
  Variable $k$ is computed at \refline{sufficiency:commit:4}.
  Hence, $k \geq l$ when $p$ calls \refline{sufficiency:commit:fix:2}.
\end{proof}

\begin{lemma}
  \lablem{sufficiency:liveness:fix:3}
  Let $\f$ be a correct cyclic family and $g$, $g'$ and $g''$ be three groups in $\f$ with $g \inter g' \neq \emptySet$ and $g \inter g'' \neq \emptySet$.
  Assume that message $m$ with $\dst(m)=g$ is locked in both in both $\LOG_{g \inter g'}$ and $\LOG_{g \inter g''}$.
  Then, $m$ has the same position in the two logs.
\end{lemma}

\begin{proof}
  Name $p$ and $p'$ respectively the processes that bump (and lock) $m$ in $\LOG_{g \inter g'}$ and $\LOG_{g \inter g''}$.
  Let $l$ and $l'$ be the respective positions of $m$ in the two logs before the bump.
  Applying \reflem{sufficiency:liveness:fix:1}, both processes call the same instance $\CONS_{m,\f'}$ at \refline{sufficiency:commit:fix:2} before executing \refline{sufficiency:commit:6}, for some appropriate family $\f'$.
  Let $q$ be the process whose value, say $\kappa$, is decided in $\CONS_{m,\f'}$.
  Processes $p$ and $p'$ execute $\bumpAndLock(m,\kappa)$ on the two logs.
  Observe that as $g' \in f'$, $\hat{\f} \in \Fa(q)$ for some cyclic family $\hat{f}$ with $g,g' \in \hat{\f}$
  The same observation is true for the group $g''$.
  Thus, by \reflem{sufficiency:liveness:fix:2}, $\kappa \geq \max(l,l')$.
  It follows that $m$ occupies the same final position in the two logs.
\end{proof}

To track the relative positions of messages in the logs, we introduce a few notations.
First, the predecessors of a message $m$ in the logs \emph{at a correct process} $p$ is denoted $\pred(m,p) \equaldef \{m' : \exists L \in \LOGS(p) \sep m' <_{L} m \}$, where $\LOGS(p)$ are all the logs at process $p$.
We shall write $m' <_p m$ when $m'$ precedes $m$ at some log at $p$, i.e., $m' <_{p} m \equaldef m' \in \pred(m,p)$.
The transitive closure of relation $<_p$ across all the correct processes is denoted $<$, that is $< \equaldef (\union_{p} <_p)^{*}$.
The predecessors of $m$ according to $<$ are $\pred(m) \equaldef \{ m' : m' < m \}$.

\begin{lemma}
  \lablem{sufficiency:liveness:acyclic}
  $\forall p \in \correct \sep \alwaysTL((\phi \geq \phCommit)^p, <_p)~\text{is a pos})$
\end{lemma}

\begin{proof}

  For starters, we show that messages addressed to the same destination group form a monotonically growing totally-ordered set.
  
  \begin{claim}
    \labclaim{sufficiency:liveness:acyclic:1}
    $\alwaysTL(m <_{p} m' \land \dst(m) = \dst(m') \implies \alwaysTL(m <_{p} m'))$
  \end{claim}

  \begin{proof}
    Since the problem being solved is group sequential atomic multicast, either $m \hb m'$ or the converse holds.
    Without lack of generality, consider the former case.
    Then $\multicast(m')$ occurs at $q \in g$ and time $t$, after $m$ was delivered by $q$.
    By \reflem{sufficiency:ordering:1}, $\LOG_g^{q,t}.\locked(m)$.
    Consider $L \in \LOGS(p)$ that contains both $m$ and $m'$.      
    Message $m'$ is added to $L$, say by $q'$ at time $t'>t$.
    Clearly $q$ executes either \refline{sufficiency:amcast:3} or \refline{sufficiency:pending:5}.
    In both cases, we have $L^{t'}.\locked(m) \land m <_{L^{t'}} m'$:
    \begin{inparaenumorig}[]
    \item the case \refline{sufficiency:amcast:3} is trivial; otherwise
    \item the precondition at \refline{sufficiency:pending:3b} requires $\phase[m]=\phCommit$, \reflem{sufficiency:ordering:1} leads to $L^{t'}.\locked(m)$.
    \end{inparaenumorig}%
    We then conclude with \refclaim{sufficiency:trivias:logs:3}.    
  \end{proof}

  The case $\cardinalOf{\LOGS(p)} \leq 3$ is immediate:
  There are exactly $\cardinalOf{\Gr(p)}+\cardinalOf{\Gr(p)}*(\cardinalOf{\Gr(p)}-1)/2$ logs at $p$.
  Hence, either there is a single log at process $p$ and it should be the one of a destination group; or, there are exactly two groups intersect in $p$.
  In the two cases, \refclaim{sufficiency:liveness:acyclic:1} implies that $<_p$ is acyclic.
  
  From now, we assume $\cardinalOf{\LOGS(p)}>3$.
  
  \begin{claim}
    \labclaim{sufficiency:liveness:acyclic:2}
    $\exists K \sep \forall h \neq g \in \Gr(p) \sep \alwaysTL(\LOG_{g \inter h}.\locked(m) \implies \LOG_{g \inter h}.\pos(m) = K)$
  \end{claim}

  \begin{proof}
    Consider some family $\f=\{g,h,h'\}$, with $\f \subseteq \Gr(p)$.
    As $p$ belongs to all three groups, $\f$ is cyclic.
    By definition, process $p$ is correct, and thus family $\f$ too.
    Applying \reflem{sufficiency:liveness:fix:3}, if $m$ is locked in both $\LOG_{g \inter h}$ and $\LOG_{g \inter h'}$ then it occupies the same position.
  \end{proof}

  \begin{claim}
    \labclaim{sufficiency:liveness:acyclic:3}
    Consider $m$, $m'$ and $m''$ with $\phase^p[m'] \geq \phCommit$.
    Then, for any $L,L' \in \LOGS(p)$, $m <_L m' <_{L'} m'' \implies (L.\pos(m) < L'.pos(m'') \lor (L.\pos(m) = L'.pos(m'') \land m < m''))$
  \end{claim}

  \begin{proof}
    As the three message are (at least) committed, they must be all locked in the logs of $p$ (\reflem{sufficiency:ordering:1}).
    Applying $m <_L m'$, $L.\pos(m) < L.pos(m') \lor (L.\pos(m) = L.pos(m') \land m < m')$
    Similarly, $(L'.\pos(m') < L'.pos(m'') \lor (L'.\pos(m') = L'.pos(m'') \land m' < m'')$
    Observe that $L.\pos(m') = L'.pos(m')$:
    if $L=L'$ this is trivial, and otherwise we apply \refclaim{sufficiency:liveness:acyclic:2}, 
  \end{proof}

  From \refclaim{sufficiency:liveness:acyclic:3}, $((\phi \geq \phCommit)^p, <_p)$ is a partially ordered set.
\end{proof}

\begin{lemma}
  \lablem{sufficiency:liveness:stable}
  $\forall p \in \correct \inter g \sep \alwaysTL(m \in \LOG_{g} \implies \eventuallyTL(\phase^{p}[m] = \phStable))$
\end{lemma}

\begin{proof}

  Let $\mathcal{C}(m,t) \equaldef \pred^{t}(m) \inter (\phi \leq \phCommit)^{t}$.
  The series $(\mathcal{C}(m,t))_{t}$ is defined over $2^{\msgSet}$, which is a Cauchy space, hence it may converge.
  Assume that the lemma below holds.

  \begin{claim}
    \labclaim{sufficiency:liveness:stable:1}
    $\lim_{t \rightarrow \infty} \mathcal{C}(m,t) = \emptySet$
  \end{claim}

  By \reflem{sufficiency:liveness:commit}, message $m$ is eventually $\phCommit$ at process $p$, say at time $t_0$.
  Consider the first moment in time at which $(\mathcal{C}(m,t))_{t}$ reaches its limit (by \refclaim{sufficiency:liveness:stable:1}), say $t_1$.
  At $\max(t_0,t_1)$, at process $p$, all the preconditions at \reflinestwo{sufficiency:stable:2}{sufficiency:stable:3} are true.
  Hence $p$ moves $m$ to the $\phStable$ phase, as required.
  The remainder of the proof is devoted to showing that \refclaim{sufficiency:liveness:stable:1} holds.

  \begin{claim}
    \labclaim{sufficiency:liveness:stable:2}
    $\forall m' \sep m' \in \mathcal{C}(m,t) \land \exists t'>t \sep m' \notin \mathcal{C}(m,t') \implies \forall t'' > t' \sep m' \notin \mathcal{C}(m,t'')$
  \end{claim}

  \begin{proof}
    By definition of $\mathcal{C}(m,t)$ and \refclaim{sufficiency:trivias:8}.
  \end{proof}
  
  \begin{claim}
    \labclaim{sufficiency:liveness:stable:3}
    $(\mathcal{C}(m,t))_{t}$ has a finite upper bound in $2^{\msgSet}$.
  \end{claim}
  
  \begin{proof}
    Immediate because the problem being solved is group sequential atomic multicast.
  \end{proof}
  
  Let $\mathcal{C}$ be the finite upper bound of $(\mathcal{C}(m,t))_{t}$.
  For $m' \in \mathcal{C}$, define $t_c$ either the point in time satisfying \refclaim{sufficiency:liveness:stable:2}, or $0$ if there is no such point.
  Let $T=\max_{m' \in \mathcal{C}} t_c$.

  \begin{claim}
    \labclaim{sufficiency:liveness:stable:4}
    $\forall t \geq T \sep \mathcal{C}(m,t+1) \subseteq \mathcal{C}(m,t)$
  \end{claim}

  \begin{proof}
    By definition of $t_c$ for every message $m' \in \mathcal{C}$.
  \end{proof}

  From what follows $(\mathcal{C}(m,t))_{t}$ is monotonically decreasing after time $T$.
  Let $\mathcal{C}$ be its limit.
  We now establish that $\mathcal{C}$ is empty.
  If this is the case, $(\mathcal{C}(m,t))_{t}$ converges toward $\emptySet$, as required by \refclaim{sufficiency:liveness:stable:1}.

  \begin{claim}
    \labclaim{sufficiency:liveness:stable:5}
    $\alwaysTL((\mathcal{C},<)~\text{is a pos.})$.
  \end{claim}

  \begin{proof}
    (By contradiction)
    Let $C=m_1 < \ldots < m_K < m_1$ be a cycle in $(\mathcal{C},<)$, with $K = \cardinalOf{C} \geq 2$.
    We note $\f'=(g_k)_{k \in [1,K]}$ the family of destination groups.
    Without lack of generality, consider that there is no sub-cycle in $C$.
    First, since there is no sub-cycle in $C$, all messages in $C$ must have different destination groups.
    Indeed, if $m_{k}$ and $m_{k'}$ are both addressed to the same destination group then they are both forever in some phase $\phi \leq \phCommit$.
    This contradicts that the problem solved is group sequential atomic multicast.
    Second, applying the definition of $<$, for each $k \in [1,K]$, there exists a correct process $p_k$ and some log $L_k \in \LOGS(p_k)$ such that $m_k <_{L_k} m_{\modOf{k+1}{K}}$.
    Hence from what precedes,
    \begin{inparaenum}
    \item $\pi=g_1g_2{\ldots}g_{K}g_1$ is a closed path in the intersection graph of $\f'$, and
    \item for every $k \in [1,K]$, $p_k \in g_k \inter g_{k+1}$ is correct.
    \end{inparaenum}%
    It follows that $\f'$ is correct.
    Hence, we can apply inductively \reflem{sufficiency:liveness:fix:3} on each $(m_{k},m_{k+1})$, with $k \in [1,K]$.
    It follows that $m_1 < m_1$;
    a contradiction to the definition of a log.
  \end{proof}
  
  By \refclaim{sufficiency:liveness:stable:5}, $\mathcal{C}$ is partially ordered by $<$.
  Pick $m_0$ an element with no predecessor in $\mathcal{C}$ wrt.relation $<$, and name $h=\dst(m_0)$.
  Below , we prove that $m_0$ is eventually $\phStable$ at every correct process in $h$.
  
  \begin{claim}
    \labclaim{sufficiency:liveness:stable:6:1}
    $\forall p \in h \inter \correct \sep \forall \f \in \Fa(p) \inter \correct \sep \forall h' \in \f(h) \sep \eventuallyTL((m_0,h') \in \LOG_{h})$
  \end{claim}
  
  \begin{proof}
    Since $\f$ is correct, there exists a correct process in $h \inter h'$, say $q$.
    At process $q$, \refline{sufficiency:stabilize:2} is eventually always true by \reflem{sufficiency:liveness:commit}, say at time $t_0$.
    By the definition of $m_0$, there exists a time $t_1$ after which \refline{sufficiency:stabilize:4} is also always true, at say $t_1$.
    Since $h' \in \Gr(q)$, $q$ eventually executes $\stabilize(m_0,h')$ after time $\max(t_0,t_1)$.
  \end{proof}

  \begin{claim}
    \labclaim{sufficiency:liveness:stable:6:2}
    $\forall p \in h \inter \correct \sep \eventuallyTL(\phase^{p}[m_0] \geq \phStable)$
  \end{claim}

  \begin{proof}
    According to \reflem{sufficiency:liveness:commit}, message $m_0$ is eventually $\phCommit$ at $p$.
    Observe that the case $h \notin \Fa$ is immediate.
    Indeed, once $m_0$ is committed, the precondition at \refline{sufficiency:stable:3} is vacuously true.
    Thus consider some family $\f \in \Fa(p)$ with $h \in \f$.
    By \refclaim{sufficiency:liveness:stable:6:1}, for each $h' \in \f(h)$, a tuple $(m_0,h')$ is eventually in $\LOG_h$.
    Hence the preconditions at \reflinestwo{sufficiency:stable:2}{sufficiency:stable:3} are both eventually true, leading $p$ to move $m_0$ to the $\phStable$ phase.
  \end{proof}
  
  Clearly $\pred(m_0,p)$ is empty if $p \notin h$.
  By \refclaim{sufficiency:liveness:stable:6:2}, it is also empty eventually at $p \in h$.
  Since $\mathcal{C}(m,t) = \pred^{t}(m) \inter (\phi \leq \phCommit)^{t}$, $m_0$ is eventually not in $\mathcal{C}(m,t)$;
  a contradiction to $m_0 \in \mathcal{C}$.  
  
\end{proof}

\begin{lemma}
  \lablem{sufficiency:liveness:delivered}
  $\forall p \in \correct \inter g \sep \alwaysTL(m \in \LOG_g \implies \eventuallyTL(\phase^{p}[m] = \phDelivered))$
\end{lemma}

\begin{proof}
  Let $\mathcal{S}(m,t) \equaldef \{m' <_{p} m : \phase^{p,t}[m'] \leq \phStable \}$.
  By \reflem{sufficiency:liveness:stable}, $m$ is eventually $\phStable$ at $p$, say at time $t_0$.
  Consider some $t \geq t_0$.
  By \reflem{sufficiency:ordering:1} and \refclaim{sufficiency:trivias:logs:4}, $\mathcal{S}(m,t+1) \subseteq \mathcal{S}(m,t)$.
  By \reflem{sufficiency:liveness:acyclic}, $\mathcal{S}(m,t)$ is totally ordered by $(<_p)^{*}$.
  Let $m_0$ be the smallest element in $\mathcal{S}(m,t)$.
  By \reflem{sufficiency:liveness:stable} and the definition of $m_0$, the preconditions at \reflinestwo{sufficiency:deliver:2}{sufficiency:deliver:3} are eventually true for $m_0$.
  Hence, eventually $m_0 \notin \mathcal{S}(m,t)$.
  Applying inductively the above reasoning, $(\mathcal{S}(m,t))_{t}$ converges toward an empty set.
  When this happens, the preconditions at \reflines{sufficiency:deliver:2}{sufficiency:deliver:3} are true for message $m$
  As a consequence, process $p$ eventually delivers $m$.  
\end{proof}

As a corollary of the above lemma, we obtain that:

\begin{proposition}
  \labprop{sufficiency:liveness}
  \refalg{sufficiency} ensures the termination property of atomic multicast.
\end{proposition}

\subsubsection{Establishing the result}
\labsection{sufficiency:correctness:result}

The lemma below is necessary to prove the genuineness of \refalg{sufficiency}.

\begin{proposition}
  \labprop{sufficiency:genuineness}
  Given two intersecting groups $g$ and $h$, if no message is addressed to $h$ during a run, then only the processes in $g \inter h$ take steps to implement an operation of $\LOG_{g \inter h}$
\end{proposition}

\begin{proof} 
  First, let us establish that if $p$ and $p'$ execute respectively the sequences of operations $\sigma$ and $\sigma'$ over $\LOG_{g \inter h}$, then $\sigma \prefix \sigma'$ or the converse holds.
  A call to $\LOG_{g \inter h}$ is executed at \reflinestwo{sufficiency:pending:5}{sufficiency:commit:6}.
  It requires that $m \in \LOG_g$ at that time.
  Moreover, a process may only call these two lines once and in this exact order.
  Consider that a process $q$ executes an operation $\op$ for $m$ then another operation $\op'$ for $m'$.
  For the sake of contradiction, assume $m' <_{\LOG_g} m$.
  Since $m' <_{\LOG_g} m$, when $q$ executes \refline{sufficiency:pending:5}, the precondition at \refline{sufficiency:pending:3b} must be true for $m'$.
  Thus the lines \reflinestwo{sufficiency:pending:5}{sufficiency:commit:6} were executed for $m'$ before.

  Then let us observe that, as processes in $g$ executes the operations of $\LOG_{g \inter h}$ in the exact same order then the run is contention free.
  As a consequence, only the AC objects are used to execute these operations and only the processes in $g \inter h$ take steps in the run, as required.  
\end{proof}

At the light of the preceding results, we may then conclude that the algorithm is correct.

\begin{theorem}
  \labtheo{sufficiency:main}
  \refalg{sufficiency} implements genuine group sequential atomic multicast.
\end{theorem}

\begin{proof}
  Integrity follows from the pseudo-code of \refalg{sufficiency}.
  By \refproptwo{sufficiency:ordering}{sufficiency:liveness}, the algorithm satisfies the ordering and termination properties of atomic multicast.
  \refprop{sufficiency:genuineness} ensures that the solution is genuine.
\end{proof}
\fi

\section{Necessity}
\labsection{necessity}

Consider some environment $\E$, a failure detector $D$ and an algorithm $\A$ that uses $D$ to solve atomic multicast in $\E$.
This section shows that $D$ is stronger than $\WFD$ in $\E$.
To this end, we first use the fact that atomic multicast solves consensus per group.
Hence $\WFD$ is stronger than $\land_{g \in \Gr}~(\Omega_g \land \Sigma_g)$.
\refsection{necessity:sigma} proves that $D$ is stronger than $\Sigma_{g \inter h}$ for any two groups $g, h \in \Gr$.
Further, in \refsection{necessity:gamma}, we establish that $D$ is stronger than $\gamma$.
This last result is established when $D$ is realistic.
The remaining cases are discussed in \refsection{discussion}.

\subsection{Emulating $\Sigma_{g \inter h}$}
\labsection{necessity:sigma}

Atomic multicast solves consensus in each destination group.
This permits to emulate $\land_{g \in \Gr} \Sigma_g$.
However, for two intersecting groups $g$ and $h$, $\Sigma_g \land \Sigma_h$ is not strong enough to emulate $\Sigma_{g \inter h}$.%
\footnote{
  The two detectors may return forever non-intersecting quorums.
}
Hence, we must build the failure detector directly from the communication primitive.
\refalg{sigma} presents such a construction.
This algorithm can be seen as an extension of the work of Bonnet and Raynal~\cite{kset-bonnet} to extract $\Sigma_k$ when $k$-set agreement is solvable.
\refalg{sigma} emulates $\Sigma_{\inter_{g \in G} g}$, where $G \subseteq \Gr$ is a set of at most two intersecting destination groups.

At a process $p$, \refalg{sigma} employs multiple instances of algorithm $\A$.
In detail, for every group $g \in G$ and subset $x$ of $g$, if process $p$ belongs to $x$, then $p$ executes an instance $\A_{g,x}$ (\refline{sigma:var:1}).
Variable $Q_g$ stores the responsive subsets of $g$, that is the sets $x \subseteq g$ for which $\A_{g,x}$ delivers a message.
Initially, this variable is set to $\{g\}$.

\begin{algorithm}[!t]

  \small
  \caption{Emulating $\Sigma_{\inter_{g \in G} g}$ -- code at process $p$}
  \labalg{sigma}
  
  \begin{algorithmic}[1]

    \begin{variables}
      \State $(A_{g,x})_{g \in G, x \subseteq g \sep p \in x}$ \labline{sigma:var:1} 
      \State $(Q_g)_{g \in G} \assign \lambda g.\{g\}$ \labline{sigma:var:2}
      \State $(\qr_g)_{g \in G} \assign \lambda g.g$ \labline{sigma:var:3}
    \end{variables}

    \ForAll{$g \in G, x \subseteq g : p \in x$} \labline{sigma:amcast:1}
    \State \textbf{let} $m$ such that $\dst(m) = g \land \payload(m)=p$ \labline{sigma:amcast:2}
    \State $A_{g,x}.\multicast(m)$ \labline{sigma:amcast:3}
    \EndFor

    \When{$A_{g,x}.\deliver(\any)$} \labline{sigma:del:1}
    \State $Q_g \assign Q_g \union \{ x \}$ \labline{sigma:del:2}
    \EndWhen
    
    \When{$\query$}
    \If{$p \notin \biginter_{g \in G} g$} \labline{sigma:query:1}
    \Return $\bot$ \labline{sigma:query:2}
    \EndIf
    \ForAll{$g \in G$} \labline{sigma:query:4}
    \State $\qr_g \assign~\text{\textbf{choose}}~\underset{y \in Q_g}{\argmax}~\rank(y)$ \labline{sigma:query:5}
    \EndFor
    \Return $(\bigunion_{g \in G} \qr_g) \inter (\biginter_{g \in G} g)$ \labline{sigma:query:6}
    \EndWhen        

  \end{algorithmic}

\end{algorithm}

\refalg{sigma} uses the ranking function defined in \cite{kset-bonnet}.
For some set $x \subseteq \procSet$, function $\rank(x)$ outputs the rank of $x$.
Initially, all the sets have rank $0$.
Function $\rank$ ensures a unique property: a set $x$ is correct if and only if it ranks grows forever.
To compute this function, processes keep track of each others by exchanging (asynchronously) ``alive'' message.
At a process $p$, the number of ``alive'' messages received so far from $q$ defines the rank of $q$.
The rank of a set is the lowest rank among all of its members.

At the start of \refalg{sigma}, a process atomic multicasts its identity for every instance $\A_{g,x}$ it is executing (\refline{sigma:amcast:3}).
When, $\A_{g,x}$ delivers a process identity, $x$ is added to variable $Q_g$ (\refline{sigma:del:2}).
Thus, variable $Q_g$ holds all the instances $\A_{g,x}$ that progress successfully despite that $g \setminus x$ do not participate.
From this set, \refalg{sigma} computes the most responsive quorum using the ranking function (\refline{sigma:query:5}).
As stated in \reftheo{necessity:sigma} below, these quorums must intersect at any two processes in $\inter_{g \in G} g$.

\begin{theorem}
  \labtheo{necessity:sigma}
  \refalg{sigma} implements $\Sigma_{\inter_{g \in G} g}$ in $\E$.
\end{theorem}

\iflongversion
\begin{proof}
  For starters, we establish that the range of \refalg{sigma} is valid.
  Let us observe that if $p$ does not belong to $\biginter_{g \in G} g$, then \refalg{sigma} always returns $\bot$ (\refline{sigma:query:2}).
  In the converse case, $p$ executes \refline{sigma:query:6} when it inquiries the emulated failure detector.
  Initially, variable $qr_g$ contains $g$.
  When it is changed to some set value $x$ (\refline{sigma:query:5}), $x \subseteq g$ holds with $p \in x$ (\refline{sigma:del:2}) 
  Thus, \refalg{sigma} always returns at $p$ a non-empty set in the range of $\biginter_{g} g$, as required.
  
  It remains to prove that \refalg{sigma} also ensure the two key properties of the failure detector.
  To this end, consider a run $R$ of \refalg{sigma} during which correct processes take an unbounded amount of steps $S$.
  Let $F$ and $H$ be respectively the failure pattern and the history of $D$ in $R$.

  \subparagraph*{(Intersection)}
  For the sake of contradiction, assume $p, q \in \biginter_g g$ return at \refline{sigma:query:6} two disjoint subsets $s$ and $s'$.
  Hence at process $p$, each $(\A_{g,qr_g})_{g}$ delivered some identity in $R$ and we have $s = (\bigunion_g qr_g) \inter (\biginter_g g)$.
  The same observation holds at process $q$ for some family $(qr'_g)_{g}$.  
  
  We claim that there exist two disjoint subsets $x$ and $y$ respectively in $(qr_g)_{g}$ and $(qr'_g)_{g}$.
  This is trivial if $G$ contains a single group, with $s=x$ and $s'=y$.
  Otherwise, $s=(qr_g \union qr_h) \inter (g \inter h)$ and $s'=(qr_g' \union qr_h') \inter (g \inter h)$, with $g \neq h$.
  Hence, we have immediately that $qr_g$ and $qr'_h$ are disjoint.
  
  Let $\hat{p}$ and $\hat{q}$ be the two identities returned respectively in instances $\A_{g,x}$ and $\A_{h,y}$ during run $R$.    
  Applying \reflem{model:1}, there exists a run $R_x$ of algorithm $\A$ with steps $S|\A_{g,x}$ and with history $D$.
  Similarly, there exists a run $R_y$ of $\A$ with steps $S|\A_{h,y}$ and the same history.
  Since $x$ and $y$ are disjoint, \reflem{model:3} tells us that there exists a run of $\A$ with steps $(S|\A_{g,x}) \union (S|\A_{h,y})$.
  In this run, processes $p$ and $q$ deliver respectively identities $\hat{p}$ and $\hat{q}$ first.  
  A contradiction to the ordering property of atomic multicast.

  \subparagraph*{(Liveness)}
  Consider some correct process $p \in \biginter_g g$.
  We show that for each $g \in G$, eventually $\qr_{g}$ contains only correct processes at process $p$.
  Let $\mathcal{C}_g$ be $g \inter \correct(F)$.
  Applying \reflem{model:1}, there exists a run $\hat{R}$ of $\A_{g,\mathcal{C}_g}$ with history $F$ and steps $S|\A_{g,\mathcal{C}_g}$.
  To the correct processes in $\hat{R}$, this run is indistinguishable from a run $\hat{R'}$ in which the faulty processes in $g$ take no steps before they fail.
  As $\A$ solves genuine atomic multicast and there is no message addressed to groups outside $g$, every process in $\mathcal{C}_g$ delivers all the identities of $\mathcal{C}_g$ in $\hat{R'}$.
  Thus, this also happens in $R$.
  Consider a point in time from which $\mathcal{C}_g$ is in $Q_g$ at process $p$.
  Assume that $\qr_g$ takes infinitely often some value $x$ in $R$.
  It follows that quorum $x$ has infinitely often a higher rank than $\mathcal{C}_g$.
  From the definition of the ranking function, $x$ is a set of correct processes.
  Thus, for each $g \in G$, there exists a point in time $t_g$ after which $\qr_{g}$ contains only correct processes at process $p$.
  After time $t=\max_{g \in G} t_g$, only correct processes are returned at \refline{sigma:query:6}, as required.
\end{proof}
\fi

\subsection{Emulating $\gamma$}
\labsection{necessity:gamma}



\subparagraph*{Target systems}
A process $p$ is failure-prone in environment $\E$ when for some failure pattern $F \in \E$, $p \in \faulty(F)$.
By extension, we say that $P \subseteq \procSet$ is failure-prone when for some $F \in \E$, $P \subseteq \faulty(F)$.
A cyclic family $\f$ is failure-prone when one of its group intersections is failure-prone.
Below, we consider that $\E$ satisfies that
\begin{inparaenumorig}[]
\item if a process may fail, it may fail at any time
  (formally,
  $
  \forall F \in \E \sep
  \forall p \in \faulty(F) \sep
  \exists F' \in \E \sep
  \forall t \in \naturalSet \sep
  \forall t' < t \sep
  F'(t') = F(t') \land F'(t) = F(t) \union \{p\}
  $).
\end{inparaenumorig}%
We also restrict our attention to realistic failure detectors, that is they cannot guess the future \cite{realistic}.%
\iflongversion
\footnote{
  The definition is recalled in \refappendix{model}.
}
\fi

\subparagraph*{Additional notions}
Consider a cyclic family $\f$.
Two closed paths $\pi$ and $\pi'$ in $\cpaths(\f)$ are equivalent, written $\pi \equiv \pi'$, when they visit the same edges in the intersection graph.
A closed path $\pi$ in $\cpaths(\f)$ is oriented.
The direction of $\pi$ is given by $\dir(\pi)$.
It equals $1$ when the path is clockwise, and $-1$ otherwise (for some canonical representation of the intersection graph).
To illustrate these notions, consider family $\f$ in \reffigure{family:ig:1}.
The sequence $\pi=g_3g_1g_2g_3$ is a closed path in its intersection graph, with $\cardinalOf{\pi}=4$ and $\pi[0]=\pi[\cardinalOf{\pi}-1]=g_3$.
The direction of this path is $1$ since it is visiting clockwise the intersection graph of $\f$ in the figure.
Path $\pi$ is equivalent to the path $\pi'=g_1g_3g_2g_1$ which visits $\f$ in the converse direction.

\subparagraph*{Construction}
We emulate failure detector $\gamma$ in \refalg{gamma}.
For each closed path $\pi \in \cpaths(\f)$ with $\pi[0] \inter \pi[1]$ failure-prone in $\E$, \refalg{gamma} maintains two variables:
\begin{inparaenumorig}[]
\item an instance $\A_{\pi}$ of the multicast algorithm $\A$, and
\item a flag $\failed[\pi]$.
\end{inparaenumorig}%
Variable $\A_{\pi}$ is used to detect when a group intersection visited by $\pi$ is faulty.
It this happens, the flag $\failed[\pi]$ is raised.
When for every path $\pi \in \cpaths(\f)$, some path equivalent to $\pi$ is faulty, \refalg{gamma} ceases returning the family $\f$ (\refline{gamma:query:2}).

In \refalg{gamma}, for every path $\pi \in \cpaths(\f)$, the processes in $\pi[0] \inter \pi[1]$ multicast their identities to $\pi[0]$ using instance $\A_{\pi}$ (\reflinestwo{gamma:amcast:1}{gamma:amcast:2}).
In this instance of $\A$, all the processes in $\f$ but the intersection $\pi[0] \inter \pi[\cardinalOf{\pi}-2]$ participate (\refline{gamma:var:1}).
As the path is closed, this corresponds to the intersection with the last group preceding the first group in the path.

When $p \in \pi[i] \inter \pi[i+1]$ delivers a message $(\any,i)$, it signals this information to the other members of the family (\refline{gamma:signal:4}).
Then, $p$ multicasts its identity to $\pi[i+1]$ (\refline{gamma:signal:5}).
This mechanism is repeated until the antepenultimate group in the path is reached (\refline{gamma:signal:3}).
When such a situation occurs, the flag $\failed[\pi]$ is raised (\refline{gamma:update:2}).
This might also happen earlier when a message is received for some path $\pi'$ equivalent to $\pi$ and visiting $\f$ in the converse direction (\refline{gamma:update:3}).

Below, we claim that \refalg{gamma} is a correct emulation of failure detector $\gamma$.

\begin{algorithm}[!t]

  \small
  \caption{Emulating $\gamma$ -- code at process $p$}
  \labalg{gamma}

  \begin{algorithmic}[1]

    \begin{variables}
      \State $(\A_{\pi})_{\pi}$ \Comment{$\forall \f \in \Fa(p) \sep \forall \pi \in \cpaths(\f) \sep p \notin \pi[0] \inter \pi[\cardinalOf{\pi}-2]$} \labline{gamma:var:1}
      \State $\failed[\pi] \assign \lambda \pi. \false$ \labline{gamma:var:2}
    \end{variables}

    \vspace{.5em}
    
    \ForAll{$\A_{\pi} : p \in \pi[0] \inter \pi[1]$} \labline{gamma:amcast:1}
    \State $\A_{\pi}.\multicast(p,0)$ to $\pi[0]$ \labline{gamma:amcast:2}
    \EndFor

    \vspace{.5em}

    \begin{action}{$\signal(\pi,i)$} \labline{gamma:signal:1}
      \Precondition $\A_{\pi}.\deliver(\any,i)$ \labline{gamma:signal:2}
      \Precondition $i < \cardinalOf{\pi}-2 \land p \in \pi[i+1]$ \labline{gamma:signal:3}
      \Effect $\send(\pi,i)$ to $\f$ \labline{gamma:signal:4}
      \Effect $\A_{\pi}.\multicast(p,i+1)$ to $\pi[i+1]$  \labline{gamma:signal:5}
    \end{action}

    \vspace{.5em}

    \begin{action}{$\update(\pi)$} \labline{gamma:update:1}
      \Precondition $\exists \pi' \equiv \pi \sep rcv(\pi,j) \land \lor~ j = \cardinalOf{\pi}-3$ \labline{gamma:update:2}
      \Precondition \hspace{8.5em} $\lor~ (\rcv(\pi',0) \land \pi[j] = \pi'[0] \land \dir(\pi) = -\dir(\pi'))$ \labline{gamma:update:3}
      \Effect $\failed[\pi] \assign \true$ \labline{gamma:update:4}
    \end{action}
    
    \vspace{.5em}

    \When{$\query$} \labline{gamma:query:1}
    \Return $\{ \f \in \Fa(p) : \exists \pi \in \cpaths(\f) \sep \forall \pi' \equiv \pi \sep \failed[\pi'] = \false \}$ \labline{gamma:query:2}
    \EndWhen      
    
  \end{algorithmic}

\end{algorithm}

\begin{theorem}
  \labtheo{necessity:gamma}
  \refalg{gamma} implements $\gamma$ in $\E$.
\end{theorem}

\iflongversion
\begin{proof}
  \input{figures-gamma}
  For starters, we observe that the range of \refalg{gamma} is correct, that is it always returns a subset of the cyclic families to which belongs the local process.
  Then, consider some failure-prone cyclic family $\f=(g_i)_{i \in [1,K>2]}$.

  \subparagraph*{(Completeness)}
  Assume that family $\f$ fails in a run $\R$ of \refalg{gamma}.
  If all the group intersections in $\f$ are faulty, then we are done.
  Otherwise, for every closed path $\pi \in \cpaths(\f)$, some group intersection $g \inter h$ visited by $\pi$ is faulty.
  Consider such a path $\pi$, and without lack of generality assume that
  \begin{inparaenum}
  \item $\pi$ visits $\f$ starting from $g_{1}$ in the direction of $g_2$, that is $\pi=g_1g_2{\ldots}g_Kg_1$,
  \item $g_{K} \inter g_{1}$ is faulty, and 
  \item some process $p \in g_{1} \inter g_{2}$ is correct.
  \end{inparaenum}%
  This is illustrated at the top of \reffigure{gamma}.

  Process $p$ executes \refline{gamma:amcast:2} for path $\pi$.
  The steps of $\A_{\pi}$ in $R$ defines a run $R'$ of $\A$ (\reflem{model:1} in \refappendix{model}).
  To processes in $\procSet \setminus (g_{K} \inter g_{1})$, the run $R'$ is indistinguishable from a run $R''$ in which the processes in $g_{K} \inter g_{1}$ takes no step before they crash.
  As $p$ is correct, a single message is addressed to $\pi[0]$ and $\A$ solves genuine atomic multicast, $p$ must deliver its identity in $R''$.
  Hence, $\A_{\pi}$ eventually delivers some message $(\any,0)$ at $p$ in $R$ at \refline{gamma:signal:2}.

  As detailed below, there are then two cases to consider.
  Both lead to the fact that eventually every correct process $q$ with $\f \in \Fa(q)$ executes $\update(\pi)$, setting $\failed[\pi]$ to $\true$.
  Since this holds, for every path in $\cpaths(\f)$, $\f$ is eventually excluded from the response value when the failure detector is queried at \refline{gamma:query:1}.
  \begin{itemize}
  \item
    For starters, consider that none of the group intersections except $g_K \inter g_1$ is faulty.
    From what precedes, $p$ executes \reflines{gamma:signal:4}{gamma:signal:5} after delivering $(\any,0)$ in instance $A_{\pi}$.
    By a short induction, eventually some correct process in $g_{K-1} \inter g_{K}$ sends a message $(\pi,K-2)$ to $\f$ at \refline{gamma:signal:4}.
    All the paths $\pi \in \cpaths(\f)$ have the same length, namely $K+1$.
    Hence, the precondition at \refline{gamma:update:2} is eventually true at all the correct processes in $\f$.
  \item
    Otherwise, there exists $j \in [0,K-2]$ such that
    \begin{inparaenumorig}[]
    \item $\pi[j] \inter \pi[j+1]$ is faulty, and
    \item for every $0 \leq k < j$, $\pi[k] \inter \pi[k+1]$ contains (at least) one correct process.
    \end{inparaenumorig}%
    Choose a correct process $q \in \pi[j-1] \inter \pi[j]$.
    By a short induction, $q$ eventually delivers a message $(\any,j-1)$ in $A_{\pi}$ at \refline{gamma:signal:2}.
    As $j-1 < K-1$ and $q \in \pi[j]$, this process sends a message $(\pi,j-1)$ to $\f$ at \refline{gamma:signal:4}.
    Let $\pi'$ be the path starting from $\pi[j]$ in the converse direction than $\pi$.
    Since $\pi[j] \inter \pi[j+1]$ is faulty, such an instance is executed in \refalg{gamma} at process $q$.
    By applying a similar reasoning to the one we conducted for $p$, eventually $\A_{\pi'}$ returns at that process.
    It follows that $q$ sends a message $(\pi',0)$ in $R$.
    As $\pi'[0]=\pi[j-1]$, the precondition at \refline{gamma:update:3} is eventually true at every correct process in $\f$.
  \end{itemize}

  \subparagraph*{(Accuracy)}
  The proof is by contradiction.
  It is illustrated at the bottom of \reffigure{gamma}.
  Consider a run $R$ of \refalg{gamma} such that
  \begin{inparaenum}
  \item family $\f$ is correct in $R$, and
  \item for some process $p$, with $\f \in \Fa(p)$, $p$ does not return $\f$ in $R$ when querying its failure detector at \refline{gamma:query:1}.
  \end{inparaenum}

  At the light of the pseudo-code of \refalg{gamma}, for all paths $\pi'$ in $\cpaths(\f)$, $p$ executes $\update(\pi)$ with $\pi \equiv \pi'$ in $R$ before calling the failure detector.
  Consider such a path $\pi$.
  According to the preconditions at \reflines{gamma:update:2}{gamma:update:3}, there exists some process $q \in \pi[0] \inter \pi[1]$ such that $q$ sends a message $(\pi,0)$ in $R$.
  For this to happen, process $q$ should have delivered a message $m=(\any,0)$ in $\A_{\pi}$ before (\refline{gamma:signal:2}).
  Let $t$ be the time at which the delivery of message $m$ happens first in the run $R$.
  
  In run $R$, until time $t$, only the processes in $\pi[0] \inter \pi[1]$ make steps.
  This comes from the fact that $\A$ is genuine and only $m$ with $\dst(m)=\pi[0]$ is multicast in the run so far.
  Furthermore, the processes in $\pi[0] \inter \pi[K-1]$ do not participate in $\A_{\pi}$ (see \refline{gamma:var:1}).
  
  The steps of $\A_{\pi}$ in $R$ until time $t$ defines a run $R'=(F,H,I,S,T)$ of $\A$, where $F$ is the failure pattern of $R$ and $H \in D(F)$.
  As $\pi \in \cpaths(\f)$, then $\pi[0] \inter \pi[1]$ is failure-prone.
  In environment $\E$, a failure-prone process may fail at any time.
  Thus, there exists $F' \in \E$ such that 
  \begin{inparaenum}
  \item $\forall t'\leq t \sep F(t')=F(t)$, and
  \item $\forall t'> t \sep F(t')=F(t) \union (\pi[0] \inter \pi[1])$.
  \end{inparaenum}%
  Because $D$ is realistic, there exists $H' \in D(F')$ such that $H'$ is identical to $H$ until time $t$.
  As no process takes a step after time $t$ in $R'$, $\hat{R}=(F',H',I,S,T)$ is a run of $\A$.  

  Now consider a (fair) run $\hat{R'}$ of $\A$ with the exact same failure pattern and failure history as in $\hat{R}$.
  In this run, processes take no step until time $t$.
  At some later time, a process in $\pi[1] \inter \pi[2]$ multicasts its identity in a message $m_1$ to $\pi[1]$.
  Then, for every $k \in [1,K-2]$, if $p \in \pi[k] \inter \pi[k+1]$ delivers a message $m_k$, then $p$ multicasts its identity in a message $m_{k+1}$ to $\pi[k+1]$.
  Since $\A$ solves atomic multicast in $\E$ and $\f \setminus (\pi[K] \inter \pi[1]) \subseteq \correct(F')$, for every $k \in [1,K-1]$, some correct process $p_k$ in $\pi[k] \inter \pi[k+1]$ delivers message $m_{k \in [1,K-1]}$ in run $\hat{R'}$.
  This happens before the delivery (if any) of $m_{k+1}$.
  As a consequence, we have during this run the following delivery relation: $m_1 \delOrder m_2 \delOrder \ldots \delOrder m_{K-1}$.
  
  Let $S'$ be the schedule of the processes in $\hat{R'}$ until the moment in time where message $m_{K-1}$ is delivered in $\pi[K] \inter \pi[K-1]$.
  The steps in $S$ and $S'$ are from different processes.
  Furthermore by construction the first step in $S'$ is taken in time after the last step in $S$.
  Applying \reflem{model:4}, there exists some $T'$ such that $\overline{\R}=(F',H',\any,S \concat S',T')$ is a run of $\A$.
  In this run, relation $\delOrder$ over $\{m,m_1, \ldots, m_{K-1}\}$ is cyclic;
  contradiction.
  %
  %
\end{proof}
\fi

\section{Variations}
\labsection{variations}

This section explores two common variations of the atomic multicast problem.
It shows that each variation has a weakest failure detector stronger than $\WFD$.
The first variation requires messages to be ordered according to real time.
This means that if $m$ is delivered before $m'$ is multicast, no process may deliver $m'$ before $m$.
In this case, we establish that the weakest failure detector must accurately detect the failure of a group intersection.
The second variation demands each group to progress independently in the delivery of the messages.
This property strengthens minimality because in a genuine solution a process may help others as soon as it has delivered a message.
We show that the weakest failure detector for this variation permits to elect a leader in each group intersection.

\subsection{Enforcing real-time order}
\labsection{variations:rt}

Ordering primitives like atomic broadcast are widely used to construct dependable services \cite{ChandraGR07}.
The classical approach is to follow state-machine replication (SMR), a form of universal construction.
In SMR, a service is defined by a deterministic state machine, and each replica maintains its own local copy of the machine.
Commands accessing the service are funneled through the ordering primitive before being applied at each replica on the local copy.

SMR protocols must satisfy linearizability \cite{loo:syn:1468}.
However, as observed in \cite{BezerraPR14}, the common definition of atomic multicast is not strong enough for this:
if some command $d$ is submitted after a command $c$ get delivered, atomic multicast does not enforce $c$ to be delivered before $d$, breaking linearizability.
To sidestep this problem, a stricter variation must be used.
Below, we define such a variation and characterize its weakest failure detector.

\subsubsection{Definition}
\labsection{variations:rt:definition}

We write $m \rt m'$ when $m$ is delivered in real-time before $m'$ is multicast.
Atomic multicast is \emph{strict} when ordering is replaced with:
\begin{inparaenumorig}
\item[(\emph{Strict Ordering})]
  The transitive closure of $(\delOrder \union \rt)$ is a strict partial order over $\msgSet$.
\end{inparaenumorig}%
Strictness is free when there is a single destination group.
Indeed, if $p$ delivers $m$ before $q$ broadcasts $m'$, then necessarily $m \delOrderOf{p} m'$.
This explains why atomic broadcast does not mention such a requirement.
In what follows, we prove that strict atomic multicast is harder than (vanilla) atomic multicast.

\subsubsection{Weakest failure detector}
\labsection{variations:rt:wfd}


\subparagraph*{Candidate}
For some (non-empty) group of processes $P$, the \emph{indicator failure detector} $1^{P}$ indicates if all the processes in $P$ are faulty or not.
In detail, this failure detector returns a boolean which ensures that:

\begin{itemize}
\item[(\emph{Accuracy})] $\forall p \in \procSet \sep \forall t \in \naturalSet \sep 1^{P}(p,t) \implies P \subseteq F(t)$
\item[(\emph{Completeness})] $\forall p \in \correct \sep \forall t \in \naturalSet \sep P \subseteq F(t) \implies \exists \tau \in \naturalSet \sep \forall t' \geq \tau \sep 1^{P}(p,t')$
\end{itemize}

For simplicity, we write $1^{g \inter h}$ the indicator failure detector restricted to the processes in $g \union h$ (that is, the failure detector $1^{g \inter h}_{g \union h}$).
This failure detector informs the processes outside $g \inter h$ when the intersection is faulty.
Notice that for the processes in the intersection, $1^{g \inter h}$ does not provide any useful information.
This comes from the fact that simply returning always $\true$ is a valid implementation at these processes.

Our candidate failure detector is $\WFD \land (\land_{g,h \in \Gr}~1^{g \inter h})$.
One can establish that $\land_{g,h \in \Gr} 1^{g \inter h}$ is stronger than $\gamma$ (see \refprop{variations:indicator} below).
As a consequence, this failure detector can be rewritten as $(\land_{g,h \in \Gr}~\Sigma_{g \inter h} \land 1^{g \inter h}) \land (\land_{g \in \Gr}~\Omega_{g})$.

\begin{proposition}
  \labprop{variations:indicator}
  $\land_{g,h \in \Gr}~1^{g \inter h} \leq \gamma$
\end{proposition}

\iflongversion
\begin{proof}
  We construct failure detector $\gamma$ as follows:
  Consider some cyclic family $\f \in \Fa$.
  For each path $\pi \in \cpaths(\f)$, when $p \in g \union h$ with $g,h \in \f$ triggers $1^{g \inter h}$, $p$ sends a message to the rest of the family.
  Once such a message is received for each class of equivalent paths in $\cpaths(\f)$, $\f$ is declared faulty.
  At the light of the accuracy and completeness properties of an indicator failure detector, this implementation is correct.
\end{proof}
\fi

\iflongversion
The comparison is strict in case there are at least two intersecting groups.
This comes from the fact that, when $\Fa=\emptySet$, or the cyclic family is initially faulty, $\gamma$ does not provide any useful information.

\begin{corollary}
  \labcor{variations:indicator}
  $(\exists g,h \in \Gr \sep g \inter h \neq \emptySet) \implies \gamma \not \leq \land_{g,h \in \Gr}~1^{g \inter h}$
\end{corollary}

\begin{proof}
  (By contradiction.)
  Assume that $\A$ emulates $\land_{g,h \in \Gr}~1^{g \inter h}$ from $\gamma$.
  Consider two (distinct) intersecting groups $g$ and $h$ in $\Gr$.
  Pick some process $p \in g \xor h$.
  \emph{(Case $\exists \f \in \Fa \sep g,h \in \f$).}
  Necessarily $\cardinalOf{\f} \geq 3$.
  Let $h'$ be a group in $\f$ distinct from $g$ and $h$.
  Consider a run in which $h'$ is initially faulty and both $g$ and $h$ are correct.
  Name $F$ its failure pattern, $H_{\gamma}$ the history of $\gamma$ in this run and $H_{\A}$ the emulated history.
  Since $g \inter h$ is correct, by the accuracy property of $1^{g \inter h}$, $H_{\A}(p,t)=0$ at all time $t$.
  Consider now the failure pattern $F'$ identical to $F$, except that $g \inter h$ is faulty at the start.
  Then, since $\f$ is always faulty, $H_{\gamma}$ is also a valid history for $F'$.
  Hence, when running $\A$ with failure pattern $F'$, $H_{\A}(p,t)=0$ at all time $t$ is a valid output.
  This contradicts the completeness property of $1^{g \inter h}$.
  \emph{(Otherwise).}
  The proof is similar to the case above.
  Namely, $\A$ cannot distinguish at any process, including $p$, a run in which $g \inter h$ is correct, from a run in which the group intersection is faulty.
\end{proof}
\fi

\subparagraph*{Necessity}
An algorithm to construct $1^{g \inter h}$ is presented in \refalg{indicator}.
It relies on an implementation $\A$ of strict atomic multicast that makes use internally of some failure detector $D$.
\refprop{variations:necessity} establishes the correctness of such a construction.

\begin{algorithm}[!t]

  \small
  \caption{Emulating $1^{g \inter h}$ -- code at process $p \in g \union h$}
  \labalg{indicator}

  \begin{algorithmic}[1]

    \begin{variables}
      \State $B \assign$ \textbf{if} $(p \in g \setminus h)$ \textbf{then} $\A_g$ \textbf{else if} $(p \in h \setminus g)$ \textbf{then} $\A_h$ \textbf{else} $\bot$ \labline{indicator:var:1} \Comment{$\A_g$ and $\A_h$ are distinct instances of $\A$}
      \State $\failed \assign \false$ \labline{indicator:var:2}
    \end{variables}

    \vspace{.5em}

    \If{$B \neq \bot$}
    \State $B.\multicast(p)$ \labline{indicator:1}
    \State \textbf{wait until} $B.\deliver(\any)$ \labline{indicator:2}
    \State $\send(\flagFailed)$ to $g \union h$ \labline{indicator:3}
    \EndIf

    \vspace{.5em}
    
    \When{$\rcv(\flagFailed)$} \labline{indicator:4}
    \State $\failed \assign \true$ \labline{indicator:5}
    \EndWhen

    \vspace{.5em}

    \When{$\query$} \labline{indicator:6}
    \Return $\failed$ \labline{indicator:7}
    \EndWhen      
    
  \end{algorithmic}

\end{algorithm}

\begin{proposition}
  \labprop{variations:necessity}
  \refalg{indicator} implements $1^{g \inter h}$.
\end{proposition}

\iflongversion
\begin{proof}
  We show successively that the two properties hold.
  
  \subparagraph*{(Accuracy)}
  Consider that a call to $\query$ returns $\true$ in a run $R=(F,H,\any,\any,\any)$ of \refalg{indicator}.
  For the sake of contradiction, assume that $p \in g \inter h$ is correct in $R$.
  In this run, some process $q$ delivers a message $m$ at \refline{indicator:2} using variable $B$.
  Without lack of generality, consider that at process $q$, variable $B$ is set to instance $\A_g$ (at \refline{indicator:var:1}).
  Let $R_g=(F,H,\any,S,T)$ be the sub-run of $\A_g$ in $R$ (\reflem{model:1} in \refappendix{model}).
  In $R_g$, only the processes in $g \setminus h$ make steps.
  Name $S' \prefix S$ the sequence of their steps until $m$ is delivered at process $q$.
  Then, for some appropriate timing $T' \prefix T$, $R'=(F,H,\any,S',T')$ is a run of $\A$.
  Let $t$ be the moment at which the last step of $S'$ occurs in $R'$.
  Run $R'$ is indistinguishable to $g \inter h$ until time $t$ from a run $\hat{R}'=(F,H,\any,\hat{S},\any)$ in which no message is multicast.
  Consider the continuation $R''$ of $\hat{R}'$ in which $p$ multicasts a message $m'$ to $h$ at time $t+1$.
  As process $p$ is correct, it eventually delivers $m'$ in $R''$.
  Since $\A$ is genuine, only the processes in $h$ take steps in $R''$.
  Let $S''$ be the sequence of steps of the destination group $h$ until $p$ delivers message $m'$.
  By \reflem{model:4}, for some appropriate timing $T''$, the run $(F,H,\any,S' \concat S'',T'')$ is a run of $\A$.
  In this run, $m \rt m'$ and $m' \delOrder_{p} m$;
  contradiction.
  
  \subparagraph*{(Completeness)}
  Consider a run of \refalg{indicator} during which $g \inter h$ is faulty.
  Then, this run is not distinguishable to the processes in $g \xor h$ from a run where $g \inter h$ is initially faulty.
  As a consequence, both $\A_g$ and $\A_h$ must deliver a message at \refline{indicator:2}.
  It follows that all the correct processes in $g \xor h$ eventually set variable $\failed$ to $\true$ at \refline{indicator:5}.
\end{proof}
\fi

\subparagraph*{Sufficiency}
The solution to strict atomic multicast is almost identical to \refalg{sufficiency}.
The only difference is at \refline{sufficiency:stable:3} when a message moves to the $\phStable$ phase.
Here, for every destination group $h$ with $h \inter g \neq \emptySet$, a process waits either that 
\begin{inparaenumorig}[]
\item $1^{g \inter h}$ returns $\true$, or
\item that a tuple $(m,h)$ appears in $\LOG_g$.
\end{inparaenumorig}%
From \refprop{variations:indicator}, we know that the indicator failure detector $1^{g \inter h}$ provides a better information than $\gamma$ regarding the correctness of $g \inter h$.
As a consequence, the modified algorithm solves (group sequential) atomic multicast.

Now, to see why such a solution is strict, consider two messages $m$ and $m'$ that are delivered in a run, with $g=dst(m)$ and $h=\dst(m')$.
We observe that when $m' \rt m$ or $m' \delOrder m$, $m'$ is stable before $m${\iflongversion~(i.e., at every time $t$, $m \in \phStable^{t} \implies m' \in \phStable^{t}$)\fi}, from which we deduce that strict ordering holds.

With more details, in the former case ($m' \rt m$), this comes from the fact that to be delivered a message must be $\phStable$ first (\refline{sufficiency:deliver:2}).
In the later ($m' \delOrder m$), when message $m$ is $\phStable$ at some process $p$, $p$ must wait a message $(m,h)$ in $\LOG_g$, or that $1^{g \inter h}$ returns $\true$.
If $(m,h)$ is in $\LOG_g$, then \refline{sufficiency:stabilize:5} was called before by some process $q$.
\iflongversion
Because both messages are delivered and $m' \delOrder m$, $m'$ must precede $m$ in $\LOG_{g \inter h}$ (\reflem{sufficiency:ordering:3} in \refsection{sufficiency}).
\else
Because both messages are delivered and $m' \delOrder m$, $m'$ must precedes $m$ in $\LOG_{g \inter h}$.
\fi
Thus the precondition at \refline{sufficiency:stabilize:4} enforces that $m'$ is $\phStable$ at $q$, as required.
Now, if the indicator returns $\true$ at $p$, $m' \delOrder m$ tells us that a process delivers $m'$ before $m$ and this must happen before $g \inter h$ fails.

\subsection{Improving parallelism}
\labsection{variations:strong}


As motivated in the Introduction, genuine solutions to atomic multicast are appealing from a performance perspective.
Indeed, if messages are addressed to disjoint destination groups in a run, they are processed in parallel by such groups.
However, when contention occurs, a message may wait for a chain of messages to be delivered first.
This chain can span outside of the destination group, creating a delay that harms performance and reduces parallelism~\cite{tempo,convoy}.
In this section, we explore a stronger form of genuineness, where groups are able to deliver messages independently.
We prove that, similarly to the strict variation, this requirement demands more synchrony than $\WFD$ from the underlying system.

\subsubsection{Definition}
\labsection{variations:strong:definition}

As standard, a run $\R$ is fair for some correct process $p$ when $p$ executes an unbounded amount of steps in $R$.
By extension, $\R$ is fair for $P \subseteq \correct(R)$, or for short $P$-fair, when it is fair for every $p$ in $P$.
If $P$ is exactly the set of correct processes, we simply say that $\R$ is fair.

\begin{itemize}
\item[(\emph{Group Parallelism})]
  Consider a message $m$ and a run $\R$.
  Note $P=\correct(R) \inter \dst(m)$.
  If $m$ is delivered by a process, or atomic multicast by a correct process in $R$, and $R$ is $P$-fair, then every process in $P$ delivers $m$ in $R$.
\end{itemize}

Group parallelism bears similarity with $x$-obstruction freedom \cite{Taubenfeld17}, in the sense that the system must progress when a small enough group of processes is isolated.
A protocol is said \emph{strongly genuine} when it satisfy both the minimality and the group parallelism properties.

\subsubsection{About the weakest failure detector}
\labsection{variations:strong:wfd}

Below, we establish that $(\land_{g, h \in \Gr}~\Omega_{g \inter h})$ is necessary.
It follows that the weakest failure detector for this variation is at least $\WFD \land (\land_{g, h \in \Gr}~\Omega_{g \inter h})$.

\subparagraph*{Emulating $\land_{g,h \in \Gr}~\Omega_{g \inter h}$}
Consider some algorithm $\A$ that solves strongly genuine atomic multicast with failure detector $D$.
Using both $\A$ and $D$, each process may emulate $\Omega_{g \inter h}$, for some intersecting groups $g,h \in \Gr$.
The emulation follows the general schema of CHT \cite{omega}.
We sketch the key steps below.
\iflongversion
The full proof appears in \refappendix{omega}.
\else
The full proof appears in \cite{longversion}.
\fi

Each process constructs a directed acyclic graph $G$ by sampling the failure detector $D$ and exchanging these samples with other processes.
A path $\pi$ in $G$ induces multiple runs of $\A$ that each process locally simulates.
A run starts from some initial configuration.
In our context, the configurations $\I=\{I_1,\ldots,I_{n \geq 2}\}$ of interest satisfy
\begin{inparaenum}
\item the processes outside $g \inter h$ do not atomic multicast any message, and
\item the processes in $g \inter h$ multicast a single message to either $g$ or $h$.
\end{inparaenum}%
For some configuration $I_i \in \I$, the schedules of the simulated runs starting from $I_i$ are stored in a simulation tree $\Upsilon_i$.
There exists an edge $(S,S')$ when starting from configuration $S(I_i)$, one may apply a step $s=(p,m,d)$ for some process $p$, message $m$ transiting in $S(I_i)$ and sample $d$ of $D$ such that $S'=S \concat s$.

Every time new samples are received, the forest of the simulation trees $(\Upsilon_i)_i$ is updated.
At each such iteration, the schedules in $\Upsilon_i$ are tagged using the following valency function:
$S$ is tagged with $g$ (respectively, $h$) if for some successor $S'$ of $S$ in $\Upsilon_i$ a process in $g \inter h$ delivers first a message addressed to $g$ (resp. to $h$) in configuration $S'(I_i)$.
A tagged schedule is \emph{univalent} when it has a single tag, and \emph{bivalent} otherwise.

As the run progresses, each root of a simulation tree has eventually a stable set of tags.
If the root of $\Upsilon_i$ is $g$-valent, the root of $\Upsilon_j$ is $h$-valent and they are adjacent, i.e., all the processes but some $p \in g \inter h$ are in the same state in $I_i$ and $I_j$, then $p$ must be correct.
Otherwise, there exists a bivalent root of some tree $\Upsilon_i$ such that for $g$ (respectively, $h$) a correct process multicasts a message to $g$ (resp., $h$) in $I_i$.
In this case, similarly to \cite{omega}, there exists a decision gadget in the simulation tree $\Upsilon_i$.
This gadget is a sub-tree of the form $(S,S',S'')$, with $S$ bivalent, and $S'$ $g$-valent and $S''$ $h$-valent (or vice-versa).
Using the group parallelism property of $A$, we may then show that necessarily the deciding process in the gadget, that is the process taking a step toward either $S'$ or $S''$ is correct and belongs to the intersection $g \inter h$.

\subparagraph*{Solution when $\Fa = \emptySet$}
In this case, \refalg{sufficiency} just works.
To attain strong genuineness, each log object $\LOG_{g \inter h}$ is implemented with $\Sigma_{g \inter h} \land \Omega_{g \inter h}$ through standard universal construction mechanisms.
When $\Fa=\emptySet$, $\WFD \land (\land_{g, h \in \Gr}~\Omega_{g \inter h})$ is thus the weakest failure detector.
The case $\Fa \neq \emptySet$ is discussed in the next section.

\section{Discussion}
\labsection{discussion}



Several definitions for atomic multicast appear in literature (see, e.g., \cite{survey,Hadzilacos94amodular} for a survey).
Some papers consider a variation of atomic multicast in which the ordering property is replaced with:
\begin{inparaenumorig}
\item[(\emph{Pairwise Ordering})]
  If $p$ delivers $m$ then $m'$, every process $q$ that delivers $m'$ has delivered $m$ before.
\end{inparaenumorig}%
Under this definition, cycles in the delivery relation ($\delOrder$) across more than two groups are not taken into account.
This is computably equivalent to $\Fa=\emptySet$.
Hence the weakest failure detector for this variation is $(\land_{g,h \in \Gr}~\Sigma_{g \inter h}) \land (\land_{g \in \Gr}~\Omega_g)$.

In \cite{gamcast}, the authors show that failure detectors of the class $\mathcal{U}_2$ are too weak to solve the pairwise ordering variation.
These detectors can be wrong about (at least) two processes.
In detail, the class $\mathcal{U}_k$ are all the failure detectors $D$ that are $k$-unreliable, that is they cannot distinguish any pair of failure patterns $F$ and $F'$, as long as the faulty processes in $F$ and $F'$ are members of a subset $W$ of size $k$ (the ``wrong'' subset).
\iflongversion
Formally,
\begin{itemize}
\item[($k$-unreliable failure detection \cite{gamcast})]
  A failure detector $D$ is $k$-unreliable in environment $\E$ when
  for every failure pattern $F \in \E$ with $\cardinalOf{\faulty(F)} \leq k$, for every history $H \in D(F)$,
  there exists a subset $W \subseteq \procSet$ with $\cardinalOf{W}=k$ and $\faulty(F) \subseteq W$,
  such that
  for every failure pattern $F'$ with $\faulty(F') \subseteq W$,
  for every time $t_0 \in \naturalSet$,
  there is a history $H' \in D(F')$ satisfying
  $\forall p \in \procSet \sep \forall t \leq t_0 \sep H'(p,t) = H (p,t)$.%
\end{itemize}
\fi
The result in \cite{gamcast} is a corner case of the necessity of $\Sigma_{g \inter h}$ when $g \inter h = \{p,q\}$ and both processes are failure-prone in $\E$.
Indeed, $\Sigma_{\{p,q\}} \notin \mathcal{U}_2$.
To see this, observe that if $q$ is faulty and $p$ correct, then $\{p\}$ is eventually the output of $\Sigma_{\{p,q\}}$ at $p$.
A symmetrical argument holds for process $q$ in runs where $q$ is correct and $p$ faulty.
In the class $\mathcal{U}_2$, such values can be output in runs where both processes are correct, contradicting the intersection property of $\Sigma_{\{p,q\}}$.

Most atomic multicast protocols \cite{ramcast,tempo,whitebox,fastcast,multiringpaxos,FritzkeIMR01,delporte,RodriguesGS98} sidestep the impossibility result in \cite{gamcast} by considering that destination groups are decomposable into a set of disjoint groups, each of these behaving as a logically correct entity.
This means that there exists a partitioning $\mathfrak{P}(\Gr) \subseteq 2^{\procSet}$ satisfying that
\begin{inparaenum}
\item for every destination group $g \in \Gr$, there exists $(g_i)_i \subseteq \mathfrak{P}(\Gr)$ with $g=\union_i g_i$,
\item each $g \in \mathfrak{P}(\Gr)$ is correct, and
\item for any two $g,h$ in $\mathfrak{P}(\Gr)$, $g \inter h$ is empty.
\end{inparaenum}%
Since $\land_{g \in \mathfrak{P}(\Gr)}~(\Sigma_g \land \Omega_g) \strongereq \WFD$, we observe that solving the problem over $\mathfrak{P}(\Gr)$ is always as difficult as over $\Gr$.
It can also be more demanding in certain cases, e.g., if two groups intersect on a single process $p$, then $p$ must be reliable.
In \reffigure{family}, this happens with process $p_2$.
In contrast, to these prior solutions, \refalg{sufficiency} tolerates any number of failures.
This is also the case of \cite{SchiperP08} which relies on a perfect failure detector.


Regarding strongly genuine atomic multicast, \refsection{variations:strong} establishes that $\WFD \land (\land_{g, h \in \Gr}~\Omega_{g \inter h})$ is the weakest when $\Fa=\emptySet$.
The case $\Fa \neq \emptySet$ is a bit more intricate.
First of all, we may observe that in this case the problem is failure-free solvable:
given a spanning tree $T$ of the intersection graph of $\Gr$, we can deliver the messages according to the order $<_T$, that is, if $m$ is addressed to $g$ intersecting with $h,h',\ldots$ with $h <_T h' <_T \ldots$, then $g \inter h$ delivers first $m$, followed by $g \inter h'$, etc.%
\footnote{
  Strictly speaking, a spanning tree is required per connected component of the intersection graph.
}
A failure-prone solution would apply the same logic.
This is achievable using $\WFD \land (\land_{g, h \in \Gr}~\Omega_{g \inter h}) \land (\land_{g, h \in \Fa}~1^{g \inter h})$, where $g \in \Fa$ holds when for some family $\f \in \Fa$, we have $g \in \f$.
We conjecture that this failure detector is actually the weakest.

\section{Conclusion}
\labsection{conclusion}

This paper presents the first solution to genuine atomic multicast that tolerates arbitrary failures without using system-wide perfect failure detection.
It also introduces two new classes of failure detectors:
\begin{inparaenumorig}[]
\item ($\gamma$) which tracks when a cyclic family of destination groups is faulty, and
\item ($1^{g \inter h}$) that indicates when the group intersection $g \inter h$ is faulty.
\end{inparaenumorig}%
Building upon these new abstractions, we identify the weakest failure detector for genuine atomic multicast and also for several key variations of this problem.
Our results offer a fresh perspective on the solvability of genuine atomic multicast in crash-prone systems.
In particular, they question the common assumption of partitioning the destination groups.
This opens an interesting avenue for future research on the design of fault-tolerant atomic multicast protocols.




{
  \bibliographystyle{plainurl}
  \bibliography{paper,bib}
}  

\iflongversion
\newpage
\appendix
\section{System model}
\labappendix{model}




\subparagraph*{Basics}
We assume a distributed message-passing system composed of a set $\procSet$ of processes.
Processes execute steps of computation.
These steps are asynchronous and there is no bound on the delay between any two steps.
For the sake of simplicity, we assume a global time model, where $\naturalSet$ is the range of the global clock.
Processes cannot access to the global clock.

\subparagraph*{Failures and environments}
Processes may fail-stop, or \emph{crash}, and halt their computations.
A failure pattern is a function $F : \naturalSet \rightarrow 2^{\procSet}$ that captures how processes crash over time.
Processes that crash never recover from crashes, that is, for all time $t$, $F(t) \subseteq F(t+1)$.
If a process fails, we shall say that it is \emph{faulty}.
Otherwise, if the process never fails, it is said \emph{correct}.
$\faulty(F) = \union_t F(t)$ are the faulty processes in pattern $F$, and $\correct(F) = \procSet \setminus \faulty(F)$ denotes the correct processes.
When failure pattern $F$ is clear from the context, we shall use respectively $\correct$ and $\faulty$ for $\correct(F)$ and $\faulty(F)$.
An environment, denoted $\E$, is a set of failure patterns.
Intuitively, an environment $\E$ describes the number and timing of failures that can occur in the system.
We denote by $\E^*$ the set of all failure patterns.

\subparagraph*{Failure detectors}
A failure detector is an oracle $D$ (also called module) that processes may query locally during an execution.
This oracle abstracts information, regarding synchrony and failures, available to the processes.
More precisely, a failure detector $D$ is a mapping that assign to a failure pattern $F$, one or more histories $D(F)$.
Each history $H \in D(F)$ defines for each process $p$ in the system, the local information $H(p,t)$ obtained by querying $D$ at time $t$.
The co-domain of $H : \procSet \times \naturalSet \rightarrow R$ is named the range of the failure detector, denoted $\range(D)$.
A failure detector is realistic when it cannot guess the future \cite{realistic}.
This means that if two failure patterns have a common prefix, a process might not distinguish them in this prefix by querying the failure detector.
Formally,
$
\forall F,F' \in \E \sep \forall H \in D(F) \sep \exists H' \in D(F') \sep \forall t \in \naturalSet \sep
(\forall t' \leq t \sep F(t')=F'(t'))
\implies
(\forall t' \leq t \sep \forall p \in \procSet \sep H(p,t')=H'(p,t'))
$.

\subparagraph*{Message buffer}
Processes communicate with the help of messages taken from some set $\mathit{Msg}$.
A message $m$ is sent by some sender ($\src(m)$) and addressed to some set of recipients ($\dst(m)$).
The sender may define some content ($\payload(m)$) before sending the message.
A message buffer, denoted $\BUFF$ , contains all the messages that were sent but not yet received.
More precisely, $\BUFF$ is a mapping from processes to elements in $2^{\mathit{Msg}}$.
When a process $p$ attempts to receive a message, it either removes some message from $\BUFF[p]$, or returns a special null message ($\msgNull$).
Note that $p$ may receive $\msgNull$ even if the message buffer does contain a message addressed to $p$.

\subparagraph*{Algorithm, step and schedule}
An algorithm $\A$ consists of a family of deterministic automata, one per process in $\procSet$.
Computation proceeds in steps of these automata.
At each step, a process $p$ executes atomically all of the following instructions:
\begin{enumerate}
\item retrieve a message $m$ from $\BUFF$; 
\item retrieves some value $d$ from the local failure detector module;
\item changes its local state according to $\A$; and
\item sends some (possibly empty) message $m$, by adding $m$ to $\BUFF$. 
\end{enumerate}
For a given automaton, a step is fully determined by the current state of $p$, the received message $m$ and the failure detector value $d$.
As a consequence, we shall write a step as a tuple $s=(p,m,d)$.
A \emph{schedule} is a sequence of steps.
We write $\schedNull$ the empty schedule.

\subparagraph*{Configuration}
A \emph{configuration} of algorithm $\A$ specifies the local state of each process as well as the messages in transit (variable $\BUFF$).
Given a predicate $P$ and a configuration $C$, we write $C \sat P$ when $P$ holds in $C$. 
In some initial configuration of $\A$, no message is in transit and each process $p$ is in some initial state as defined by $\A$.
A step $s=(p,m,d)$ is \emph{applicable} to a configuration $C$ when $m \in C.\BUFF[p]$.
In which case, we note $s(C)$ the unique configuration that results when applying step $s$
(This means that in $C$, $p$ executes the code of $\A$ considering that $m$ was fetched from $\BUFF[p]$ and $d$ from the failure detector.)
This notion is extended to schedules by induction.
In detail, a schedule $S$ is applicable to configuration $C$ when $S=\schedNull$, or $S=s_1 \ldots s_{n \geq 1}$ and $s_1$ is applicable to $C$, $s_2$ is applicable to $s_1(C)$, etc.

\subparagraph*{Run}
A run of algorithm $\A$ using failure detector $D$ in environment $\E$ is a tuple $R = (F,H,I,S,T)$ where
\begin{inparaenumorig}[]
\item $F$ is a failure pattern in $\E$,
\item $H$ is a failure detector history in $D(F)$,
\item $I$ is an initial configuration of $\A$,
\item $S$ is a (possibly empty) schedule, and
\item $T \subseteq \naturalSet$ is a growing sequence of times (intuitively, $T[i]$ is the time when step $S[i]$ is taken).
\end{inparaenumorig}
A run whose schedule is finite (respectively, infinite) is called a finite (respectively, infinite) run.
Every run $R$ must satisfy the following standard, or \emph{well-formedness}, conditions that we shall assume hereafter:
\begin{itemize}
\item No process take steps after crashing.
\item The sequences $S$ and $T$ are either both infinite, or they are both finite and have the same length.  
\item The sequence of steps $S$ taken in the run conforms to the algorithm $\A$, the timing $T$ and the failure detector history $H$.
\item Every process that infinitely often retrieves a message from $\BUFF$ eventually receives every message addressed to it.
\end{itemize}
A run $\R$ is fair for some correct process $p$ when $p$ executes an unbounded amount of steps in $R$.
By extension, $\R$ is fair for $P \subseteq \correct(F)$, or for short $P$-fair, when it is fair for every $p$ in $P$.
In case $P$ is exactly $\correct(F)$, we simply say that $\R$ is fair.
%

\subparagraph*{Input/output variables}
A process $p$ interacts with the external world by reading an input queue $\IN(q)$ and writing to an output queue $\OUT(p)$.
Both queues are part of the local state of the process and contain finite binary strings.
For some given run $R$, $\IN(R)$ defines the \emph{input} of $R$.
This function maps each process $p$ to a growing sequence of pairs $(v,t)$ such that $p$ fetches $v$ from $\IN(p)$ at time $t$ in $R$.
$\OUT(R)$ is the \emph{output} of $R$, and is defined similarly.

\subparagraph*{Problems}
An input (or output) vector associates to a process $p$ a growing sequence of pairs $(v,t)$, with $v \in \{0,1\}^*$ and $t \in \naturalSet$.
A problem $\P$ specifies a desired relation between input and output vectors.
In detail, $\P$ is a set of tuples $(F,\IN,\OUT)$, where $F$ is a failure pattern, and $\IN$ and $\OUT$ are respectively an input and an output vector.
Intuitively, $(F,\IN,\OUT) \in \P$ holds if and only if when $F$ is the failure pattern and $\IN$ the input, $\OUT$ is an output that satisfies $\P$.

\subparagraph*{Solving a problem}
Consider a problem $\P$, an algorithm $\A$, a failure detector $D$, and $\E$ an environment.
Then,
\begin{itemize}
\item A run $R=(F,H,I,S,T)$ of $\A$ using $D$ in $\E$ \emph{satisfies} $\P$ if and only if $(F,\IN(R),\OUT(R))$ is in $\P$.
\item $\A$ \emph{solves} $\P$ using $D$ in $\E$ if and only if every run $R$ of $\A$ using $D$ in $\E$ satisfies $\P$.
\end{itemize}

\subparagraph*{Comparing failure detectors}
As observed in \cite{JayantiT08}, a failure detector $D$ defines itself a distributed problem in some environment $\E$.
This is the problem of building a linearizable implementation of $D$ in $\E$.
More precisely, one defines $\P_D$ as all the tuples $(F,\IN,\OUT)$ such that
\begin{inparaenum}
\item $F \in \E$, and 
\item for some history $H \in D(F)$, if ``query'' is in $I(p,t_1)$ and $r$ is a matching response in $O(p,t_2)$, then for some $t \in [t_1, t_2]$, $r$ equals $H(t)$.
\end{inparaenum}
We say that $D'$ is weaker than $D$ in $\E$ if there is an algorithm $T_{D \rightarrow D'}$ that transforms $D$ to $D'$ in $\E$.
This means that one can solve problem $\P_{D'}$ using failure detector $D$.
Notice that if $D'$ is weaker than $D$ in $\E$, then every problem that can be solved with $D'$ in $\E$ can also be solved with $D$ in $\E$.
Two failure detectors are \emph{equivalent} in $\E$ if each is weaker than the other in $\E$.
We write $D \leq_{\E} D'$ when $D$ is weaker than $D'$ in environment $\E$.
By extension, $D \leq D'$ holds when $D$ is weaker than $D'$ in every environment.

\subparagraph*{Weakest failure detector}
A failure detector $D$ is the weakest failure detector to solve problem $\P$ in environment $\E$ if and only if the following hold:%
\footnote{
  Strictly speaking, one should talk about ``a weakest'' and not ``the weakest'' failure detector because several detectors may be weakest yet not identical.
  Indeed, as pointed out in \cite{JayantiT08}, if $D$ is weakest then so is any sampling of $D$.
  However, as common in literature, we do not distinguish a failure detector from its equivalence class.
}
\begin{itemize}
\item[(Sufficiency)] $D$ can be used to solve $\P$ in $\E$.
\item[(Necessity)] For any failure detector $D'$, if $D'$ can be used to solve $\P$ in $\E$, then $D$ is weaker than $D'$ in $\E$.
\end{itemize}
In \cite{JayantiT08}, the authors prove that in every environment $\E$, for every problem $\P$, if there exists a failure detector to solve $\P$ in $\E$ then there exists a weakest failure detector for $\P$ in $\E$.

\subparagraph*{Sub-algorithms}
An algorithm $\ATwo=(\ATwo_p)_{p \in \procSet}$ is a sub-algorithm of $\A=(\A_p)_p$ if for every process $p \in \procSet$, there exists some automaton $\AThree_p$ such that $\A_p=\ATwo_p \times \AThree_p$.
Given a schedule $S$ and an algorithm $\A$, $S|\A$ is the projection of $S$ over $\A$, that is the sequence of steps of $\A$ in $S$.
In particular, such steps include reading and writing to respectively the input and output queues of $\A$.

\subsection{Technical Lemmas}
\labappendix{model:lemmas}

The results below are derived from the model detailed in the prior section.
Their proofs is left to the reader.

\begin{lemma}
  \lablem{model:1}
  Consider that $\ATwo$ is a sub-algorithm of $\A$ and pick a run $R$ of $\A$ with steps $S$ and history $H$.
  Then, there exists a run $R'$ of $\ATwo$ with steps $S'=S|\ATwo$ and history $H$.
\end{lemma}

Given some schedule $S$, relation $\hb_S$ denotes the happens-before relation in $S$.
We write $S|P$ the projection of $S$ over $P \subseteq \procSet$.
It is \emph{sound} iff for every event $e \in S|P$, if $e' \hb_{S} e$ then $e' \in S|P$.

\begin{lemma}[Indistinguishability]
  \lablem{model:2}
  \lablem{flp}
  Assume a run $R=(F,H,I,S,T)$ of $\A$ and some set of processes $P$.
  If $S|P$ is sound, then $R'=(F,H,I,S|P,T)$ is a run of $\A$.
\end{lemma}

In what precedes, we shall say that $R$ is indistinguishable from $R'$ to $P$.
The two lemmas below are analogous to Lemma 1 in FLP \cite{flp}.
They establish that if two runs with the same failure pattern and history are executed by disjoint sets of processes then they can be glued together.

\begin{lemma}
  \lablem{model:3}
  Let $R=(F,H,I,S,T)$ and $R'=(F,H,I,S',T')$ be two runs of $\A$.
  If $\proc(S) \inter \proc(S') = \emptySet$ then there exists $\hat{S}$ and $\hat{T}$ such that $(F,H,I,\hat{S},\hat{T})$ is a run of $\A$ and $\steps(\hat{S})=\steps(S) \union \steps(S')$.
\end{lemma}

\begin{lemma}
  \lablem{model:4}
  Let $R=(F,H,I,S,T)$ and $R'=(F,H,I,S',T')$ be two runs of $\A$.
  Assume that $\proc(S) \inter \proc(S') = \emptySet$ and that the last step of $S$ occurs in real time before the first step of $S'$.
  Then, there exists $\hat{T}$ such that $(F,H,I,S.S',\hat{T})$ is a run of $\A$.
\end{lemma}

\section{Emulating $\Land_{g,h \in \Gr} \Omega_{g \inter h}$}
\labappendix{omega}


\begin{algorithm}[!t]
 
  \scriptsize
  \caption{Emulating $\Omega_{g \inter h}$ -- code at process $p$}
  \labalg{omega}
   
  \begin{algorithmic}[1]

    \begin{variables}
      \State $\result \assign \text{\textbf{if} $(p \in g \inter h)$ $p$ \textbf{else} $\bot$}$ \labline{omega:var:1}
      \State $G \assign (\emptySet, \emptySet)$; $k \assign 0$ \labline{omega:var:2}
      \State $(\Upsilon_i)_{i \in [0,n]} \assign \lambda i.  (\{\schedNull\},\emptySet)$ \labline{omega:var:3} \Comment{$\schedNull$ is the empty schedule}
      \State $(\tags(S,i))_{S \in \pathSet, i \in [0,n]} \assign \lambda S i. \emptySet$ \labline{omega:var:4} \Comment{$\pathSet$ is the universal set of paths}
    \end{variables}

    \vspace{.5em}
    
    \Loop \labline{omega:0}
      \State $\sample$  \labline{omega:1}
      \If{$p \in g \inter h$} \labline{omega:2}
        \State $\simulate; \tagging; \result \assign \extract$ \labline{omega:3}
      \EndIf
    \EndLoop

    \vspace{.5em}
    
    \When{\query}
    \Return $\result$ \labline{query:1}
    \EndWhen

    \vspace{.5em}
    
    \begin{procedure}{\sample}
      \State $k \assign k + 1$ \labline{sampling:1}
      \State $d \assign D.\query()$ \labline{sampling:2}
      \State $G \assign (G.V \union \{(p,d,k)\}, G.E \union \{ ((p',d',k'), (p,d,k)) : (p',d',k') \in G.V \})$ \labline{sampling:3}
      \State $\sendTo{G}{\procSet}$ \labline{sampling:4}
      \ForAll{$g : \rcv(g)$} \labline{sampling:5}
      \State $G \assign G \union g$ \labline{sampling:6}
      \EndFor
    \end{procedure}

    \vspace{.5em}
    
    \begin{procedure}{\simulate}
      \ForAll{$i \in [0,n]$} \labline{simulation:1}
      \ForAll{$\pi \in \paths(G)$} \Comment{Following $\prefix$, non-empty ($\pi \neq \seqNull$)} \labline{simulation:2}
      \State $\mathcal{Q} \assign \langle \rangle$; $\mathcal{Q}.\push(\schedNull,0)$ \labline{simulation:3}
      \ForAll{$(S,k) \in \mathcal{Q}.\pop()$} \labline{simulation:4}
      \State \textbf{let} $(q,d,\any) = \pi[k]$ \labline{simulation:5}
      \ForAll{$m \in S(I_i).\BUFF[q] \union \{\msgNull\}$} \labline{simulation:6} \Comment{$\msgNull$ is the null message}
      \State $S' \assign S \concat (q,m,d)$ \labline{simulation:7}
      \State $\Upsilon_i \assign (\Upsilon_i.V \union \{S'\}, \Upsilon_i.E \union \{ (S, S', (q,m,d)) \} )$ \labline{simulation:8}
      \State \textbf{if} {$\cardinalOf{\pi}-1 > k$} \textbf{then} {$\mathcal{Q}.\push(S',k+1)$} \labline{simulation:9}
      \EndFor 
      \EndFor 
      \EndFor 
      \EndFor 
    \end{procedure}

    \vspace{.5em}
    
    \begin{procedure}{\tagging}
      \ForAll{$i \in [0,n]$} \labline{tagging:1}
      \ForAll{$S \in \leafs(\Upsilon_i)$} \labline{tagging:2}
      \If{$S(I_i) \sat m_x \delOrderOf{p} m_{\xbar}$} \Comment{$x \in \{g,h\} \land x=g \iff \xbar=h$} \labline{tagging:3} 
      \ForAll{$S' \in \Upsilon_i \sep S' \prefix S$} \labline{tagging:4}
      \State $\tags(S',i) \assign \tags(S',i) \union \{x\}$ \labline{tagging:5}
      \EndFor
      \EndIf
      \EndFor
      \EndFor
    \end{procedure}

    \vspace{.5em}
    
    \begin{procedure}{\extract}
      \ForAll{$i \in [0,n]$} \labline{extract:1} 
      \If{$\tags(\schedNull,i)=\{g\} \land tags(\schedNull,j)={h} \land I_i \ndist{q} I_j$} \labline{extract:2}
        \Return $q$ \labline{extract:3}
      \ElsIf{$\tags(\schedNull,i)=\{g,h\}$} \labline{extract:4}
        \ForAll{$Q \in 2^{\procSet} : Q \neq \emptySet$} \labline{extract:5}
          \State \textbf{let} $q = \locate(Q,i)$ \labline{extract:6}
          \If{$q \in g \inter h$} \labline{extract:7}
            \Return $q$ \labline{extract:8}
          \EndIf            
        \EndFor 
      \EndIf
      \Return $p$ \labline{extract:9}
    \EndFor
    \end{procedure}

    \vspace{.5em}
    
    \begin{function}{$\locate(Q=\{q_0,\ldots\},i)$}
      \State $S \assign \schedNull$; $k \assign 0$ \labline{locate:1}
      \Loop \labline{locate:2}
      \State \textbf{let} $q = q_{k}$, $m$ = oldest message (if none, then $\msgNull$) addressed to $q$ in $S(I_i)$ \labline{locate:3}
      \ForAll{$S \concat (q,m,d) \in \Upsilon_i : \exists x \in \tags(S \concat (q,m,d),i)$} \Comment{In order $<$} \labline{locate:4}
      \State \textbf{let} $S_x = S \concat (q,m,d)$ \labline{locate:5} 
      \If{$\tags(S_x,i)=\{g,h\}$} \labline{locate:6}
        \State $S \assign S_x$; $k \assign \modOf{(k+1)}{\cardinalOf{Q}}$  \labline{locate:7}
        \Go line~\ref{line:alg:locate:2} \labline{locate:8}
      \Else \labline{locate:9}
        \ForAll{$S' \in \successors(S) : \exists d' \in \range(D) \sep \xbar \in \tags(S' \concat (q,m,d'),i)$} \Comment{In order $<$} \labline{locate:10}
        \State \textbf{let} $S_{\xbar} = S' \concat (q,m,d')$ \labline{locate:11}
        \If{$\tags(S_{\xbar},i)=\{g,h\}$} \labline{locate:12}
          \State $S \assign S_{\xbar}$; $k \assign \modOf{(k+1)}{\cardinalOf{Q}}$ \labline{locate:13}
          \Go line~\ref{line:alg:locate:2} \labline{locate:14}
        \ElsIf{$\forall \hat{S} \in [S,S') \sep \hat{S} \concat (q,m,d') \in \Upsilon_i \land \cardinalOf{\tags(\hat{S} \concat (q,m,d'),i)}=1$} \labline{locate:15}
          \Return $\deciding(S,S_x,S_{\xbar})$ \labline{locate:16}
        \EndIf          
        \EndFor
      \EndIf
      \EndFor
      \Return $\bot$ \labline{locate:17}
      \EndLoop
    \end{function}
    
  \end{algorithmic}

\end{algorithm}

Consider an arbitrary environment $\E$, a failure detector $D$ and a strongly genuine solution $\A$ that uses $D$.
Given $g,h \in \Gr$, the construction of $\Omega_{g \inter h}$ from $D$ and $\A$ is depicted in \refalg{omega}.
This algorithm follows the general schema of CHT \cite{omega}, with some differences that we detail hereafter.

\refalg{omega} consists of four procedures that are run in a loop (\reflines{omega:0}{omega:3}).
Procedure $\sample$ is a collaborative sampling of failure detector $D$.
Procedure $\simulate$ executes runs of $A$ using this sampling and multiple initial configurations.
These runs form a forest which is tagged appropriately by the $\tagging$ procedure.
Based upon those tags, $\extract$ computes an eventual leader for the group intersection $g \inter h$.

The sections that follow detail each of these procedures and the guarantees they offer.
Then, \reftheo{correctness} shows that \refalg{omega} implements $\Omega_{g \inter h}$.

\subsection{Collaborative sampling of $D$}
\labappendix{omega:sampling}

The first procedure of \refalg{omega} implements a collaborative sampling of failure detector $D$.
The output of this sampling is stored in variable $G$ which contains, at each process, a directed acyclic graph.
Hereafter, we use selectors $G.V$ and $G.E$ to refer respectively to the vertices and the edges of $G$.

The $\sample$ procedure works as follows.
When a process $p$ executes this procedure for the $k$-th time, it retrieves a datum $d$ from failure detector $D$ (\refline{sampling:2}).
Process $p$ then adds a vertex $(p,d,k)$ to $G$ and an edge from this vertex to every vertex already existing in $G$ (\refline{sampling:3}).
Then, the new value of $G$ is sent to all the processes in the system.
Upon reception of such a message, every process merges this sample into its local graph (\refline{sampling:6}).

From the above logic, we may observe that every path in $G$ at a correct process eventually appears at every correct process.
This simply follows from the fact that links are reliable.
For the moment, we shall assume that the procedures $\simulate$, $\tagging$ and $\extract$ always return.
This result is proved later.

\begin{lemma}
  \lablem{sampling:1}
  Consider a correct process $p$ and some path $\pi$ that eventually appears in $G$.
  Eventually, $\pi$ appears in $G$ at every correct process.
\end{lemma}

\begin{proof}
  First of all, observe that variable $G$ is monotonically growing over time (\reflinestwo{sampling:3}{sampling:6}).
  Consider a point in time $t$ at which $\pi$ is in $G$.
  Then, $\pi$ is in $G$ at every later time $t'>t$.
  Consider some correct process $q$.
  By assumption, process $p$ eventually executes \refline{sampling:4}, sending ${G}^{p,t'}$ to $q$.
  Since links are reliable, $q$ eventually receives this sample and merges it with variable $G$ (\refline{sampling:6}).
\end{proof}

Now, let us assume a run $R=(F,H,\any,\any,T)$ of the $\sample$ procedure.
Consider a sequence $\pi$ of tuples $(p,d,k)$, with $p \in \procSet$, $d \in \range(D)$ and $k \in \naturalSet$.
We say that $\pi$ \emph{is a sampling of $D$ in $R$} when there exists a mapping $\tau$ from the elements in $\pi$ to $T$ such that for every $v=(p,d,\any)$ in $\pi$,
\begin{inparaenum}
\item $p$ is not faulty at time $\tau(v)$, i.e., $p \notin F(\tau(v))$,
\item $d$ is a valid sample of $D$ at that time, that is $d=H(p,\tau(v))$, and
\item if $v$ precedes $v'$ in $\pi$ then $\tau(v) < \tau(v')$.
\end{inparaenum}
The mapping $\tau$ is named a \emph{sampling function} of $\pi$.
When the context is clear, $\pi$ is simply called a sampling.
\refprop{sampling:2} establishes that every path in the graph built with the $\sample$ procedure is a sampling.

\begin{proposition}
  \labprop{sampling:2}
  Consider some point in time of a run $R=(F,H,\any,\any,\any,T)$ of \refalg{omega}.
  Every path $\pi$ in $G$ is a sampling of $D$ in $R$.
\end{proposition}

\begin{proof}
  If $\pi$ is the empty sequence (that is, $\pi = \seqNull$), $\tau=\emptySet$ is a sampling function.
  Otherwise, $\pi=(q_1,d_1,k_1) (q_2,d_2,k_2), \ldots$.
  For each tuple $v=(q_i,d_i,k_i)$ in $\pi$, $\tau(v)$ is set to the time process $q_i$ executes \refline{sampling:2} to retrieve $d_i$ for the $k_i$-th time.
  One observes that at time $\tau(v)$, process $q_i$ is not faulty and $d$ equals $H(q_i,\tau(v))$.
  Moreover, according to \refline{sampling:3}, if $v$ precedes $v'$ in $\pi$ then $\tau(v) < \tau(v')$.
  From what precedes, we deduce that $\pi$ is a sampling.
\end{proof}

In the proof above, $\tau$ is built by taking the time at which each sample in $\pi$ is computed in $R$.
Hereafter, the sampling function of $\pi$ refers to this unique way of computing $\tau$.

A family $(\pi_k)_{k \in \naturalSet}$ is a \emph{sampling sequence} when for every $k \in \naturalSet$, $\pi_k$ is a sampling and $\pi_k$ strictly prefixes $\pi_{k+1}$.
%
The sampling sequence $(\pi_k)_k$ is \emph{replicated} when for every correct process $p$, for every $k \in \naturalSet$, $\pi_k$ is eventually always in $\paths(G^{p})$.
Given some process $p$, $(\pi_k)_k$ is $p$-fair when for every $m>0$, there exists $k$ such that $p$ appears $m$ times in $\pi_k$.
By extension, $(\pi_k)_k$ is $P$-fair when it is $p$-fair for every $p \in P$.
When $P$ is the set of correct processes, we shall simply say that $(\pi_k)_k$ is fair.
\refprop{sampling:3} below characterizes how replicated sampling sequences grow at the correct processes.

\begin{proposition}
  \labprop{sampling:3}
  Consider a fair run $\R$ of \refalg{omega}, some correct processes $P \subseteq \correct(R)$ and $p \in P$.
  Let $t$ be some point in time during $R$.
  There exists $(\pi_k)_{k \in \naturalSet}$ such that
  \begin{inparaenum}
  \item for each $k$, only processes in $P$ appears in $\pi_k$, and
  \item for every path $\pi$ in $G^{p,t}$, $(\pi \concat \pi_k)_{k \in \naturalSet}$ is a $P$-fair replicated sampling sequence.
  \end{inparaenum}
\end{proposition}

\begin{proof}
  Graph $G$ is initially empty.
  Hence, by a short induction, when $G$ is extended by executing either \refline{sampling:3} or \refline{sampling:6}, this computation terminates.
  It follows that every correct process executes infinitely often the loop in \reflines{sampling:1}{sampling:6} in $R$.
  Consider a moment in time $t'>t$ at which some correct process $q$ executes an iteration of this loop.
  Since links are reliable, it is true that:
  \begin{inparaenumorig}[\em(1)]
  \item Every correct process eventually receives ${G}^{q,t'}$ and merges it with variable $G$; and
  \item At every correct process $q'$, for every $v \in {G}^{q,t'}.V$, there exists eventually some edge $(v,(q',\any,\any))$ in $G^{q'}$.
  \end{inparaenumorig}
  
  We build inductively $(\pi_k)_{k \in \naturalSet}$ using the above two observations.
  In detail, $\pi_1$ is defined as $\seqNull$, the empty sequence.
  Then, for every $k>1$, $\pi_{k+1}$ is constructed from $\pi_{k}$ as follows:
  \begin{construction}{replicated}
    Initially, $\pi_{k+1}$ is set to $\pi_{k}$.
    For every correct process $q \in P$, starting from $p$ at time $t$, waits that $v'$ the last element in $\pi_{k+1}$ appears in $G$.
    (If there is no such element, this condition is vacuously true.)
    Then, let $(q,\any,\any)$ be the next vertex added by $q$ to $G$ (at \refline{sampling:3}).
    Append $v$ to $\pi_{k+1}$.
  \end{construction}
  Observations (1) and (2) above imply that, for every $k>1$, $\pi_k$ is indeed built with \refconstruction{replicated}.
  From the pseudo-code, if $(q,\any,\any)$ is in some $\pi_k$ then necessarily $q \in P$.
  In addition, for every process $q \in P$, for every $m>0$, there exists $k$ such that $q$ appears $m$ times in $\pi_k$.
  Then, consider some $\pi \in \paths(G^{p,t})$.
  For some $k$, consider the process $q$ that creates the last vertex $v$ of $\pi_{k}$.
  By a short induction, $\pi \concat \pi_{k}$ is in $\paths(G^q)$ at that time.
  Thus, applying \refprop{sampling:3}, it is a sampling.
  Moreover, as $q$ is correct, observation (1) implies that eventually $\pi \concat \pi_{k}$ is always in $\paths(G)$ at every correct process.
  From what precedes, $(\pi \concat \pi_k)_k$ is a $P$-fair replicated sampling sequence.
\end{proof}

\subsection{Building the simulation forest}
\labappendix{omega:simulation}

A key observation in CHT \cite{omega} is that every path in the sampling graph induces a valid schedule for the current failure pattern.
Based on this observation, each process builds a simulation forest to explore an unbounded \emph{yet countable} number of runs of algorithm $\A$.

In our context, the runs of algorithm $\A$ that bring interest satisfy that
\begin{inparaenum}
\item the processes outside $g \inter h$ do not atomic multicast any message, and
\item the processes in $g \inter h$ multicast a single message to either $g$ or $h$.
\end{inparaenum}
Let us note $\I=\{I_1,\ldots,I_{n \geq 2}\}$ this set of initial configurations.
Configurations $I_i$ and $I_j$ are adjacent with respect to $q$, written $I_i \ndist{q} I_j$, when they differ in at most the state of process $q$.

Procedure $\simulate$ constructs a simulation forest $(\Upsilon_i)_{i \in [0,n]}$ from the schedules of $\A$ induced by the sampling graph $G$ and the initial configurations $(I_i)_{i \in [0,n]}$.
More precisely, each vertex in the tree $\Upsilon_i$ is a schedule starting from the initial configuration $I_i$.
The root of $\Upsilon_i$ is the empty schedule $S_\bot$ (\refline{omega:var:1}).
There is an edge labeled $e$ from $S$ to $S'$ in $\Upsilon_i$ when $S'=S.e$ holds for some step $e$.

Procedure $\simulate$ builds the forest $\Upsilon_i$ by considering all the paths in $G$ (\refline{simulation:2}).
For each such path $\pi$, $\simulate$ creates in $\Upsilon_i$ the schedules induced by $\pi$ (\reflines{simulation:3}{simulation:9}).
To this end, the procedure relies on a queue $\mathcal{Q}$ of pairs $(S,k)$, where $S$ is a schedule of $\A$ and $k$ tracks the next sample in $\pi$ used to create a step.
Initially, $\mathcal{Q}$ contains $(\schedNull,0)$, the empty schedule (\refline{simulation:3}).
Starting from a pair $(S,k)$ in $\mathcal{Q}$, $\simulate$ considers the $k$-th sample $(p,d,\any)$ in $\pi$ (\reflinestwo{simulation:4}{simulation:5}).
For every message $m$ addressed to $p$ in $S(I_i)$, the schedule $S \concat (p,m,d)$ is added to $\Upsilon_i$ (\reflines{simulation:6}{simulation:8}).
This process is repeated until all the schedules \emph{compatible with} $\pi$ have been explored (\refline{simulation:9}).
In detail, a schedule $S$ is compatible with a path $\pi$ when 
\begin{inparaenum}
\item $S=\schedNull$ and $\pi=\seqNull$, or
\item $S$ and $\pi$ have the same length and denoting $\pi=(q_1,d_1,k_1)(q_2,d_2,k_2)\ldots$, there exist some some (possibly null) message $m_1,m_2,\ldots$ such that $S=(q_1,m_1,d_1)(q_2,m_2,d_2)\ldots$.
\end{inparaenum}
The lemma below proves a key result regarding the simulation.

\begin{lemma}
  \lablem{simulation:0}
  A schedule $S \in \Upsilon_i$ is always compatible with some path in $G$ and applicable to $I_i$.
  Conversely, if $S$ is applicable to $I_i$ and compatible with a path $\pi$ in $G$ then $S$ is in $\Upsilon_i$.
\end{lemma}

\begin{proof}
  The first part of the lemma is obtained by induction:
  In detail, this is trivial when $S=\schedNull$.
  Then, consider when $S$ is added to $\Upsilon_i$ at \refline{simulation:8}.
  We have $S=S' \concat (q,m,d)$.
  By induction hypothesis, there exists a path $\pi'$ such that $S'$ is compatible with $\pi'$ and $S'$ applicable to $I_i$.
  $S$ is compatible with $\pi' \concat (q,m,\any)$, where $(q,m,\any) = \pi[k]$, as written at \refline{simulation:5}.
  As $S$ suffixes $S'$ and $S'$ is applicable to $I_i$, so is $S$.

  Regarding the second part of the lemma, assume that $S$ is applicable to $I_i$ and compatible with $\pi \in G$.
  If $S=\schedNull$ then according to \refline{omega:var:3}, $S$ is in $\Upsilon_i$.
  Otherwise, $\pi$ is added to $G$ at \refline{omega:3}
  Hence, \refline{simulation:2} is executed afterward with path $\pi$.
  (In our system model, this happens at the same logical time.)
  When \refline{simulation:2} is executed, by induction on the pseudo-code at \reflines{simulation:3}{simulation:9} and using the fact that $S$ is applicable to $I_i$, $S$ is added to $\Upsilon_i$.
\end{proof}

From which, we deduce easily the following result:

\begin{corollary}
  \labcor{simulation:0a}
  Consider a correct process $p$ and pick some schedule $S$ that eventually appears in $\Upsilon_i$ at $p$.
  Then, schedule $S$ is also eventually in $\Upsilon_i$ at every correct process.
\end{corollary}

\begin{proof}
  Pick some correct process $q$.
  From \reflem{simulation:0}, schedule $S$ is compatible with some path $\pi$ and applicable to $I_i$.
  By \reflem{sampling:1}, $\pi$ eventually appears in $G$ at $q$.
  Applying again \reflem{simulation:0}, $S$ is eventually in $\Upsilon_i$ at $q$.
\end{proof}

We now explain how each schedule in the simulation forest induces a run of algorithm $\A$ for the current failure pattern.
Further, we prove that each $P$-fair replicated sampling sequence induces a sequence of runs of $\A$ that converges toward a $P$-fair run.

For starters, consider a run $R$ of \refalg{omega} and a correct process $p$.
At $p$, each schedule $S$ in a tree $\Upsilon_i$ corresponds to some run of algorithm $\A$.
This run is built using function $\runOf{S,R,i}$ defined as follows:
\begin{construction}{run}
  Let $\pi$ be a path such that $S$ is compatible with $\pi$ (by \reflem{simulation:0}).
  Let $F$ and $H$ be respectively the failure pattern and failure detector history of run $R$.
  Define $\tau$ as the sampling function of $\pi$ (by \refprop{sampling:2}). 
  Let $T$ be $(\codomain(\tau),<)$, with $<$ the usual total order on naturals.
  Function $\runOf{S,R,i}$ maps $S$ to $\langle F, H, I_i, S, T \rangle$.
\end{construction}
The result below establishes that $\runOf{S,R,i}$ is indeed a run of $\A$.

\begin{proposition}
  \labprop{simulation:1}
  For every schedule $S$ in $\Upsilon_i$, $\runOf{S,R,i}$ is a (finite) run of $\A$.
\end{proposition}

\begin{proof}
  Let $\langle F, H, I_i, S, T \rangle$ be the value of $\runOf{S,R,i}$.
  According to \refconstruction{run}, $F$ and $H$ are respectively the failure pattern $R.F$ and the history $R.H \in D(R.F)$.
  $I_i$ is an initial configuration of $\A$, as defined above.
  If $S=\schedNull$, then according to \refconstruction{run}, $T=\emptySet$.
  Otherwise, assume inductively that for each element $(S',\any)$ in $\mathcal{Q}$, $\runOf{S',R,i}$ is a run of $\A$.
  Schedule $S$ is built by concatenating a step $(q,m,d)$ at \refline{simulation:7} from some $(S',k-1)$ in $\mathcal{Q}$.
  By our induction hypothesis, $\runOf{S',R,i}$ is a run of $\A$.
  Denote $\langle F, H, I_i, S', T' \rangle$ the value of $\runOf{S',R,i}$.
  According to \refconstruction{run},  we have $T = T'\concat t$. 
  Furthermore, $\cardinalOf{S} = \cardinalOf{T}$, $H(q,t)=d$ and $q$ is not faulty at time $t$.
  From the pseudo-code at \reflines{simulation:4}{simulation:7},
  \begin{inparaenum}
  \item $S = S' \concat (q,m,d)$,
  \item $(q,d,\any) = \pi[k]$, and
  \item $m \in S'(I_i).\BUFF[p] \union \{\msgNull\}$.
  \end{inparaenum}
  As a consequence, $\runOf{S,R,i}$ is a run of $\A$.
\end{proof}

Consider some path $\pi$ in graph $G$ at process $p$.
In what follows, $\sched(\pi,I_i)$ denotes the schedule starting from $I_i$ built from $\pi$ with $\simulate$ when at \refline{simulation:6} the oldest message addressed (if none, then $\msgNull$) to $p$ is always retrieved from $\BUFF[p]$.
Notice that by construction, $\sched$ is a function of $\pathSet \times \{I_1, \ldots, I_{n \geq 2}\}$.

\begin{lemma}
  \lablem{simulation:1}
  For every path $\pi$ in graph $G$, for every initial configuration $I_i$, $\sched(\pi,I_i)$ is compatible with $\pi$ and applicable to $I_i$.
\end{lemma}

\begin{proof}
  We proceed by induction on the length of $\pi$.
  If $\pi=\seqNull$, then $\sched(\pi,I_i)=\schedNull$.
  By definition $\schedNull$ is compatible with $\seqNull$ and applicable to $I_i$.
  Then, consider that $\pi=\pi'{\concat}(p',d',k')$ and $\sched(\pi',I_i)$ compatible with $\pi'$ and applicable to $I_i$.
  $\sched(\pi,I_i)$ is computed at \refline{simulation:6} such that $\sched(\pi,I_i)=S \concat (q,m,d)$.
  By a short induction, at that time,
  \begin{inparaenum}
  \item variable $S$ equals $\sched(\pi',I_i)$,
  \item the tuple $(q,d,\any)$ equals $(p',d',k')$, and
  \item variable $m$ is the oldest message addressed (if none, then $m=\msgNull$) to $p$ in $\sched(\pi',I_i).(I_i).\BUFF[p]$.
  \end{inparaenum}
  It follows that $\pi$ is compatible with $\sched(\pi,I_i)$ and $\sched(\pi,I_i)$ is applicable to $I_i$.
\end{proof}

\begin{lemma}
  \lablem{simulation:2}
  Consider $\pi,\pi' \in G$ with $\pi={\pi'}{\concat}{\pi''}$.
  Let
  $S=\sched(\pi,I_i)$
  $S'=\sched(\pi',I_i)$,
  and $E=\sched(\pi'',S'(I_i))$.
  It is true that $S=S' \concat E$.
\end{lemma}

\begin{proof}
  By induction.
\end{proof}

\begin{proposition}
  \labprop{simulation:2}
  Consider a run $\R$ of $\refalg{omega}$ and a set of correct processes $P \subseteq \correct(\R)$.
  Let $(\pi_k)_{k \in \naturalSet}$ be a $P$-fair replicated sampling sequence and $I_i$ some initial configuration.
  The family of schedules $(S_k)_k=(\sched(\pi_k,I_i))_k$ satisfies that 
  \begin{inparaenum}
  \item $S_k$ strictly prefixes $S_{k+1}$,
  \item every $S_k$ is eventually always in $\Upsilon_i$ at each $p \in P$, and
  \item $(\runOf{S_k,R,i})_k$ converges toward a $P$-fair run of $\A$.
  \end{inparaenum}
\end{proposition}

\begin{proof}
  Consider a replicated sampling sequence $(\pi_k)_k$.
  For starters, we observe that:
  \begin{inparaenum}
  \item for every $k$, as $\pi_k$ is eventually in $G$, then $S_k$ is eventually always in $\Upsilon_i$ (\reflemtwo{simulation:0}{simulation:1}),
  \item since $\pi_k$ prefixes $\pi_{k+1}$, $S_k$ prefixes $S_{k+1}$ (\reflem{simulation:2}), and
  \item for every $k$, $R_k=\runOf{S_k,R,i}$ is a run of $\A$ (\refprop{simulation:1}).
  \end{inparaenum}
  
  Let $\mathcal{A}^{\omega}$ be the space all the $\omega$-words built atop the alphabet $\mathcal{A} = \procSet \times \mathit{Msg} \times \naturalSet$.
  $\mathcal{A}^{\omega}$ is a metric space for the usual distance function $d(u,v) = \inf \{ 2^{\cardinalOf{p}} : u \prefix p \land p \prefix v\}$.
  Define $S_{\infty}$ as $\lim_{k \rightarrow \infty} \sched(I_i,\pi_k)$.
  Observe that $(S_k)_k$ is a growing sequence of $\omega$-words which satisfies that $d(S_{k+1},S_{\infty}) < d(S_{k+1},S_{\infty})$.
  Hence, $(S_k)_k$ converges toward $S_{\infty}$ in $\mathcal{A}^{\omega}$.
  The very same reasoning applies to construct $T_{\infty}$ the sequence of times built with the sampling functions $(\tau_k)_k$ of $(\pi_k)_k$.
  Thus, for some appropriate metric space, $(\runOf{S_k,R,i})_k$ converges toward $\hat{\R}=\langle F, H, I_i, S_{\infty}, T_{\infty} \rangle$ 
  In $\hat{\R}$, only the processes in $P$ takes an unbounded amount of steps.
  As $P \subseteq \correct(F)$, this run is $P$-fair.
  Moreover in $\hat{\R}$,
  \begin{inparaenum}
  \item No process take steps after crashing.
  \item Both $S_{\infty}$ and $T_{\infty}$ are infinite.
  \item By induction, $S_{\infty}$ conforms to $\A$, the timing $T_{\infty}$ and $H$.
  \item Every message addressed to some process in $P$ is eventually delivered.
  \end{inparaenum}
  It follows that $\hat{\R}$ is a (well-formed) run of $\A$.
\end{proof}

\begin{proposition}
  \labprop{simulation:3}  
  Consider some run $R$ of \refalg{omega} and $t$ a point in time during $R$.
  Let $S$ be a schedule in $\Upsilon_i^{p,t}$ and $m$ be a message multicast in $S$ by some correct process.
  If $p \in \dst(m)$ then there exists a (finite) schedule $E$ such that
  \begin{inparaenum}
  \item $S \concat E$ is eventually always in $\Upsilon_i^p$,
  \item only the correct processes in $\dst(m)$ take steps in $E$, and
  \item process $p$ delivers $m$ in $S \concat E(I_i)$.
  \end{inparaenum}
  Furthermore, for every schedule $S' \in \Upsilon_i^{p,t}$, if $E$ is applicable to $S'(I_i)$, then $S' \concat E$ is eventually always in $\Upsilon_i^p$
\end{proposition}

\begin{proof}
  Define $P$ as $\dst(m) \inter \correct(\R)$.
  By \refprop{sampling:3}, there exists $(\pi_k)_{k \in \naturalSet}$ such that
  \begin{inparaenum}
  \item for each $k$, only processes in $P$ appears in $\pi_k$,
  \item for every $\pi \in \paths(G^{p,t})$, $(\pi \concat \pi_k)_{k \in \naturalSet}$ is a $P$-fair replicated sampling sequence.
  \end{inparaenum}

  Choose some schedule $S \in \Upsilon_i^{p,t}$ and choose $\pi \in \paths(G^{p,t})$ such that $S$ is compatible with $\pi$.
  Applying \refprop{simulation:2} to $(\pi \concat \pi_k)_{k \in \naturalSet}$, the family of schedule $S_k=(\sched(\pi \concat \pi_k,I_i))_k$ is such that
  \begin{inparaenum}
  \item $S_k$ is eventually always in $\Upsilon_i^p$, and
  \item $(\runOf{S_k,R,i})_k$ converges toward some $P$-fair run $\hat{\R}$ of $\A$.
  \end{inparaenum}
  By strong genuineness (see \refsection{variations:strong:definition}), all of the correct processes in $\dst(m)$ deliver $m$ in $\hat{\R}$,
  Thus, as $p \in P$, for some $l \in \naturalSet$, this also happens in $\runOf{S_l,R,i}$ to process $p$. 
  Since $\pi \prefix \pi \concat \pi_l$, \reflem{simulation:2} tell us that $S_l=S \concat E$ with $E=\sched(\pi_l,S(I_i))$.
  Now, from what precedes, $S_l$ is eventually always in $\Upsilon_i^p$. 
  Moreover, by definition of $(\pi \concat \pi_k)_k$ only processes in $P$ take steps after $S$ in $S_l$. 

  Consider some schedule $S' \in \Upsilon_i^{p,t}$.
  Pick $\pi' \in \paths(G^{p,t})$ such that $S'$ is compatible with $\pi'$ (by \reflem{simulation:0}).
  By definition of $(\pi_k)_k$, at some point in time $t'>t$, $(\pi' \concat \pi_l) \in G^{p,t'}$.
  Then consider the next moment in time $t''>t'$ at which $p$ calls procedure \simulate.
  As $(\pi' \concat \pi_l) \in G^{p,t''}$, $p$ executes \reflines{simulation:2}{simulation:9} for this path.
  By a short induction on these lines of code, if $E$ is applicable $S'(I_i)$, then $S' \concat E$ is added to $\Upsilon_i^p$ at that time.
\end{proof}

\subsection{Tagging the forest}
\labappendix{omega:tagging}

Procedure $\tagging$ associates to each schedule $S$ in the simulation tree $\Upsilon_i$ a set of tags $\tags(S,i) \subseteq \{g,h\}$.
These tags correspond to how messages are delivered in $S$ from configuration $I_i$.
More precisely, $\tags(S,i)$ contains $g$ (resp. $h$) if and only if for some successor $S'$ of $S$, a process delivers first a message $m_g$ addressed to $g$ (resp. $m_h$ to $h$) in $S'(I_i)$ (\reflines{tagging:3}{tagging:5}).
Notice that at \refline{tagging:3}, for simplicity, $\xbar$ equals $h$ if $x=g$ holds, and $h$ otherwise.
A subtree of $\Upsilon_i$ is \emph{stable} when every schedule in the subtree has a non-empty set of tags which does not change over time.
\refprop{tagging:1} proves that every subtree in $\Upsilon_i$ is eventually stable at the processes in $g \inter h$.

\begin{proposition}
  \labprop{tagging:1}
  Consider a correct process $p \in g \inter h$ and some schedule $S$ that eventually appears in the simulation tree $\Upsilon_i$ at $p$.
  At process $p$, $\tags(S,i)^{t}$ is monotonically growing and converges over time toward some non-empty value.
\end{proposition}

\begin{proof}
  Consider some run $R$ of \refalg{omega}, some schedule $S$ eventually in $\Upsilon_i$.
  Function $\tags(S,i)^t$ is monotonically growing over time (\refline{tagging:5}) and clearly bounded;
  thus, it is convergent.
  Choose some path $\pi \in G$ such that $S$ is compatible with $\pi$ (\reflem{simulation:0}).
  Name $m$ the message sent by $p$ and addressed to some group $x$.
  Applying \refprop{simulation:3}, since $p \in \dst(m)$ then there exists a (finite) schedule $S'$ suffixing $S$ such that
  \begin{inparaenumorig}[]
  \item process $p$ delivers $m$ in $S'(I_i)$, and
  \item $S'$ is eventually in $\Upsilon_i$.
  \end{inparaenumorig}
  Let $t$ be the time at which $S'$ is in $\Upsilon_i$.
  Without lack of generality, assume that $S'(I_i) \sat m_{x} \delOrderOf{p} m_{\xbar}$.
  When procedure $\tagging$ executes after time $t$, $x$ is added to $\tags(S,i)$ (\reflinestwo{tagging:4}{tagging:5}).
\end{proof}

From which, we may deduce that:

\begin{corollary}
  \labcor{tagging:2}
  There exists a non-empty set of tags, such that for any correct processes $p$, $\tags(S,i)^{p,t}$ converges over time toward this value.
\end{corollary}

\begin{proof}
  Consider some correct processes $p$.
  Applying~\refprop{tagging:1}, $\tags(S,i)^{p,t}$ converges toward some non-empty value $T$.
  Consider some group $x \in T$.
  According to \refline{tagging:3}, it must be the case that for some schedule $S'$ suffixing $S$, $S'(I_i) \sat m_x \delOrderOf{p} m_{\xbar}$.
  Applying \refcor{simulation:0a}, $S$ and $S'$ are eventually included in $\Upsilon_i$ at every other correct process $q$.
  Assume this happens at time $t'$.
  At time $t'$, when $q$ executes $\tagging$, it adds $x$ to $\tags(S,i)$ (\reflines{tagging:3}{tagging:5}).
\end{proof}

\subsection{Extracting the leader}
\labappendix{omega:extract}

Outside $g \inter h$, a process returns the value $\bot$ when querying its failure detector module (\refline{query:1}).
In the intersection, a process traverses the simulation forest and eventually locates a correct leader in $g \inter h$.
To this end, it relies on $\extract$, the last building block of \refalg{omega} (\reflines{extract:1}{extract:9}).

As in \cite{omega}, the procedure $\extract$ uses valency arguments to compute the eventual leader.
A schedule $S$ in a tree $\Upsilon_i$ is \emph{$g$-valent} (respectively, \emph{$h$-valent}) when $\tags(S,i)$ equals $\{g\}$ (resp. $\{h\}$).
In the case where $\tags(S,i)$ contains both groups, $S$ is \emph{bivalent}.
When $S$ is not bivalent, yet tagged, we shall say that it is \emph{univalent}.

From \refprop{tagging:1}, eventually the tags of every schedule in $\Upsilon_i$ get stable.
An index $i \in [0,n]$ is \emph{critical} when either
\begin{inparaenum}
\item the root of $\Upsilon_i$ is stable and bivalent, and for each group $x \in \{g,h\}$, some correct process multicasts a message to $x$ in $I_i$, or
\item the root of $\Upsilon_{i}$ is stable and $g$-valent, the root of $\Upsilon_j$ is stable and $h$-valent, and $I_i$ and $I_j$ are adjacent.
\end{inparaenum}
In the first case, index $i$ is \emph{bivalent} critical and in the other, it is \emph{univalent} critical.

For a process $p \in g \inter h$, $\extract$ iterates over all the trees in the simulation forest (\refline{extract:1}).
If $i$ is univalent critical, then $p$ returns the process connecting the two configurations (\refline{extract:2}).
If now $i$ is bivalent critical, function $\locate(P,i)$ is executed for every subset $P$ of $g \inter h$ (\refline{extract:6}).
This function traverses $\Upsilon_i$ to locate a correct process in a particular subtree, named a decision gadget (\refline{locate:16}).
If the traversal fails, or the valency of the root of $\Upsilon_i$ is not established yet, $\extract$ returns the local process (\refline{extract:9}).

In \refprops{extract:1}{extract:6}, several results are established regarding the behavior of $\extract$ at a correct process in $g \inter h$.
Based on these results, \reftheo{correctness} proves that \refalg{omega} properly emulates $\Omega_{g \inter h}$.

In detail, consider the simulation forest $(\Upsilon_i)_i$ at some correct process $p \in g \inter h$.
\refprop{extract:1} proves that eventually some index $i$ is critical.
If $i$ is univalent, \refprop{extract:2} shows that the process connecting the two configurations must be correct.
The case of a bivalent critical index is considered in \refprops{extract:3}{extract:6}.
In particular, \refprop{extract:3} explains how a correct process in $g \inter h$ is computed from a decision gadget.
Then, \refprops{extract:4}{extract:6} prove that a decision gadget eventually shows up in $\Upsilon_i$.

\tikzstyle{acircle} = [draw,circle,minimum size=3.2em,inner sep=0em]

\begin{figure}[t]
  \centering
  \captionsetup{justification=centering}
  \hspace{-5em}%
  \begin{subfigure}[t]{.5\textwidth}
    \centering
    \begin{tikzpicture}[
        node distance=6em,
        scale=.75,
        transform shape
      ]
      \begin{scope}[<-,thick]
        \node[adot,label=90:{$J_0,\{g\}$}] (j0) at (0,0) {};
        \node[adot,right of=j0] (j1) {};
        \node[adot,right of=j1,xshift=3em,label=90:{$J_k,\{g,h\}$}] (jk) {};
        \node[adot,right of=jk] (jn1) {};
        \node[adot,right of=jn1,label=90:{$J_{n},\{h\}$}] (jn) {};
        \node[below of=j0,yshift=5em] () {$G$};
        \node[below of=j1,yshift=5em] () {$hG$};
        \node[below of=jk,yshift=5em] () {$?{\color{OliveGreen}{G}}?$};
        \node[below of=jn1,yshift=5em] () {$gH$};
        \node[below of=jn,yshift=5em] () {$H$};
        \path[basic] (j0) edge node [midway,above] {$\ndist{p_1}$} (j1);
        \path[basic] (jn1) edge node [midway,above] {$\ndist{p_{n-1}}$} (jn);
        \path[basic,dashed] (j1) -- (jk) -- (jn1);
        \node[adot,below of=j1,label=90:{$J',\{x\}$}] (jprime) {};
        \node[below of=jprime,yshift=4.9em] () {$?{\color{OliveGreen}{hG}}?$};
        \path[basic] (jk) edge node [midway,above] {$\ndist{q}$} (jprime);
        \node[adot,right of=jprime,label=90:{$\hat{J},\{x\}$}] (jhat) {};
        \node[below of=jhat,yshift=5em] () {$\overline{X}{\color{OliveGreen}{hG}}\overline{X}$};
        \path[basic,dashed] (jprime) -- (jhat);
      \end{scope}
    \end{tikzpicture}
  \end{subfigure}
  \hspace{5em}%
  \begin{minipage}{.3\textwidth}
    \vspace{-6.5em}%
    \begin{subfigure}[t]{.3\textwidth}
      \centering
      \begin{tikzpicture}[
          node distance=6em,
          scale=.75,
          transform shape
        ]
        \begin{scope}[<-,thick]
          \node[adot,label=90:{$\hat{J},\{h\}$}] (jhat) {};
          \node[below of=jhat,yshift=5em] () {$G{\color{OliveGreen}{hG}}G$};
          \node[adot,right of=jhat,label=90:{$J_c,\{g\}$}] (jc) {};
          \node[below of=jc,yshift=5em] () {$G{\color{OliveGreen}{G}}G$};
          \path[basic] (jhat) edge node [midway,above] {$\ndist{?}$} (jc);
        \end{scope}
      \end{tikzpicture}
    \end{subfigure}

    \vspace{2em}
    
    \begin{subfigure}[t]{.3\textwidth}
      \centering
      \begin{tikzpicture}[
          node distance=6em,
          scale=.75,
          transform shape
        ]
        \begin{scope}[<-,thick]
          \node[adot,label=90:{$\hat{J},\{g\}$}] (jhat) {};
          \node[below of=jhat,yshift=5em] () {$H{\color{OliveGreen}{hG}}H$};
          \node[adot,right of=jhat] (jhat1) {};
          \node[below of=jc,yshift=5em] () {$H{\color{OliveGreen}{hhG}}H$};
          \path[basic] (jhat) edge node [midway,above] {$\ndist{?}$} (jhat1);
          \node[adot,right of=jhat1,label=90:{$J_c$}] (jc) {};
          \node[adot,right of=jc,label=90:{$?,\{h\}$}] (jhatn) {};
          \node[below of=jhatn,yshift=5em] () {$H{\color{OliveGreen}{H}}H$};
          \path[basic,dashed] (jhat1) -- (jc) -- (jhatn);
        \end{scope}
      \end{tikzpicture}
    \end{subfigure}
  \end{minipage}
  \caption{
    Illustration of the proof of \refprop{extract:1}.
    \textit{(left)} The traversal of $\I$ to find configuration $J'$ then, if it is $x$-valent, configuration $\hat{J}$.
    \textit{(right)} If $\hat{J}$ is $x$-valent, there exists a critical configuration $J_c$, found by consecutive adjacency.
    The two cases $x=g$ \textit{(bottom right)} and $x=h$ \textit{(top right)} are detailed.
    \labfigure{critical}
  }
\end{figure}
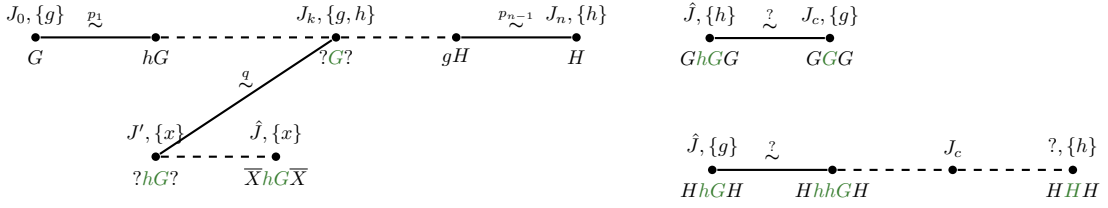

\begin{proposition}
  \labprop{extract:1}
  Eventually some index $i \in [0,n]$ is critical.
\end{proposition}

\begin{proof}
  The proof is illustrated in \reffigure{critical}.
  By \refprop{tagging:1}, $\tags(\schedNull,i)$ reaches its limit valency, say $v_i$, after time $t_i$ at process $p$.
  Hence, after time $t = \max_{i \in [0,n]} t_i$, every root of a simulation tree has reached its limit valency.
  
  Assume that $g \inter h = \{p_1, \ldots, p_{\upsilon}\}$.
  Consider the configurations $(J_i)_{i \in [0,\upsilon]}$ such that in configuration $J_i$, every process $p_j$ multicasts a message $m_h$ to $h$, if $j \leq i$, and a message $m_g$ to $g$ otherwise.
  Clearly, $(J_i)_i \subseteq \I$.
  For $i \in [0,\upsilon)$, by construction $J_i \ndist{p_i} J_{i+1}$.    
  Since all the processes in $g \inter h$ multicast a message to $g$, $J_0$ must be $g$-valent.
  Similarly, $J_{\upsilon}$ is $h$-valent.
  As a consequence, some configuration $J_k$ is either univalent critical, and we are done, or bivalent.  
  
  If for each group $x \in \{g,h\}$, some correct process multicasts a message to $x$ in $J_k$, we are done.
  Otherwise, all the processes in $P=(g \inter h){\inter}\correct(R)$ must multicast a message to say group $g$.
  (Messages multicast by $P$ are in {\color{OliveGreen}{green}} in \reffigure{critical}.)
  Name $J'$ a configuration identical to $J_k$, except that some process $q \in P$ now multicasts a message to $h$.
  If $J'$ is bivalent, we are done.
  
  Otherwise, $J'$ must be $x$-valent for some $x$.
  Consider $\hat{J}$ identical to $J'$ for all the processes in $P$ but where all the processes in $(g \inter h) \setminus P$ multicast a message to $\xbar$.
  If $\hat{J}$ is bivalent, we are done.
  Otherwise, it must be $x$-valent:
  Indeed, consider a schedule $S$ where only the correct process take steps from $J'$ and during which $p$ delivers a message. 
  By \refprop{simulation:3}, such a schedule eventually appears at process $p$.
  Configuration $S(J')$ is $x$-valent.
  As schedule $S$ is also applicable to $\hat{J}$, this configuration must be also $x$-valent. 

  We continue our traversal of $\I$ as follows:
  For each process $q'$ in $P$ that multicasts a message to $x$, let $\hat{J'}$ the configuration adjacent to $\hat{J}$ in which $q'$ multicasts a message to $\xbar$.
  If $\hat{J'}$ is bivalent we are done.
  Otherwise, we set $\hat{J}$ to $\hat{J'}$ and pursue our traversal.
  At the end of the traversal, either a bivalent configuration is found, or all the processes multicast a message to $\xbar$.
  Thus, necessarily $\hat{J'}$ is $\xbar$-valent and the configuration that precedes it in the traversal is $x$-valent.
  It follows that there exists a critical configuration $J_c$.
  The two cases, $x=g$ and $x=h$, are illustrated in \reffigure{critical}(right).
\end{proof}

\begin{proposition}
  \labprop{extract:2}  
  Fix some configuration $I_i$.
  Assume that for some configuration $I_j \in \I$ and process $q$ with $I_i \ndist{q} I_j$, the root of $\Upsilon_i$ is stable and $g$-valent, while the root of $\Upsilon_j$ is stable and $h$-valent.
  Then, process $q$ is correct and in $g \inter h$.
\end{proposition}

\begin{proof}
  By definition, the processes outside $g \inter h$ do not multicast a message in the initial configurations $\I$.
  It follows that $q$ belongs to $g \inter h$.
  For the sake of contradiction, assume that $q$ is faulty.

  Applying \refprop{simulation:3}, there exists a schedule $S$ starting from $I_i$ such that
  \begin{inparaenum}
  \item process $q$ takes no step in $S$,
  \item some message $m$ multicast in $I_i$ by a correct process (e.g., $p$) is delivered at $p$ in $S(I_i)$, and
  \item $S$ is eventually in $\Upsilon_i^{p}$, say at some time $t$.
  \end{inparaenum}
  Pick $\pi \in G^{p,t}$ such that $S$ is compatible with $\pi$ (\reflem{simulation:0}).
  Consider the first time after $t$ at which process $p$ executes \refline{omega:3}.
  Configuration $I_j$ is identical to $I_i$ for $\procSet \setminus \{q\}$, process $q$ takes no step in $\pi$ and $\BUFF$ is empty.
  Hence, by a short induction, $S$ is also applicable to $I_j$ and $S(I_j)=S(I_i)$.
  Furthermore, by \reflem{simulation:1}, $S \in \Upsilon_j$.
  Index $i$ is univalent critical, thus $\tags(\schedNull,I_i) = \{g\}$.
  Now, since $\schedNull \prefix S$, $S \in \Upsilon_j$ and $S(I_{j})=S(I_i)$, $\tags(\schedNull,I_{j}) = \{g\}$;
  contradiction.
\end{proof}

Let us now turn our attention to the case where the root of $\Upsilon_i$ is bivalent.
Similarly to \cite{omega}, we show the existence of a decision gadget in the simulation tree.
This decision gadget entails a correct process.
Furthermore, if index $i$ is critical, this process must be in $g \inter h$.

\tikzstyle{edge} = [draw, -latex',Black,->]
\tikzstyle{dedge} = [draw, -latex',Black,dashed,-]
\tikzstyle{vertex} = [inner sep=.1pt, circle, fill]

\begin{figure}[t]
  \centering
  \captionsetup{justification=centering}
  \hspace{-4em}%
  \begin{subfigure}[t]{.3\textwidth}
    \centering  
    \begin{tikzpicture}
      \begin{scope}[<-,thick]
        \node[adot,label=90:{$S$}] (S1) at (0,0) {};
        \node[adot,label=270:{$S',\{x\}$}] (S2) at (-1,-1) {};
        \node[adot,label=270:{$S'',\{\xbar\}$}] (S3) at (1,-1) {};      
        \path[edge] (S1) edge node [midway,left] {$(q,m,d)$} (S2);
        \path[edge] (S1) edge node [midway,right] {$(q,m,d')$} (S3);
      \end{scope}
    \end{tikzpicture}
    \caption{
      \labfigure{fork}
      A fork.
    }
  \end{subfigure}
  \begin{subfigure}[t]{.3\textwidth}
    \centering  
    \begin{tikzpicture}
      \begin{scope}[thick]
        \node[adot,label=90:{$S$}] (S3) at (2,-2) {};
        \node[adot,label=270:{$S',\{x\}$}] (S4) at (1,-3) {};
        \node[adot] (S5) at (3,-3) {};
        \node[adot,label=270:{$S'',\{\xbar\}$}] (S6) at (2,-4) {}; 
        \path[edge] (S3) edge node [midway,left] {$(q,m,d')$} (S4);
        \path[edge] (S3) edge node [midway,right] {$(q',\any,\any)$} (S5);
        \path[edge] (S5) edge node [midway,right] {$(q,m,d')$} (S6);
      \end{scope}
    \end{tikzpicture}
    \caption{
      \labfigure{hook}
      A hook.
    }
  \end{subfigure}
    \begin{subfigure}[t]{.3\textwidth}
    \centering  
    \begin{tikzpicture}
      \begin{scope}[thick]
        \node[adot,label=45:{$S$}] (S1) at (0,0) {};
        \node[adot,label=270:{$S_{x},\{x\}$}] (S2) at (-1.5,0) {};
        \node[adot] (S2b) at (-1,-1) {};
        \node[adot] (S3) at (2,-2) {};
        \node[adot] (S4) at (1,-3) {};
        \node[adot] (S5) at (3,-3) {};
        \node[adot,label=270:{$S_{\xbar},\{\xbar\}$}] (S6) at (2,-4) {}; 
        \path[edge] (S1) edge node [midway,left,above] {$(q,m,d)$} (S2);
        \path[edge] (S1) edge node [rotate=45,midway,below] {$(q,m,d')$} (S2b);
        \path[dedge] (S1) edge node [] {} (S3);
        \path[edge] (S3) edge node [midway,left] {$(q,m,d')$} (S4);
        \path[edge] (S3) edge node [midway,right] {$(q',\any,\any)$} (S5);
        \path[edge] (S5) edge node [midway,right] {$(q,m,d')$} (S6);
      \end{scope}
    \end{tikzpicture}
    \caption{
      \labfigure{subtree}
      A deciding subtree.
    }
  \end{subfigure}
    \caption{      
      \labfigure{gadget}%
      In $\Upsilon_i$, a decision gadget is a tuple $(S,S',S'')$, with $S'$ $x$-valent and $S''$ $\xbar$-valent.
      It can be a fork, or a hook.
      In a fork \textit{(a)}, process $q$ is correct.
      In a hook \textit{(b)}, this is the case of $q'$.
      A deciding subtree \textit{(c)} is of the form $(S,S_{x},S_{\xbar})$, with $S_{x}$ $x$-valent and $S_{\xbar}$ $\xbar$-valent.
      Such a subtree must contain a decision gadget.
  }
\end{figure}
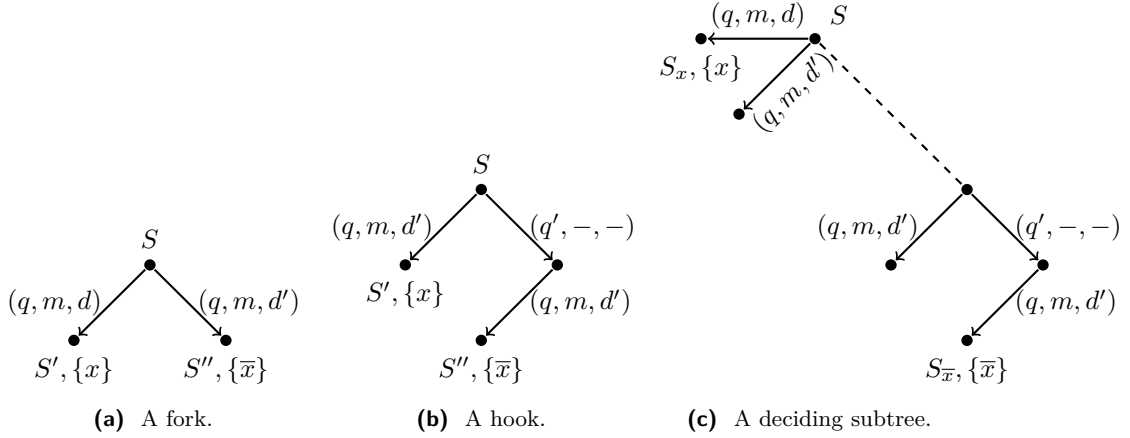

Consider a schedule $S$ in $\Upsilon_i$ having two successors $S'$ and $S''$.
Let $(S,S',S'')$ be the subtree formed by the two paths $[S,S']$ and $[S,S'']$.
This subtree is a \emph{decision gadget} when it is either a fork or a hook (see \reffigure{gadget}).
A \emph{fork} contains three vertices with $S' = S \concat (q,m,d)$ and $S'' = S \concat (q,m,d')$.
A \emph{hook} contains additionally a successor of $S$ with $S' = S \concat (q,m,d')$ and $S'' = \hat{S} \concat (q',\any,\any) \concat (q,m,d')$.
In both cases, $S$ is bivalent, while $S'$ and $S''$ are respectively $x$-valent and $\xbar$ valent, for some valency $x$.
A decision gadget induces a \emph{deciding process}.
This process is defined as $q$ for a fork, and the process $q'$ for a hook.
Intuitively, this process takes a step that fixes the valency of the run.

With more details, a subtree $(S,S_x,S_{\xbar})$ of $\Upsilon_i$ is \emph{deciding} when for some successor $S'$ of $S$, some message $m$ addressed to a process $q$, and some values $d$ and $d'$ of the failure detector $D$,
\begin{inparaenum}
\item $S$ is bivalent,
\item $S_{x} = S \concat (q,m,d)$ is $x$-valent,
\item $S_{\xbar} = S' \concat (q,m,d')$ is $\xbar$-valent, and
\item for every schedule $\hat{S}$ in $[S,S')$, $\hat{S} \concat (q,m,d')$ is univalent.
\end{inparaenum}

As illustrated in \reffigure{gadget}, when $(S,S_x,S_{\xbar})$ is deciding it must contain at least one decision gadget.
Indeed, if $S \concat (q,m,d')$ is $\xbar$-valent, then $S'=S$ and $(S, S \concat (q,m,d'), S \concat (q,m,d'))$ forms a fork.
Otherwise, there exist two neighboring vertices $S''$ and $S'$ in $[S,S_{\xbar}]$ such that $S'' \concat (q,m,d')$ is $x$-valent and $S' \concat (q,m,d')$ is $\xbar$-valent, leading to the fact that $(S',S'' \concat (p,m,d'),S' \concat (p,m,d'))$ is a hook.

During the traversal of $\Upsilon_i$, when the $\extract$ procedure encounters a deciding subtree, function $\deciding$ is executed.
In detail, following \cite{omega}, since a decision gadget is a finite graph, it can be encoded as a natural.
This allows to totally order them.
Calling $\deciding(S,S_x,S_{\xbar})$ returns the deciding process of the first decision gadget for such an order in the subtree $(S,S_x,S_{\xbar})$.

\refprops{extract:3}{extract:6} establishes two fundamental properties related to decision gadgets.
First, that the deciding process of a decision gadget must be correct and a member of $g \inter h$.
Second, that a decision gadget eventually shows up in every bivalent tree.

\begin{proposition}
  \labprop{extract:3}
  Consider a stable deciding subtree in $\Upsilon_i$ and let $\hat{q}$ be its deciding process.
  Process $\hat{q}$ is correct.
  Moreover, if $i$ is critical, $\hat{q}$ belongs to $g \inter h$.
\end{proposition}

\begin{proof}
  Assume a point in time $t$ where $\Upsilon_i$ contains a stable deciding subtree.
  Let $(S,S',S'')$ be the first deciding gadget in the subtree.
  \begin{itemize}
  \item[($\hat{q} \in \correct(R)$)]
    We proceed by contradiction, assuming that process $\hat{q}$ is faulty.
    Name $\hat{m}$ a message multicast by $p$ in $I_i$.
    Applying \refprop{simulation:3}, there exists a schedule $E$ such that
    \begin{inparaenum}
    \item $S' \concat E$ is eventually in $\Upsilon_i^p$,
    \item $\hat{q}$ takes no steps in $E$, and
    \item $p$ delivers $\hat{m}$ in $S' \concat E(I_i)$,
    \item for every schedule $\hat{S} \in \Upsilon_i^{p,t}$, if $E$ is applicable to $\hat{S}(I_i)$ then eventually $\hat{S} \concat E$ is in $\Upsilon_i^{p}$.
    \end{inparaenum}
    
    Consider a point in time $t'>t$ where $S' \concat E$ is in $\Upsilon_i$ (by \textit{(i)}).
    Schedule $S' \concat E$ is tagged at time $t'$ (by \textit{(iii)}).
    Because $S'$ is $x$-valent and the deciding gadget stable at time $t$, $S' \concat E$ must be also $x$-valent.
    Process $\hat{q}$ takes no steps in $E$ (by \textit{(ii)}).
    Every process $q'$ outside of $\hat{q}$ is in the same state in both $S'$ and $S''$, $S'.\BUFF[q'] \subseteq S''.\BUFF[q']$ (see \reffigure{gadget}).
    Hence, $E$ is also applicable to $S''(I_i)$ and $S' \concat E(I_i) = S'' \concat E(I_i)$.
    At some point in time $t''>t$, $S'' \concat E$ is in $\Upsilon_i^p$ with a tag $x$ (by \textit{(iv)})
    This contradicts that the deciding subtree is stable and $S''$ $\xbar$-valent at time $t$.
  \item[($i~\text{critical} \implies \forall x \in \{g,h\} \sep \hat{q} \in x$)]
    The proof is similar to the previous case.
    Instead of message $\hat{m}$, we consider this time the message $m_x$ addressed to $x$ and multicast by a correct process.
    Because $i$ is bivalent critical such a message $m_x$ exists.
  \end{itemize}  
\end{proof}

Procedure $\extract$ calls $\locate(Q,i)$, for each subset $Q=\{q_1,\ldots\}$ of $\procSet$ (\refline{extract:5}).
Starting from $\schedNull$, function $\locate$ considers the processes of $Q$ in a round-robin manner (\reflines{locate:2}{locate:3}).
For each process $q$, $\locate$ tries to extend the current schedule $S$ with a step $(q,m,d)$, where $m$ is the last message addressed to $q$.
If this step leads to a bivalent schedule, the search continues (\refline{locate:8}).
Otherwise, $\locate$ determines if it may lead to a decision gadget (\reflines{locate:9}{locate:16}), and returns its deciding process.
If none of the above cases applies, e.g., the successors of $S$ have no valency yet, the traversal aborts and returns $\bot$ (\refline{locate:17}).
Hence, the search for a leader in $\Upsilon_i$ eventually terminates.

Function $\locate$ traverses the tree $\Upsilon_i$ following a depth-first search of the bivalent and univalent vertices.
The search occurs in an order $<$ over the vertices of $\Upsilon_i$ (\reflinestwo{locate:4}{locate:10}), picking the smaller schedules first.
To emulate $\Omega_{g \inter h}$, this order should satisfy that the depth-first method eventually stabilizes.
This requires that if $S'$ is chosen at some point in time to extend $S$, the set of schedules smaller than $S'$ is bounded.
This property is easily obtained from an ordering of the paths the schedules are compatible with (by \reflem{simulation:0}).
In detail, for two samples $s=(p,d,k)$ and $s'=(p',d',k')$, let $s < s'$ be defined as $(k,p) < (k',p')$ with the usual semantic.
For two paths $\pi$ and $\pi'$, $\pi < \pi'$ holds when for some $j$ and every $k<j$, $\pi[j] < \pi'[j] \land \pi[k] = \pi[k']$ is true.
%
Two schedules $S$ and $S'$ in $\Upsilon_i$ are ordered according to the smallest paths in $G$ they are compatible with.

Consider some time $t_0$ after which the root of $\Upsilon_i$ is bivalent.
Let $(S_k)_{k \in \naturalSet^0}$ be the sequence of values taken after time $t_0$ by variable $S$ at $p$ right before the call to $\locate(Q,i)$ returns (\refline{extract:6}).

\begin{proposition}
  \labprop{extract:4}
  It is true that
  \begin{inparaenum}
  \item each $S_k$ is bivalent, and
  \item $(S_k)_{k}$ converges.
  \end{inparaenum}
\end{proposition}

\begin{proof}
  We prove successively each part of the proposition.
  \begin{itemize}
  \item (Each $S_k$ is bivalent.)
    Initially, $S=\schedNull$ (\refline{locate:1}).
    Since $\locate(Q,i)$ is called after time $t_0$, $\schedNull$ is bivalent.
    Then, observe that $S$ is set either at \refline{locate:7} or \refline{locate:13}.
    In both cases, after the assignment, $S$ is still bivalent.
  \item ($(S_k)_{k}$ is convergent.)
    \begin{itemize}
    \item ($Q \neq \correct(R)$.)
      For some faulty process $q$, there is eventually no step $(q,m,d)$ to extend variable $S$ at \refline{locate:4}.
      Hence $S_k$ eventually always takes the same value.
    \item (Otherwise.)
      For the sake of contradiction, consider that $(S_k)_{k}$ does not converge.
      From $(S_k)_{k}$, we build a sub-sequence $(\hat{S}_k)_{k \in \naturalSet^0}$ as follows:
      \begin{construction}{convergent}
        Schedule $\hat{S}_{0}$ is set to $\schedNull$.
        Assuming $\hat{S}_k$, $\hat{S}_{k+1}$ is built as follows.
        Variable $S$ is changed either at \refline{locate:1}, \refline{locate:7}, or \refline{locate:13}.
        As $\hat{S}_{k}$ prefixes every $S_{r \geq r_k}$, eventually for every new call to $\locate$, variable $S$ is set to $\hat{S}_{k}$.
        Let $q$ be the process considered next at \refline{locate:3} after that assignment.
        Since $(S_k)_k$ does not converge, eventually for some step $(q,m,d)$, $\hat{S}_{k} \concat (q,m,d)$ is in $\Upsilon_i$.
        By definition of $<$, the loop at \refline{locate:4} eventually stabilizes for some failure detector sample $d$.
        $\hat{S}_{k+1}$ is set to the schedule $\hat{S}_{k} \concat (q,m,d)$.
      \end{construction}

      The series $(\hat{S})_{k}$ satisfies that
      \begin{inparaenum}
      \item $\hat{S}_0=\schedNull$, 
      \item for all $k \geq 0$, $\hat{S}_{k+1} = \hat{S}_k \concat (q,m,d)$ with $q=q_{\modOf{k}{\cardinalOf{Q}}}$,
      \item each $\hat{S}_k$ strictly prefixes $\hat{S}_{k+1}.$, and
      \item for all $k \geq 0$, there exists a rank $r_k$ such that $\hat{S}_k$ prefixes every $S_{r \geq r_k}$.        
      \end{inparaenum}
      Properties \textit{(i)}-\textit{(iii)} follow from \refconstruction{convergent}.
      Property \textit{(iv)} is obtained by induction.
      In detail, this is immediate if $\hat{S}_k = \schedNull$.
      Next, assume \textit{(iii)} holds at rank $k$.
      Step $(q,m,d)$ is such that $\hat{S}_k \concat (q,m,d)$ is eventually always the smallest schedule at \refline{locate:4}.
      Hence, for some rank $r_{k+1} \geq r_{k}$, for every $r \geq r_{k+1}$, $\hat{S}_{k+1} \prefix S_{r}$.

      Let $(\pi_k)_k$ be paths in $G$ with which the schedules in $(\hat{S}_k)_k$ are compatible.
      It is easy to see that $(\pi_k)_k$ is a $Q$-fair replicated sampling sequence.
      Moreover, according to the pseudo-code of $\locate$, $\hat{S}_k=\sched(\pi_k,I_i)$.
      Applying \refprop{simulation:2}, $(\runOf{S_k,R,i})_{k}$ converges toward a $Q$-fair run $\hat{\R}$ of $\A$.
      Since $Q=\correct(R)$, $p \in Q$ and $p$ multicasts a message in this run, $p$ eventually delivers it in $\hat{\R}$.
      Hence, $S_{k \geq \kappa}$ is univalent from some rank $\kappa$;
      a contradiction.
    \end{itemize}
  \end{itemize}
\end{proof}

\begin{proposition}
  \labprop{extract:5}
  Eventually \locate(Q,i) always returns $\bot$, or for some stable and deciding subtree $(S,S_x,S_{\xbar})$ of $\Upsilon_i$, it always returns $\deciding(S,S_x,S_{\xbar})$.
\end{proposition}

\begin{proof}
  Applying \refprop{extract:4}, $(S_t)_{t}$ converges toward some schedule $\overline{S}$.
  Consider this happens at some time $t_1 \geq t_0$
  According to the pseudo-code of $\locate$, after assigning $\overline{S}$ to variable $S$, the function returns.
  If $\locate(Q,i)$ eventually always returns at \refline{locate:17}, we are done.
  Otherwise, it returns at \refline{locate:16} infinitely often the value of $\deciding(S,S_x,S_{\xbar})$
  Let $t_2 \geq t_1$ be a time after which this holds.

  Consider that $\locate(Q,i)$ returns at \refline{locate:16}.
  The variables in function $\locate$ are such that 
  \begin{inparaenum}
  \item $S'$ is a successor of $S$,
  \item $S_x = S \concat (q,m,d)$ is $x$-valent,
  \item $S_{\xbar} = S' \concat (q,m,d')$ is $\xbar$-valent, and
  \item every schedule $\hat{S}$ along the path from $S$ to $S'$ satisfies that $\hat{S} \concat (q,m,d')$ is univalent.
  \end{inparaenum}
  Hence $(S,S_x,S_{\xbar})$ is a stable and deciding subtree of $\Upsilon_i$.
  
  From the definition of $<$, there is a finite number of successors $S''$ of $S$ such that for some $d''$, $S'' \concat (q,m,d'') < S' \concat (q,m,d')$.
  Hence, after some time the loop at \refline{locate:10} stabilizes.
  Let this happens at time $t_3 \geq t_1$.
  
  Based on what precedes, there exists a stable and deciding subtree $(S,S_x,S_{\xbar})$ in $\Upsilon_i$ such that $\locate(Q,i)$ always returns $\deciding(S,S_x,S_{\xbar})$ from time $t \geq \max(t_2,t_3)$.
\end{proof}

\begin{proposition}
  \labprop{extract:6}
  For every $i \in [0,n]$, $\locate(\correct(R),i)$ cannot return $\bot$ forever.
\end{proposition}

\begin{proof}
  Assume $\locate(\correct(R),i)$ returns $\bot$ at \refline{locate:17}.
  Let $q \in Q$ be the next process to schedule in $S$ and $m$ the oldest message (if none, then $\msgNull$) addressed to $q$ in $S(I_i)$ (\refline{locate:3}).
  As $\locate(Q,i)$ returns $\bot$, either
  \begin{itemize}
  \item[(1)] there is no sample $d$, such that $S \concat (q,m,d)$ is in $\Upsilon_i$ and tagged (\refline{locate:4}), or
  \item[(2)] $S \concat (q,m,d)$ is $x$-valent, and for every successor $S'$ of $S$, either
    \begin{itemize}
    \item[(2.a)] no schedule $S' \concat (q,m,d')$ is $\xbar$-valent (\refline{locate:10}), or
    \item[(2.b)] some schedule $S' \concat (q,m,d')$ is $\xbar$-valent, but for some $\hat{S} \in [S,S')$, $\hat{S} \concat (q,m,d')$ has no tag yet (\refline{locate:15}).
    \end{itemize}
  \end{itemize}

  Applying \refprop{sampling:3}, process $q$ eventually takes some step $(q,m,d)$ from $S$.
  Moreover by \refprop{tagging:1}, $S \concat (q,m,d)$ is eventually tagged.
  Hence, case \textit{(1)} cannot happen infinitely often.
  
  Consider that $S \concat (q,m,d)$ is $x$-valent.
  Since $S$ is bivalent, then some successor $S'$ of $S$ is tagged with $\xbar$.
  Hence, case \textit{(2.a)} cannot occurs infinitely often.

  Then, applying again \refprop{tagging:1}, for every $\hat{S} \in [S,S')$, $\hat{S} \concat (q,m,d')$ is eventually tagged.
  From what precedes, \textit{(2.b)} cannot happen infinitely often.
  
  By \refprop{extract:5}, there exists a limit for $(S_k)_k$.
  It follows that \refline{locate:16} occurs infinitely often.
\end{proof}

A procedure is stable when from some point in time, it always returns at the same line with the same values for its variables.
The result below establishes that this happens eventually to $\extract$.

\begin{proposition}
  \labprop{extract:7}
  Procedure \extract is eventually stable.
\end{proposition}

\begin{proof}
  By~\refcor{tagging:2}, the valency of the root of $\Upsilon_i$ is eventually stable.
  Applying~\refprop{extract:5}, for every $Q \subseteq \procSet$, $\locate(Q,i)$ is also eventually stable.
  Thus, procedure $\extract$ eventually stabilizes.
\end{proof}

\begin{proposition}
  \labprop{extract:8}
  Eventually \extract always returns the same correct process in $g \inter h$ picked at \refline{extract:3} or \refline{extract:8}.
\end{proposition}

\begin{proof}
  Let $t_0$ be the first time at which $\extract$ is stable (by \refprop{extract:7}).
  Applying \refprop{extract:1}, there exists eventually a critical index $i \in [0,n]$ at process $p$.
  We note $t_1$ the time at which this happens first.
  Then, consider some time $t>\max(t_0,t_1)$.
  As $i$ is critical,
  \begin{itemize}
  \item (Case $i$ is univalent)
    By definition, index $i$ must pass the test at \refline{extract:2}.
  \item (Case $i$ is bivalent)
    Applying \refprop{extract:6}, $\locate(\correct(R),i)$ eventually returns some process $q$ at \refline{extract:6}.
    By \refprop{extract:5}, there exists a stable and deciding subtree $(S,S_x,S_{\xbar})$ of $\Upsilon_i$, such that $q=\deciding(S,S_x,S_{\xbar})$.
    By \refprop{extract:3}, process $q$ belongs to $g \inter h$, thus it passes the test at \refline{extract:7}.
  \end{itemize}
  From which it follows that $p$ eventually never executes \refline{extract:9}.
  Assume now that $\extract$ forever returns at \refline{extract:3}.
  By \refprop{extract:2}, $q$ is correct and in $g \inter h$.
  Otherwise, it returns forever at \refline{extract:8}.
  Then, by \refprop{extract:3}, process $q$ is correct and since the test at \refline{extract:7} was passed, $q$ belongs to $g \inter h$.
\end{proof}

\subsection{Correctness of \refalg{omega}}
\labappendix{omega:correctness}

Based on the prior results, we may now establish that:

\begin{theorem}
  \labtheo{correctness}
  \refalg{omega} implements $\Omega_{g \inter h}$.
\end{theorem}

\begin{proof}
  Name $R$ some run of \refalg{omega} and let $P=\correct(R) \inter (g \inter h)$.

  For starters, we show that the range of the failure detector implemented with \refalg{omega} is correct.
  If $p$ is outside $g \inter h$, a call to $\query$ returns $\bot$ at \refline{query:1}.
  Otherwise, there are three cases to consider in the $\extract$ procedure.
  As detailed below, the procedure always returns a process identifier that belongs to $g \inter h$.
  \begin{itemize}
  \item (\refline{extract:3}) the call returns some process in $q$; by \refprop{extract:2}, $q \in P$.
  \item (\refline{extract:8}) function $\locate$ is called and the result satisfies the test at \refline{extract:7}.
  \item (\refline{extract:9}) $p$ simply retrieves its identity.
  \end{itemize}
  
  By \refprop{extract:8}, every correct process in $p \in P$ eventually always elects some correct process $l_p \in g \inter h$.
  Applying~\refcortwo{simulation:0a}{tagging:2} and \refprop{extract:7}, for any two $p,q \in P$, $l_p=l_q$.
\end{proof}

\fi

\end{document}